\documentclass[12pt,a4paper]{article}
\usepackage[utf8]{inputenc}%
\usepackage[left=1.5in,right=1.5in,top=1.5in,bottom=1.5in]{geometry}
\usepackage{amsfonts}

\usepackage{amsthm,amssymb,amsmath}

\theoremstyle{definition}
\newtheorem{definition}{Definition}[section]
\theoremstyle{plain}
\newtheorem{assp}{Assumption}
\DeclareMathOperator*{\argmax}{argmax}

\usepackage{mathtools}

\newtheorem{prop}{Proposition}
\newtheorem{corollary}{Corollary}
\newtheorem{theorem}{Theorem}
\newtheorem{claim}{Claim}

\theoremstyle{definition}
\newtheorem{example}{Example}
\theoremstyle{plain}

\newtheorem{lemma}{Lemma}
\usepackage[font=small,labelfont=bf]{caption} 
\usepackage{subfig}
\usepackage{soul}
\usepackage{xcolor}
\usepackage{bbm}
\usepackage{bm}
\usepackage{chngcntr}
\usepackage [english]{babel}
\usepackage [autostyle, english = american]{csquotes}
\MakeOuterQuote{"}
\usepackage[round]{natbib}
\usepackage{hyperref}
\hypersetup{
    colorlinks=true,
    linkcolor=blue,
    filecolor=blue,      
    urlcolor=blue,
    citecolor=blue,
}
\usepackage{graphicx}

\usepackage{enumitem}
\setlist{nolistsep}

\title{Competing Persuaders in Zero-Sum Games\thanks{We are grateful to Navin Kartik and Elliot Lipnowski for their continued advice and encouragement. We thank (in no particular order) Alessandro Pavan, Jacopo Perego, Mark Dean, Alessandra Casella, Yeon-Koo Che, Laura Doval, Iain Bamford, Leonard Goff and various seminar participants for their helpful comments.}}
\author{Dilip Ravindran\thanks{Humboldt University Berlin. Email: dilip.ravindran@hu-berlin.de} and Zhihan Cui\thanks{University of California Los Angeles, Anderson School of Management. Email: zhihan.cui@anderson.ucla.edu}}

\date{\today}

\begin{document}

\maketitle

\begin{abstract}
We study Bayesian Persuasion with multiple senders who have access to conditionally independent experiments (and possibly others). Senders have zero-sum preferences over information revealed. We characterize when any set of states can be pooled in equilibrium and when all equilibria are fully revealing. The state is fully revealed in every equilibrium if and only if sender utility functions are `globally nonlinear'. With two states, this is equivalent to some sender having nontrivial preferences. The upshot is that `most' zero-sum sender preferences result in full revelation. We explore what conditions are important for competition to result in such stark information revelation.

\end{abstract}

\newpage

\section{Introduction}
A key question in the economics of persuasion is the effect of competition on information provision. This question is of interest in many contexts in which agents with opposing interests control the information available to one or multiple decision makers. For instance, a jury or judge's decision to acquit or convict a defendant is informed by the evidence collected by defense and prosecution attorneys. Competing firms design advertisements in order to convince a consumer to buy their products. In order to persuade voters, politicians may hire experts to provide information which might validate their platforms and not their opponents'. In all these settings, the information providers (senders) are competing to influence the decision maker(s). How does this competition affect what information they reveal? 

While it has been shown that information disclosure increases with competition in some settings (\citet{battaglini2002multiple}; \citet{milgrom1986relying}; \citet{shin1998adversarial}), in others competition has the opposite effect (\citet{emons2019strategic}; \citet{kartik2017investment}). In this paper we address the question by modelling two or more senders persuading one or more receiver(s) about an unknown state. The senders influence the receiver's beliefs by disclosing information in the manner of Bayesiyan Persuasion. Unlike existing work, our senders simultaneously choose conditionally independent experiments; the receiver observes these experiments and their realizations and updates her belief. To fix ideas, consider competing lobbyists commissioning reports to persuade a politician (or entire legislature) to vote yes/no on a climate change bill.\footnote{e.g. a bill which would mandate use of alternative energy sources.} Here the state may be whether or not climate change is a significant threat. The politician would only like the bill to pass if it is while environmental lobbyists support the bill and corporate lobbyists oppose it regardless of the state. The lobbyists can commission reports from climate change experts of their choice; these reports reveal information about the state and influence the politician's belief and hence decision.

To capture the disagreement between the senders in the examples above, we consider a model in which senders are maximally-competitive \textemdash the senders' payoffs are zero-sum functions of the receiver's posterior belief. This assumption is natural, for instance, in the lobbying example. Lobbyists may only care about the probability the bill passes/fails; these probabilities sum of one (zero is just a normalization \textemdash any constant-sum payoffs will do) and may depend on the politician's poterior belief. Our question is: how does competition affect how much information is revealed in equilibrium and how does this change with the number of senders?

There is always an equilibrium of this game in which all senders fully reveal the state. Our main result is that typically the state is fully revealed in \emph{every} equilibrium. We find that \textemdash under mild technical assumptions \textemdash when sender utility functions are globally nonlinear (in particular, are nonlinear on every edge of the simplex) then regardless of the number of senders the state is fully revealed in every equilibrium.\footnote{The \emph{every} quanitifier implies, by standard upper hemicontinuity arguments, that if the conditions for full revelation are met for some zero-sum utilities, then for utilities close to those \emph{all} equilibria are \emph{almost} fully revealing.} If utility functions are linear on every edge of the simplex, a knife-edge case, there are equilibria in which the receiver does not always learn the state. Two implications are worth mentioning. In the binary-state case, the state is fully revealed in every equilibrium if and only if all senders are not indifferent across all strategy profiles. If the receiver chooses from a finite set of actions, then generically  the receiver learns enough to take her first-best action; furthermore, the state is fully revealed in every equilibrium if and only if the receiver prefers a different action in every state.

The intuition for our results can be seen from the two-sender binary-state case. The first observation is that, as a sender is always free to fully reveal the state and the game is zero-sum, each sender must do exactly as well in any equilibrium as she would from full revelation. When utility functions are nonlinear, there are some posteriors at which one sender has an 'advantage' and one has a 'disadvantage' in the following sense. Conditional on such a posterior belief, the sender with an advantage is getting a higher payoff than she would from fully revealing the state while the sender with a disadvantage would rather the state be revealed. Loosely, each sender will try and maximize the probability that the receiver's posterior lies in her regions of advantage. While no sender has unilateral control over the receiver's posterior, a sender can affect the posterior conditional on her opponent's signal realization not being fully revealing. We show that if the state is not being fully revealed, at least one sender can use extreme signals\footnote{Signals that induce posteriors in favor of one state close to fully revealing that state.} to force some posteriors into regions she has an advantage. She can do this in a way that gives her higher utility than she would get from full revelation, which means this cannot be an equilibrium.

This idea extends to arbitrary finite state spaces and more than two senders. Given choices of experiments for each sender, we say that a set of states is \emph{not pooled} if the receiver never assigns positive probability to all of them. If a set of states is not pooled, the receiver will always be able to distinguish between (at least) some of these states. We show that a subset of states is not pooled in every equilibrium if and only if conditional on the receiver learning the state is in this subset, some sender has strict preferences over what further information to reveal. For instance, a pair of states is not pooled in all equilibria if any only if some sender's utility is nonlinear on the edge of the simplex between those two states.\footnote{i.e. the line joining degenerate beliefs on the two states.} It follows that nonlinearity on every edge is necessary and sufficient for full revelation in every equilibrium.

An implication of our main result is that competition of a zero-sum form cannot decrease, in the sense of \citet{blackwell1953equivalent}, the information provided in equilibrium.\footnote{In other words, if a set of senders did not have zero-sum preferences and we added a sender to make the game zero-sum, then equilibrium information will not decrease.} One contribution of this paper is to identify a natural and applicable environment \textemdash zero-sum preferences and conditionally independent experiments \textemdash for which this is the case. In our setting the effect of competition is stark \textemdash typically it generates full information. We show our results do not rely on restricting attention to conditional independence; in particular, if senders have access to additional experiments (i.e. senders are able to correlate their experiments partially or arbitrarily), they still hold. Our analysis hence generalizes, up to a technical assumption, the zero-sum game results of \cite{gentzkow2017bayesian}, who consider a similar model in which senders are able to arbitrarily correlate experiments (see discussion below). However, when senders have access to only a very limited set of conditionally independent experiments, our results fail and zero-sum competition can lead to a strict decrease in equilibrium information. We explore what technology/strategies are needed for zero-sum competition to not decrease information and identify simple classes of experiments for which this is the case. For instance, if each sender only has access to a single base experiment (satisfying a mild technical condition) but can repeat it as many times as desired, then a version of our results still holds. This example maps on to applications in which senders commission scientific studies and can choose only the sample size or number of clinical trials used in the studies.

Finally, we consider a few variants of our model. When senders begin the game with bounded private information\footnote{i.e. senders' private beliefs are bounded away from the boundaries of the simplex.} there is full revelation in all equilibria but for in a knife-edge case of sender preferences. This setting is applicable, for instance, to the prosecutor/defense attorney example in which the defense lawyer may hold private information on her client's guilt/innocence. We show that  there is full revelation in all equilibria for a larger set of sender preferences in our baseline model than in a version of the game where senders move sequentially; we relate this version of the game to the sequential Bayesian Persuasion games of \citet{li2018sequential}, \citet{li2018bayesian}, \citet{wu2017coordinated}, and \citet{dworczak2020robust}. In Supplementary Appendix C, we then apply our analysis to a few extensions of our model which deviate from our zero-sum assumption and are applicable to real-world settings.

The rest of the paper is organized as follows. We discuss related literature below. In Section 2 we introduce our model and in Section 3 we solve the special case of two senders and a binary state space. In Section 4 we extend the intuition from Section 3 to the general model to obtain our main results. Section 5 considers two applications of our model, one with a single receiver and one with multiple receivers. In Section 6 we discuss to what extent our results are robust to our two most substantive assumptions: zero-sum preferences and conditionally independent experiments. We conclude in Section 7. Appendix A contains proofs of all the main results and Supplementary Appendix B contains discussions and proofs of extensions and robustness results. Appendix C contains extensions to two nonzero-sum applications.

\textbf{Related literature.} The effect of competition on information has been studied in environments of cheap talk (e.g. \citet{krishna2001model}, \citet{battaglini2002multiple}), costly signalling (e.g. \citet{kartik2020information}), disclosure (e.g. \citet{milgrom1986relying}) and, more recently, Bayesian Persuasion (e.g. \citet{gentzkow2016competition}, \citet{gentzkow2017bayesian}). \citet{gentzkow2016competition} and \citet{gentzkow2017bayesian} (henceforth GK) show that when senders have access to a sufficiently rich (Blackwell connected) set of experiments, competition (weakly) increases equilibrium information. Under Blackwell connectedness, GK (2017b) obtain a full relevation result for zero-sum games very similar to ours. Blackwell connectedness requires that senders are able to arbitrarily correlate their experiments; in contrast we study the case of conditional independence \textemdash a common assumption in information economics and an important benchmark for many applications. We discuss the relationship between our paper and GK (2017b) in Section 6.

\citet{boleslavsky2018limited} and \citet{au2020competitive} study two senders persuading a receiver. However their setups are substantially different from ours because each sender can only reveal information about part of the state (her own type); as a consequence, they find unique non-fully revealing equilibria. \citet{koessler2022interactive} also consider an environment where each of many senders can reveal information only about one dimension of a multidimensional state space.   \citet{li2018sequential}, \citet{li2018bayesian}, and \citet{wu2017coordinated} consider Bayesian Persuasion with multiple senders moving sequentially. Finally, in a concurrent paper, \citet{dworczak2020robust} (henceforth DP) study a single persuader who is uncertain about what additional information nature may give the receiver and chooses an experiment to maximize her worst-case payoff. This setting is related to competition between two senders in our model (our case of more than two senders is less related). While their baseline model allows nature to arbitrarily correlate her experiment with the persuader's, they address the case of conditionally independent experiments in a supplementary appendix and obtain results related to ours. However, due to differences between the models, our results concerning the total information revealed in equilibrium are stronger. See Section 4.3 for discussion.

\section{Model}

There is a state $\omega \in \Omega = \{1,...,N\}$. All agents have a common prior belief on $\omega$ with full support $\pi \in$ int$\Delta(\Omega)$. There are $M>1$ senders, $1,...,M$, who persuade a receiver.\footnote{As we do not explicity model the receiver acting, the model allows for any number of receivers. We discuss the receiver(s) in more detail and explicitly model them in Sections 4 and 5.}

Fix a set of signals $S$ with $|S| \geq N$. The game starts with each sender $i$ simultaneously choosing a set $S_i \subseteq S$, $|S_i| <\infty$ and an experiment $\Pi_i: \Omega \rightarrow \Delta(S_i)$. Each $\Pi_i$ gives the probability of the receiver receiving each signal in $S_i$ conditional on each state. As $|S_i|< \infty$, senders may only choose \emph{finite signal} experiments. Implicit in this definition of experiments is that senders' experiments are independent conditional on the state. We discuss these assumptions below.

The receiver observes the choices of $\Pi_1,...,\Pi_M$ (and implicitly $S_1,...,S_M$). Then, the state is realized (but not observed by the receiver) and signals from each of the $M$ experiments, $s_1 \in S_1,...,s_M \in S_M$, are realized and observed by the receiver. The receiver is Bayesian and updates her belief on $\omega$ to some posterior $\beta \in \Delta(\Omega)$. Senders receive their payoffs and the game ends.

Senders' payoffs depend only on the receiver's posterior belief $\beta$. Each sender $i$ has a piecewise analytic utility function $u_i: \Delta(\Omega) \rightarrow \mathbb{R}$.\footnote{That is, each $u_i$ is defined by a finite partition of $\Delta(\Omega)$ into convex sets and a real analytic function for each element of the partition. Note this restriction is not necessary; see Section 4.3 for discussion.} Crucially, we assume senders' payoffs are zero-sum: $u_1(\beta) + ... + u_M(\beta) = 0$ for all $\beta \in \Delta(\Omega)$.\footnote{Zero is just a normalization; any constant-sum game will do.}\footnote{This could represent the reduced form of a game where the receiver chooses an action $a \in A$ after observing experiment realizations. The receiver has preferences $u_r(a,\omega)$ and the senders have preferences $\{u_i(a,\omega)\}_i$ which are zero-sum: $\sum_i u_i(a,\omega) = 0$ for all $a \in A, \omega \in \Omega$. } For any state $l \in \Omega$ let $\delta_l \in \Delta(\Omega)$ be the belief that puts probability $1$ on state $l$. Due to the structure of the game, we can make the following normalization: $u_i(\delta_l)=0$ for all senders $i=1,...,M$ and all states $l \in \Omega$ (see explanation for this normalization below after reading the definitions in the next paragraph).

A strategy profile is a choice of experiment for each sender $(\Pi_1,...,\Pi_M)$. Let $U_i(\Pi_1,...,\Pi_M) = \mathbb{E}_{\Pi_1,...,\Pi_M}[u_i(\beta)]$ be sender $i$'s ex-ante expected utility from $(\Pi_1,...,\Pi_M)$; the expectation is over experiment realizations, of which $\beta$ is a function. Senders choose experiments to maximize their ex-ante expected utility.

\textbf{Normalizing $\boldsymbol{u_i(\delta_l)=0}$.} To see why this is a normalization, suppose senders have utility functions $u_1',...,u_M'$ with $u_1'(\beta)+...+u_M'(\beta)=0$ for all $\beta$. For $i=1,..,M$ let $\alpha_i : \Delta(\Omega) \rightarrow \mathbb{R}$ be the affine function $\alpha_i(\beta) = - \sum_l \beta_l u_i'(\delta_l)$. For each $i$, let $u_i = u_i' + \alpha_i$. Then $u_i(\delta_l)=0$ for all $l \in \Omega$. Note that utility function $u_i$ preserves the same preferences over strategy profiles as $u_i'$ as for any strategy profile $(\Pi_1,...,\Pi_M)$, $\mathbb{E}_{\Pi_1,...,\Pi_M} [u_i(\beta)] = \mathbb{E}_{\Pi_1,...,\Pi_M} [u_i'(\beta)] - \sum_l \pi_l u_i'(\delta_l)$ and the latter term is a constant. Finally note that $\alpha_1(\beta)+...+\alpha_M(\beta)=0$ for all $\beta \in \Delta(\Omega)$, so $u_1+....+u_M=0$. While the normalization changes senders' preferences over the receiver's posterior, preferences over strategy profiles are unchanged and these are what is relevant for equilibrium analysis.

\textbf{Discussion of strategies.} There are a few aspects of senders' strategies that are worth discussing. First, we restrict senders to picking conditionally independent experiments. In the lobbyist example, this corresponds to the reports lobbyists commission being independently commissioned, researched, and written. The conditionally independent case is natural for many applications and an important benchmark to consider. It contrasts with environments with richer strategies in which senders are able to correlate their experiments' realizations. For instance, \citet{gentzkow2017bayesian} study a `Blackwell connected' environment in which senders can arbitrarily correlate experiments; our setting is not Blackwell connected. In Section 6 we allow our senders to play correlated experiments in addition to conditionally independent ones.

Second, while we restrict attention to finite signal experiments,\footnote{Restricting attention to finite signal experiments is commonly done in the Bayesian Persuasion literature (e.g. \citet{kamenica2011bayesian}, \citet{gentzkow2017bayesian}).} this is for convenience and we show in Supplementary Appendix B that our main result goes through when the assumption is dropped.\footnote{Note all equilibria with the finite signal restriction are equilibria without it.} As we allow for $S$ to be infinite, our baseline model allows for the possibility of senders choosing experiments with arbitrary numbers of finite signals. As is common in Bayesian Persuasion work, it is important for our arguments that $|S| \geq N$.\footnote{It is well known that a single persuader has such a Bayesian Persuasion solution using at most $N$ signals; hence in our game senders will always have such a best response. Note that restricting the cardinality of $S$ is not without loss in multi-sender settings as arbitrarily large numbers of signals may be used in equilibrium.}

\textbf{A special case.} It is worth noting a simple and important class of games covered by our model. Suppose there is a single receiver who maximizes her utility given her posterior belief $\beta$ by taking one of two actions $a_1,a_2$. Suppose senders get payoffs depending only on this action; these payoffs need not be zero-sum. This setup is natural for many applications we are interested in: lobbyists persuading a politician to vote yes/no on a bill, attorneys persuading a judge/jury to acquit/convict a defendant, or two politicians competing for a single voter's vote. As long as at least two senders prefer different actions, we can normalize sender payoffs to make the game zero-sum. To see this, suppose sender $i$ strictly prefers action $a_1$; then when comparing two strategy profiles, $i$ prefers the one that induces the receiver to take $a_1$ more often. Normalizing sender payoffs to be zero-sum will not affect this preference as long as ordinal preferences over actions are preserved.

\textbf{Interim beliefs.} Instead of thinking of sender $i$ picking $\Pi_i$, it is easier to think of $i$ choosing a distribution over the receiver's \emph{interim} beliefs. For any $i$ and choice of $\Pi_i$, let $\Gamma_i \in \Delta(\Omega)$ be the random variable representing the receiver's belief on $\omega$ if she observes only the realization of $\Pi_i$, $s_i \in S$. $\Gamma_i$ represents the \emph{interim} belief of the receiver after she observes information from $\Pi_i$ but before viewing the realizations from $\{\Pi_j\}_{j \neq i}$ and updating to her posterior belief.\footnote{As the receiver is Bayesian, the order in which she views signal realizations does not matter.} 

A random variable $\Gamma$ is Bayes-plausible if $\mathbb{E}[\Gamma] = \pi$. Following \citet{kamenica2011bayesian}, it is without loss for us to recast the choice of experiment of each sender $i$ as a selection of a Bayes-plausible distribution of the interim beliefs, $\Gamma_i$, the experiment induces. As we have restricted senders to picking finite signal experiments which employ at most $|S|$ signals, a pure strategy for sender $i$ is a selection of a Bayes-plausible $\Gamma_i$ with support on a finite number of beliefs that is at most $|S|$.\footnote{Note that any finite mixture of pure strategies is also a pure strategy.} Henceforth, when we use $\Gamma_i$ it implicit that this random variable is Bayes-plausible and has finite a support with at most $|S|$ elements. A strategy profile is a vector $(\Gamma_1,...,\Gamma_M)$. Fixing any strategy profile and sender $i$, let $\Gamma_{-i}$ denote the experiment induced by observing realizations $\{\Gamma_j\}_{j \neq i}$.

There are two benchmark experiments to consider. We say $\Gamma_i$ is fully revealing, or $\Gamma_i = \Gamma^{FR}$, if $\Pr(\Gamma_i = \delta_l) = \pi_l$  $\forall l \in \Omega$. If any sender chooses a fully revealing experiment the receiver learns the state with certainty. The second benchmark is the fully uniformative experiment which we denote $\Gamma^{U}$; $\Gamma_i = \Gamma^{U}$ if $\Pr(\Gamma_i = \pi) =1$.

\textbf{Equilibrium.} The solution concept is Nash Equilibrium; a Nash Equilibrium of this game is a vector of random variables $(\Gamma_1,...,\Gamma_M)$ such that no sender $i$ can strictly improve her ex-ante expected utility, $U_i(\Gamma_1,...,\Gamma_M)$, by deviating.

There is a trivial NE of this game: $(\Gamma^{FR},...,\Gamma^{FR})$. All senders are left indifferent across all experiment choices as the state will be fully revealed by other senders' experiments regardless. Our results characterize when the state is fully revealed in \emph{every} equilibrium.

\section{Two Senders and a Binary State Space}

First we derive the main results for the two-sender binary-state case. The intuition will extend to the general case.

\begin{center}
\includegraphics[scale=0.6]{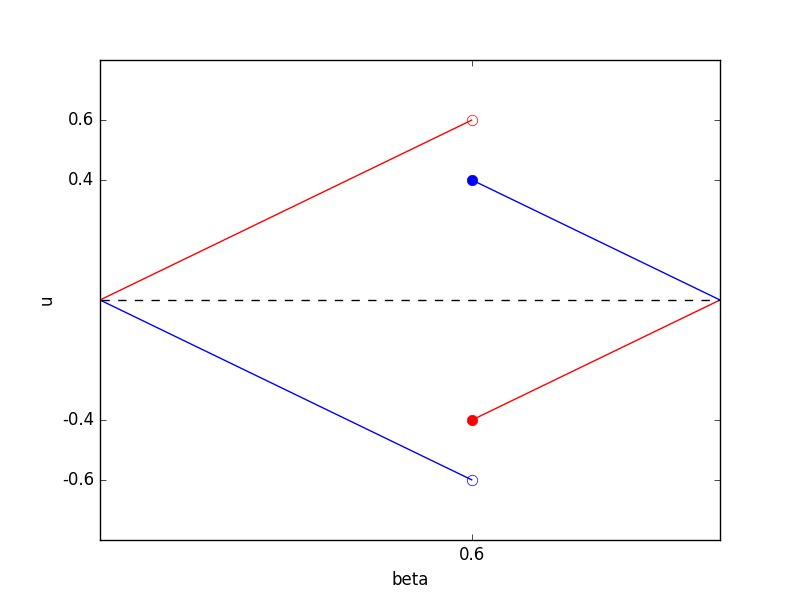}
\captionof{figure}{Example of $u_1$ (blue) and $u_2$ (red). $u_1(\beta)=\beta$ for $\beta<0.6$ and $u_1(\beta)=1-\beta$ for $\beta \geq 0.6$. Sender 1's preferences are those in Kamenica and Gentzkow (2011)'s leading prosecutor persuading a judge example with discontinuity at $0.6$ and normalization $u_1(0)=u_1(1)=0$.}
\end{center}

Let $\Omega = \{0,1\}$. A belief here is a scalar representing the probability the state is $\omega=1$. Figure 1 shows an example of sender preferences. A sender $i$'s strategy is a choice of interim belief random variable $\Gamma_i \in [0,1]$. Note that for any $\Gamma_1,\Gamma_2$ chosen, the receiver's posterior belief can be written as a function of the interim beliefs realized from both experiments. If $\Gamma_1=x$ and $\Gamma_2=y$, then the posterior is:

\begin{equation}  \label{beta}
\beta(x,y) = \frac{(1-\pi)xy}{xy-\pi x- \pi y+ \pi} \\
\end{equation}

Note that $\beta(1,y)=\beta(x,1)=1$ and $\beta(0,y)=\beta(x,0)=0$; if either interim belief fully reveals the state, the other is irrelevant. Note $\beta(0,1)$ and $\beta(1,0)$ are not well defined but this is not an issue as it is impossible for one sender to fully reveal $\omega = 0$ while the other reveals $\omega=1$.

For either sender $i$, given any strategy $\Gamma_i$, consider the distribution of $\Gamma_i$ conditional on $\Gamma_j=x$ ($j \neq i$): $\Pr(\Gamma_i=y|\Gamma_j=x)$. $\Pr(\Gamma_i=y|\Gamma_j=x)$ can be constructed by taking the signal structure $\Pi_i$ that corresponds to $\Gamma_i$ and deriving the distribution over interim beliefs it induces if $x$ and not $\pi$ was the receiver's prior.\footnote{Formally $\Pr(\Gamma_i=y|\Gamma_j=x) = \frac{\Pr(\Gamma_i=y)}{\pi (1 - \pi)} (xy - \pi x - \pi y + \pi)$}

Given an opponent choice of $\Gamma_j$, define $W_i(x)$ as sender $i$'s expected payoff conditional on generating $\Gamma_i=x$. For a fixed $\Gamma_2$ $W_1(x)$ is written:

\begin{equation} \label{W_i}
 W_1(x) = \sum_{y \in supp[\Gamma_2]} u_1(\beta(x,y)) \Pr(\Gamma_2=y|\Gamma_1=x) \\
\end{equation}

Note that $W_1(1)=W_2(0)=W_1(0)=W_2(1)=0$; if either players' experiment generates a fully revealing interim belief then the other experiment is irrelevant. Two special cases are important. When $\Gamma_2=\Gamma^{FR}$ then, regardless of $u_1$ or the prior,  $W_1(x)=0$ for all $x$. This is because $\Gamma_2$ will reveal the state to be $0$ or $1$; any interim belief sender 1 produces can only affect the relative probability of these events, both of which yield $u_1=0$.  Meanwhile, when $\Gamma_2=\Gamma^U$, then $W_1(x) = u_1(x)$ as $x$ will be the receiver's posterior.


\subsection{Analysis}

The result below will be useful.

\begin{lemma}
In any equilibrium $(\Gamma_1,\Gamma_2)$:
\begin{enumerate}
\item $U_1(\Gamma_1,\Gamma_2)=U_2(\Gamma_1,\Gamma_2)=0$.
\item $W_1(x) \leq 0$ and $W_2(x) \leq 0$ for all $x \in [0,1]$.
\end{enumerate}
\end{lemma}

Property (1) follows from the game being zero-sum and the observation that each sender $i$ can guarantee a payoff $U_i = 0$ by fully revealing the state. While property (1) says that sender equilibrium payoffs equal those from full revelation, it does not say that we must have full revelation in equilibrium: for instance if the $u_1$ and $u_2$ are linear, any $(\Gamma_1,\Gamma_2)$ constitute an equilibrium.

Property (2) holds because any violation leads to a contradiction of (1). Fix any $\Gamma_j$ such that sender $i \neq j$ has $W_i(x)>0$ for some $x$. We can find a $\Gamma_i$ with support only on $\{x,0,1\}$; as $i$ gets strictly positive expected utility whenever $x$ is realized and $0$ otherwise, $U_i(\Gamma_i,\Gamma_j)>0$. Hence such $\Gamma_j$ cannot be played in equilibrium. Figure \ref{figure_lemma1} shows how to construct such a $\Gamma_1$ in our main example when $\Gamma_2 = \Gamma^U$ and hence $W_1(0.7)>0$.

\begin{center} 
\includegraphics[scale=0.6]{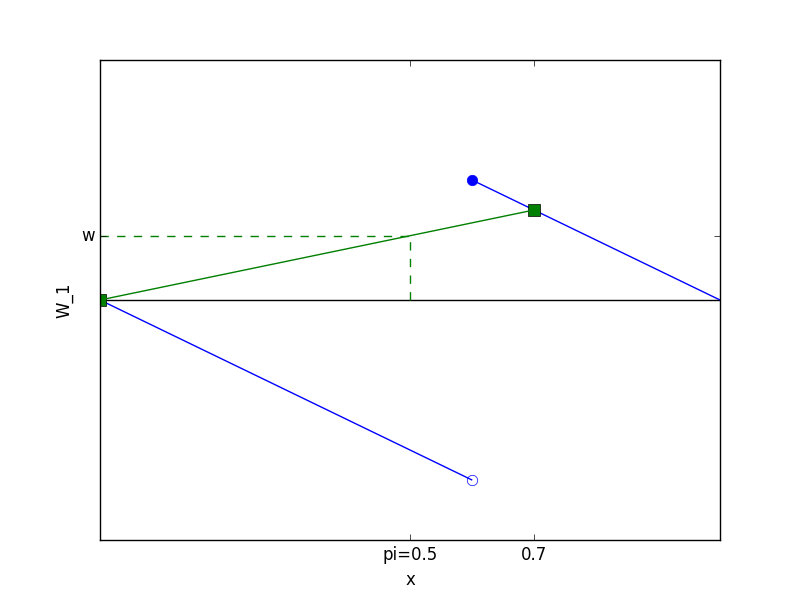}
\captionof{figure}{The blue curve is $W_1(x)$ when $\Gamma_2=\Gamma^U$. Sender 1 can construct $\Gamma_1$ with support $\{0,0.7\}$ with $Pr(\Gamma_1=0)$ and $Pr(\Gamma_1=0.7)$ chosen to respect Bayes-plausibility. $U_1(\Gamma_1,\Gamma_2)=w>0$.}
\label{figure_lemma1}
\end{center}

More generally, when $W_i(x)>0$ we can construct an appropriate $\Gamma_i$ as follows. (1) If $x<\pi$ define $\Gamma_i$ by $Pr(\Gamma_i=1)=\frac{\pi-x}{1-x}$ and $Pr(\Gamma_i=x)=\frac{1-\pi}{1-x}$, (2) if $x> \pi$ then let $Pr(\Gamma_i=0)=\frac{x-\pi}{x}$ and $Pr(\Gamma_i=x)=\frac{\pi}{x}$, and (3) if $x = \pi$ let $\Gamma_i=\Gamma^U$. Note that in each case $|supp[\Gamma_i]| \leq 2$ and hence $\Gamma_i$ can always be implemented under the assumption $|S| \geq N=2$. 

The main result for the two-sender binary-state case relies on Lemma 1. We show that if utility functions are nonlinear, then in equilibrium at least one sender $i$ must choose $\Gamma_i = \Gamma^{FR}$, or else $W_j(x)$ will violate property (2).

\begin{prop}
The state is fully revealed in every equilibrium if and only if $u_1$ is nonlinear.
\end{prop}

The `only if' direction is trivial \textemdash if $u_1$ is linear then senders are indifferent between all strategy profiles. Hence the result can be restated as:
\begin{align*}
& \text{There is full relevation in every equilibrium} \\
& \iff \exists (\Gamma_1,\Gamma_2),(\Gamma_1',\Gamma_2') \text{ and a sender $i$ with } U_i(\Gamma_1,\Gamma_2) \neq U_i(\Gamma_1',\Gamma_2').
\end{align*}

In the rest of this section we prove the `if' direction: $u_1$ nonlinear $\implies$ full revelation in every equilibrium.

The idea can be seen using the example in Figure 1 with any prior. Note that for all $r \in [0.6,1)$, $u_1(\beta)>0$ for all $\beta \in [r,1)$; fix any such $r$. Suppose for contradiction that sender 2 plays a non-fully revealing strategy $\Gamma_2$ in some equilibrium. As $\Gamma_2 \neq \Gamma^{FR}$, $\Pr(0<\Gamma_2<1)>0$; let $\underbar{y} = \min supp[\Gamma_2]\setminus \{0,1\} \in (0,1)$ be in the smallest interior belief in the support of $\Gamma_2$. Using the definition of $\beta(x,y)$, define $\underbar{x}$ by $\beta(\underbar{x},\underbar{y}) = r$. Conditional on $\Gamma_1=x \in [\underbar{x},1)$, $\beta(x,y) \in [r,1)$ for all interior $y$ in $\Gamma_2$'s support. But then for all $x \in [\underbar{x},1)$ we have:

\begin{align*}
& W_1(x) = \underbrace{u_1(\beta(x,0))}_{=0} \Pr(\Gamma_2=0|\Gamma_1=x) + \underbrace{u_1(\beta(x,1))}_{=0} \Pr(\Gamma_2=1|\Gamma_1=x) + \\
& \sum_{y \in supp[\Gamma_2] \setminus \{0,1\}} \underbrace{u_1(\beta(x,y))}_{>0} \underbrace{\Pr(\Gamma_2=y|\Gamma_1=x)}_{>0} > 0.\\
\end{align*}

This contradicts Lemma 1 property (2). 

The broader intuition is as follows. We say a sender $i$ has an \emph{advantage} on any subset of $[0,1]$ on which $u_i$ is strictly positive; for instance in the example, sender 1 has an advantage on $[0.6,1)$.\footnote{We use the word advantage because Lemma 1 tells us that both senders will get ex-ante expected utility $0$ in equilibrium. Any posteriors that yield strictly better utility than this for a sender are relatively advantageous to that sender.} While senders would like the receiver's posterior to fall in their regions of advantage with high probability, neither sender's experiment unilaterally controls the posterior. However, in the example sender 1 has an advantage at the end of the unit interval, $[r,1)$. If sender 2 chooses $\Gamma_2$ that is interior with positive probability then sender 1 can find extreme enough interim beliefs $x \geq \underbar{x}$ guaranteeing that, conditional on $x$ being realized and $\Gamma_2$ being interior, the receiver's posterior is in $[r,1)$. Whenever $\Gamma_2$ fully reveals the state both senders get utility $0$, and so overall sender 1 gets a strictly positive expected payoff from generating an interim belief $x \in [\underbar{x},1)$; Lemma 1 says this is not possible in equilibrium.

This argument does not depend on the specific $u_1$ in the example. As $\beta \rightarrow 1$, whenever $u_1$ approaches $u_1(1)=0$ from above, as is the case in the example, we can find an $r \in (0,1)$ such that sender 1 has an advantage on $[r,1)$. Hence for any such $u_1$, we can replicate the same argument to show that any $\Gamma_2 \neq \Gamma^{FR}$ cannot be played in equilibrium. If $u_1$ approaches $0$ from below as $\beta \rightarrow 1$ then $u_2$ must approach from above, and so we must have $\Gamma_1 = \Gamma^{FR}$ in any equilibrium. The same argument applies whenever $u_1$ or $u_2$ approach $0$ from above as $\beta \rightarrow 0$. 

Piecewise analycity guarantees that when utilities are nonlinear in a neighborhood of $0$ or $1$, then some $i$ sender has an advantage on an interval $(0,r]$ or $[r,1)$. These intervals of advantage at the extremities of the interval are important as senders face uncertainty over the realization of their opponent's experiment. Conditional independence of experiments means $i$ cannot precisely control the receiver's posterior given every interior realization of $\Gamma_j$ and hence $i$ must have an extreme interval of advantage that she can ensure all interior posteriors fall in.

As $u_1,u_2$ are piecewise analytic, there is only one other case to consider: $u_1,u_2$ nonlinear and $u_1(\beta)=u_2(\beta)=0$ for all $\beta$ in some neighborhoods of both $0$ and $1$. Here too there will be a sender with an advantage closest to the ends of $[0,1]$ who can find a violation of Lemma 1 property (2) whenever her opponent does not fully reveal the state. Let $r = \sup \{\beta \in [0,1]: u_1(\beta) \neq 0\}$ be the supremum of posteriors at which $u_1,u_2$ are nonlinear (note $r<1$). WLOG, assume either $u_1(r)>0$ or that there are beliefs approaching $r$ from below at which $u_1>0$ (zero-sumness implies this must hold for either $u_1$ or $u_2$). Suppose $\Gamma_2 \neq \Gamma^{FR}$. Defining $\underbar{y}$ and $\underbar{x}$ as before, if $u_1(r)>0$, then $W_1(\underbar{x}) = u_1(r) \Pr(\Gamma_2=\underbar{y}|\Gamma_1=\underbar{x}) >0$. If $u_1(r) \not >0$, then $W_1(\underbar{x} - \epsilon)>0$ for some small enough $\epsilon>0$. 

\section{Main Result}
Now we apply the logic from the previous section to $N \geq 2$ states and $M \geq 2$ senders. For any $T \geq 1$ and experiments $\Gamma_1,...,\Gamma_T$, let $\beta(\Gamma_1,...,\Gamma_T)$ be the receiver's posterior belief after observing all $T$ realizations.\footnote{See Appendix A for an explicit formula.} Fixing opponents strategies $\Gamma_{-i}$, we can define $W_i(x)$ just as before: $W_i(x) = \sum_{y \in supp[\Gamma_{-i}]} u_i(\beta(x,y)) \Pr(\Gamma_{-i}=y|\Gamma_i=x)$. Lemma 1 then extends to this more general setting:

\addtocounter{lemma}{-1}
\begin{lemma}
In any equilibrium $(\Gamma_1,...,\Gamma_M)$, for all senders $i$:
\begin{enumerate}
\item $U_i(\Gamma_1,...,\Gamma_M)=0$.
\item $W_i(x) \leq 0$ for all $x \in \Delta(\Omega)$.
\end{enumerate}
\end{lemma}

Property (1) follows from the same arguments as in the previous seciton. Property (2) then follows by showing that for any $x \in \Delta(\Omega)$ there exists a Bayes-plausible experiment with support on $x$ and at most $N-1$ elements of $\{\delta_1,...,\delta_N\}$.\footnote{This is because $\pi$ is in the convex hull of $x$ and some $N-1$ elements of $\{\delta_1,...,\delta_N\}$.}

For any strategy profile $(\Gamma_1,...,\Gamma_M)$ and subset of states $\Omega' \subseteq \Omega$, we say $\Omega'$ is \emph{pooled} if $\Pr(\beta_l(\Gamma_1,...,\Gamma_M)>0$ $\forall l \in\Omega')>0$ (otherwise, $\Omega'$ is \emph{not pooled}). When $\Omega'$ is not pooled, the receiver will always be able to rule out at least one of the states in the set. For any $\Omega' \subseteq \Omega$, let $\Delta(\Omega') = \{\gamma \in \Delta(\Omega): \sum_{l \in \Omega'} \gamma_l = 1\}$ be the subset of the simplex assigning probability $1$ to $\omega \in \Omega'$. Note that for two states $l,k$, $\Delta(\{l,k\})$ is the edge of the simplex between $\delta_l$ and $\delta_k$. Hence Proposition 1 can be restated as: states $0,1$ cannot be pooled in any equilibrium if and only if some $u_i$ is nonlinear on $\Delta(\{0,1\})$. Proposition 2 generalizes Proposition 1 and characterizes when any set of states can be pooled in equilibrium.

\begin{prop}
A set $\Omega' \subseteq \Omega$ is pooled in no equilibrium if and only if there is a sender $i$ for whom $u_i$ is nonlinear on $\Delta(\Omega')$.
\end{prop} 


Suppose the receiver learns that $\omega \in \Omega'$. Conditional on this event, if $u_i$ is linear on $\Delta(\Omega')$ for all $i$ then all senders are indifferent across all additional information that can be revealed. Meanwhile, if for some $i$ $u_i$ is nonlinear on $\Delta(\Omega')$, then there is some additional experiment that $i$ either strictly prefers or disprefers to not providing any additional information. Proposition 2 says that conditional on the receiver learning that $\omega \in \Omega'$, some sender having strict preferences over revealing additional information characterizes $\Omega'$ being not pooled in every equilibrium. 

Taking this logic one step further, the following corollary gives another implication of Proposition 2.

\begin{corollary} \label{cor_conflict} 
Conditional on any posterior belief induced in equilibrium no sender can find an experiment that strictly improves or reduces her payoff.
\end{corollary}


Corollary \ref{cor_conflict} tells us that in every equilibrium senders reveal enough information to remove all further conflict between them. An implication of this result is that no information revelation, the strategy profile $(\Gamma^U,...,\Gamma^U)$, is an equilibrium if and only if all senders are indifferent across all strategy profiles (i.e. senders have trivial preferences).

We now turn to the proof of Proposition 2. Proving the `only if' direction is straightforward: if all $u_i$ are linear on $\Delta(\Omega')$ then all senders fully revealing the state whenever $\omega \in \Omega \setminus \Omega'$ and revealing no further information whenever $\omega \in \Omega'$ is an equilibrium that pools $\Omega'$.\footnote{Formally, each sender plays $\Gamma$ s.t. $\Pr(\Gamma=\delta_n)=\pi_n$ for all $n \not \in \Omega'$ and $Pr(\Gamma=y \text{ s.t. } y_l=\frac{\pi_l}{\sum_{k \in \Omega'} y_k} \text{ } \forall l \in \Omega') = \sum_{k \in \Omega'} \pi_k$.} Conditional on $\omega \not \in \Omega'$, no sender $i$ has the incentive to deviate as the receiver will learn the state from $\Gamma_{-i}$. Conditional on $\omega \in \Omega'$, $\Gamma_{-i}$ will reveal this fact to the receiver, ensuring that $\beta \in \Delta(\Omega')$; as $u_i$ is linear on $\Delta(\Omega')$ there is no additional information $i$ can reveal profitably.

Now for the `if' direction; we provide the main intuition here but leave to proof to Appendix A. For any $\Omega' \subseteq \Omega$ and sender $i$, fixing an opponent strategy $\Gamma_{-i}$ consider $W_i(x)$ on $\Delta(\Omega')$. As noted above, $\Gamma_i =x \in \Delta(\Omega')$ $\implies$ $\beta(x,\Gamma_{-i}) \in \Delta(\Omega')$ w.p. $1$. Generating an interim belief in $\Delta(\Omega')$ tells the receiver that $\omega \in \Omega'$, ensuring the posterior is also on this set. Conditional on $x \in \Delta(\Omega')$, the only information $\Gamma_{-i}$ can convey is relative probabilities of states in $\Omega'$. When evaluating $W_i(x)$ on $\Delta(\Omega')$, sender $i$ can treat $\Gamma_{-i}$ as an experiment just about states in $\Omega'$. 

If some $u_i$ is nonlinear on $\Delta(\Omega')$, we can apply a similar argument to Proposition 1. Firstly, at least one sender has an advantage somewhere on $\Delta(\Omega')$. We can find some sender $j$ such that whenever $\Pr(\Gamma_{-j}=y$ s.t. $y_l >0$ $\forall l \in \Omega')>0$ (i.e. $\Gamma_{-j}$ pools $\Omega'$), $W_j(\underbar{x})>0$ for some $\underbar{x} \in \Delta(\Omega')$. Like in Proposition 1, $j$ will be a sender with an advantage closest to the extremes (boundaries) of $\Delta(\Omega')$ and $\underbar{x}$ will be extreme enough to ensure that whenever $\Gamma_{-j}$ assigns positive probability to all states in $\Omega'$, $\beta(\underbar{x},\Gamma_{-j})$ falls in a region of $j$'s advantage with positive probability. Otherwise, $\beta(\underbar{x},\Gamma_{-j})$ will fall where $j$ gets $0$ utility. Hence $W_j(\underbar{x})>0$, violating Lemma 1 and implying that $\Gamma_{-j}$ (and hence the whole strategy profile) must not pool $\Omega'$ in equilibrium.

The bulk of the proof involves finding this $\underbar{x}$ as a function of a $\Gamma_{-j}$ that pools $\Omega'$. When $|\Omega'|=2$, conditional on $\Gamma_i=x \in \Delta(\Omega')$ we are in a binary-state world; hence we can find $\underbar{x}$ on edge $\Delta(\Omega')$ just as in Proposition 1. As our main result, Theorem 1, will only rely on Proposition 2 for $|\Omega'|=2$, we relegate a full proof of the case $|\Omega'|>2$ to Appendix A (noting that the broad intuition is the same).

Whenever no pair of states can be pooled in any equilibrium, the state is fully revealed all equilibria. Applying Proposition 2 to every pair of states:

\begin{theorem}
The state is fully revealed in every equilibrium if and only if for every pair of states $\{l,k\}$ there is a sender $i$ for whom $u_i$ is nonlinear on $\Delta(\{l, k\})$.
\end{theorem}

This is an immediate corollary of Proposition 2. Theorem 1 shows that preferences being \emph{globally nonlinear}, or nonlinear on every edge of the simplex, characterizes all equilibria being fully revealing. The state may not be fully revealed only if all senders have linear preferences on an edge of the simplex. Linearity along any edge for any sender, let alone all senders, is knife-edge and so for typical sender preferences the state is fully revealed in all equilibria. 

\textbf{Remark.} While we have assumed that $\sum_{i=1}^M u_i =0$, Theorem 1's condition for full revelation in all equilibria only requires a subset of senders to have zero-sum utilities. Specifically, suppose the game with senders $1,...,M$ is not zero-sum\footnote{i.e. there exists $\beta,\beta' \in \Delta(\Omega')$ with $\sum_{i=1}^M u_i(\beta) \neq \sum_{i=1}^M u_i(\beta')$} but that there exists $I \subseteq \{1,...,M\}$ with $\sum_{i \in I} u_i(\beta) =0$ for all $\beta \in \Delta(\Omega)$. Then the state is fully revealed in every equilibrium if for every $l,k \in \Omega$ there is a sender $i \in I$ with $u_i$ nonlinear on $\Delta(\{l,k\})$.\footnote{We can extend Proposition 2 similarly. Some sender in $I$ having nonlinear preferences on $\Delta(\Omega')$ is sufficient for $\Omega'$ to be pooled in no equilibrium.} This is easy to see: our analysis applies directly to senders $I$. Note that we no longer have the necessary condition for full revelation in all equilibria: it is possible that all senders in $I$ have linear utilities on some $\Delta(\{l,k\})$ but  all equilibria of the game are fully revealing.

\subsection{Single receiver with finite actions}

Thus far, the receiver's only role in the model has been to update her beliefs and hence we have defined sender preferences $\{u_i\}_i$ over the receiver's posterior belief. In this section we explicitly model the receiver. In particular, we consider a game in which after updating to her posterior a single receiver picks an action from a finite set $A$. When we can microfound the game in this way, Theorem 1 takes a clean form.

Suppose after observing all experiment realizations the receiver picks an action $a \in A$ and receives a payoff $u_r(a,\omega)$ while each sender $i$ gets payoff $u_i(a,\omega)$. We make the generic assumption that no agent is indifferent between any actions at any state. The solution concept is Perfect Bayesian Equilibrium (PBE) and we assume the receiver breaks ties by choosing `higher' actions when indifferent (as actions can be reordered arbitrarily, we make this assumption just to fix some tie-breaking rule).\footnote{It turns out that fixing a tie-breaking rule is not necessary for the results below to go through; only the assumption on lack of indifference between actions at every state is crucial. However, the tie-breaking rule ensures that in all equilibria of the this game, sender preferences over the receiver's posterior will be piecewise analytic and hence we will be able to immediately apply our prior results.}


In the space of posteriors, senders have piecewise linear utility functions. For any $\Omega' \subseteq \Omega$ these functions are linear on $\Delta(\Omega')$ if and only if the receiver has the same best action at every state in $\Omega'$. By Proposition 2, a set of states $\Omega'$ cannot be pooled in any equilibrium if and only if the receiver has different best actions at least two states in $\Omega'$. This implies a version of Theorem 1: 

\begin{corollary} \label{cor_finite_action}
Suppose a single receiver choosing actions from a finite set breaks indifferences in favor of higher actions. Generically, the state is fully revealed in every equilibrium if and only if the receiver has a different best action at every state.
\end{corollary}

Further, when any subset of states $\Omega' \subseteq \Omega$ is pooled in equilibrium, Proposition 2 implies the receiver has the same best action at all states in $\Omega'$. Hence more information would not change the receiver's action which means \emph{the receiver learns enough to take her first best action}. We leave the details behind these two results to Supplementary Appendix B. 

In Section 5.1 we work through an example of a game with a single receiver with a finite set of actions. In Section 5.2, we present an application with multiple receivers and discuss some of the differences between single and multiple receiver games.

\subsection{Competition and information}

Consider our game played by only a subset of senders $\{1,...,M'\} \subset \{1,...,M\}$ ($M'<M$); note the game with $M'$ senders may no longer be zero-sum. In order to assess the impact of zero-sum competition on information provision, we compare the information revealed in equilibria with senders $\{1,...,M'\}$ and with senders $\{1,...,M\}$ (if $M'=1$, then the `equilibria' of the $M'$ sender game are just sender $1$'s Bayesian Persuasion solutions). We say a strategy profile is equivalent to another if they induce the same distribution over posteriors. A strategy profile is (strictly) more informative than another if it is (strictly) Blackwell more informative \citep{blackwell1953equivalent}. A strategy profile is no more informative than another if it is not strictly more informative.

Clearly, zero-sum competition (weakly) increases information provision whenever the $M$ senders have globally nonlinear preferences. It is easy to find cases in which zero-sum competition strictly increases information provision; for instance, in our leading example, all equilibria with any one sender (i.e. all Bayesian Persuasion solutions) are strictly less informative than the unique equilibrium with two senders (full revelation). 

When there is some linearity in the $M$ senders' preferences, and hence there are non-fully revealing $M$ sender equilibria, we can show that any set of states that can be pooled in an $M$ sender equilibrium can also be pooled with $M'<M$ senders. A consequence of this is the following result comparing the informativeness of equilibria in both games.

\begin{prop} \label{prop_comp}
Let $(\Gamma_1,...,\Gamma_M)$ be an equilibrium of the $M$ sender game. Then: \begin{enumerate}
\item There is an $M'$ sender equilibrium that is no more informative than $(\Gamma_1,...,\Gamma_M)$.
\item If an $M'$ sender equilibrium $(\Gamma'_1,...,\Gamma'_{M'})$ is more informative than $(\Gamma_1,...,\Gamma_M)$, then there is an $M$ sender equilibrium equivalent to $(\Gamma'_1,...,\Gamma'_{M'})$.
\end{enumerate}
\end{prop}


When combined, point (1) and point (2) provide a sense in which zero-sum competition does not decrease information provision. Information which can be revealed in an $M'$ sender equilibrium must have at least one of the two following properties. Either it can also be revealed in an $M$ sender equilibrium, or it is no more informative than \emph{all} $M$ sender equilibria. Moving from $M'$ to $M$ senders can only remove equilibrium outcomes that are no more informative than all outcomes of the new game. Equilibria that are removed will also be strictly worse than some $M$ sender equilibria (e.g. the fully revealing one).

The weakness of Proposition \ref{prop_comp} is that, when there are non-fully revealing $M$ sender equilibria, these need not be Blackwell comparable to some $M'$ sender equilibria; this means zero-sum competition need not \emph{increase} information provision according to many orders used in the literature (e.g. those defined in \cite{gentzkow2016competition}, \cite{che2019weak}, \cite{milgrom1994monotone}). The exception is in the binary-state case, where it is easy to show $M$ sender equilibria are more informative than $M'$ sender equilibria according to the `strong set order' (\citet{veinott1989lattice}, \citet{milgrom1994monotone}). 

In the case of a single receiver with finite actions, as all $M$ sender equilibria deliver the receiver her first-best payoff, moving from the $M'$ sender game to the $M$ sender game must make the receiver weakly better off.\footnote{As the receiver learns everthing valuable to her in all $M$ sender equilibria, comparing the informativeness of $M$ and $M'$ sender equilibria is less important.} This has normative implications for how a decision maker should choose sources of information or experts. Suppose that before making a decision, the receiver can chooses a set of experts who then choose what information (experiments) to reveal. The experts may have their own vested interests in the receiver's action. If the receiver has a selected some set of $M'$ experts and can choose an additional expert, she should pick one who maximally disagrees with the others.

One contribution of this paper is to identify a natural and applicable environment \textemdash zero-sum preferences and conditionally independent experiments \textemdash for which competition cannot decrease information provision. In this setting, typically, competition increases information starkly. In Section 6 we discuss to what extent our assumptions on preferences and sender technologies/strategies can be relaxed while still maintaining this result. 


\subsection{Discussion}

Here we discuss to what extent our results and analysis are robust to changes in our modelling assumptions.

\textbf{Robustness to zero-sumness.} Using standard upper hemicontinuity arguments, we can show Theorem 1 is robust. 

\begin{prop} \label{prop_robustness}
Suppose senders' utility functions converge to zero-sum and utilities are globally nonlinear in the limit. Whenever convergent, the information revealed along any sequence of equilibria converges to full revelation.
\end{prop}

Convergence for utilities is in the $\sup$ norm. For information, the notion is convergence in distribution of the receiver's posterior. We leave the details to Supplementary Appendix B.  

Proposition \ref{prop_robustness} tells us that our results are not knife-edge. Typically, when preferences are close to zero-sum, all equilibria are close to fully revealing.\footnote{In Supplementary Appendix B we show a similar robustness result for Proposition 2.} This robustness is one reason we focus on conditions for full revelation in \emph{all} equilibria. Note that if the limiting preferences are linear on every edge of the simplex, it is still possible for the information revealed in all equilibria to converge to full revelation \textemdash Proposition \ref{prop_robustness} is just a sufficient condition.

While our results are informative about preferences close to zero-sum, our analysis does not apply far from it. In Section 6 we discuss nonzero-sum preferences more generally.


\textbf{Piecewise analytic utility.} Our assumption that utility functions are piecewise analytic is not necessary for our results. Proposition 1 relied on being able to find an interval of advantage for one sender at the extremes of the unit interval. For this result, we just need to rule out pathological utility functions that, under our normalization, take values oscillating infinitely about $0$ close to the ends of the unit interval. A sufficient condition for this would be the piecewise analycity of utilities in some neighborhoods of each degenerate belief. Theorem 1, which only relies on Proposition 2 applied to pairs of states, also goes through under this weaker condition. For Proposition 2, we need utility functions to not oscillate infinitely about $0$ on any path of beliefs in the simplex.

\textbf{Experiments without finite signals.} We have focused on finite signal equilibria because the results are cleaner. The same intuition applies when senders can choose any conditionally independent experiments (with countably or uncountably many signals). However in this case we only obtain a sufficient condition for full revelation in every equilibrium  \textemdash satisfied in all but a knife-edge case \textemdash but not a necessary one.

For any states $l,k$ let $v^{l,k} \in \mathbb{R}^{N}$ be the vector from $\delta_l$ to $\delta_k$.\footnote{$v^{l,k}_k=1$, $v^{l,k}_l = - 1$, $v^{l,k}_n = 0$ for all $n \neq l,k$.} For any sender $i$ let $\nabla_{v^{l,k}} u_i(\cdot)$ be the directional derivative of $u_i$ moving along $v^{l,k}$. Note that moving from any $\delta_l$ to $\delta_k$ along $v^{l,k}$, we cross a finite number of analytic `pieces' of each $u_i$ (by piecewise analycity). Hence while $\nabla_{v^{l,k}} u_i(\beta)$ need not exist for all $\beta$, the following quantity will always exist for all $i$ and $l,k$:
\begin{align*}
& d_i^{l,k} = \lim\limits_{\substack{\beta \in \Delta(\{l,k\})\setminus \{\delta_k\} \\ \beta \rightarrow \delta_k} } \nabla_{v^{l,k}} u_i(\beta)
\end{align*}

Note that $d_i^{l,k}$ exists even when $\nabla^{l,k} u_i(\delta_k)$ does not exist.

We say some $u_i$ satisfies \textbf{Condition 1} on edge $\Delta(\{l,k\})$ if either $d_i^{l,k} \neq u_i(\delta_k)-u_i(\delta_l)$, $d_i^{k,l} \neq u_i(\delta_l)-u_i(\delta_k)$, or both.

Not satisfying Condition 1 is knife-edge. It is not satisfied by $u_i$ on $\Delta(\{l,k\})$, for instance, if $u_i$ is linear on this edge or if $u_i$ is linear in neighborhoods of $\delta_l$ and $\delta_k$. Condition 1 is satisfied if $u_i$ looks like $u_1$ or $u_2$ in our leading example (Figure 1) on the edge. Most importantly, Condition 1 is generically satisfied for all $u_i$ on edge $\Delta(\{l,k\})$ in a model with a single receiver with a finite set of actions if and only if the receiver prefers different actions at $l$ and $k$. If for each edge of the simplex some sender's preferences satisfy Condition 1, we have full revelation in all equilibria.

\begin{prop} \label{prop_inf}
Suppose senders can choose any (conditionally independent) experiments. The state is fully revealed in every equilibrium if for every pair of states $\{l,k\}$ there exists a sender $i$ for whom $u_i$ satisfies Condition 1 on $\Delta(\{l,k\})$.
\end{prop}

Proposition \ref{prop_inf} says that we should still expect full revelation in all equilibria for typical sender preferences when we drop the finite signal restriction on experiments. We leave the proof to Supplementary Appendix B, but for intuition consider the two-sender binary-state case. Note that when proving Proposition 1 using our leading example, we made use of the fact that if $\Gamma_2$ is not fully revealing, it has in its support some minimum interior interim belief $\underbar{y}$. When $\Gamma_2$ does not have finite support, $\underbar{y}$ need not exist and whether sender 1 can find a point $\underbar{x}$ for which $W_1(\underbar{x})>0$ will depend on the limiting behavior of $u_1$ (hence Condition 1).


\textbf{Sender private information.} In many settings of persuasion agents hold private information. We can easily incorporate receiver private information into our model: at each public posterior $\beta$ realized from sender experiments, we can compute sender expected payoffs by taking an expectation over what private information the receiver may obtain.\footnote{See discussion in \citet{kamenica2011bayesian}.} Here we deal with the more interesting case: sender private information. This is an important consideration in many of our applications. At the start of a court case, a defense attorney may have information about the guilt/innocence of her client that the prosecution and judge/jury are not privy to. Alternatively, companies may possess private information about the quality of products they advertise.

Suppose when the game begins each sender receives a private signal. We assume these signals are bounded\footnote{They induce beliefs bounded away from the simplex's boundaries.} and realized from finite signal conditionally independent experiments.\footnote{The latter two assumptions are for convenience.} The solution concept is PBE. In equilibrium, senders could potentially signal their private information through their choice of experiment. However, for typical sender preferences the takeaway from Theorem 1 remains the same.

\begin{prop} \label{prop_private}
Suppose senders receive private signals before the game. The state is fully revealed in every equilibrium if for every pair of states $\{l,k\}$ there exists a sender $i$ for whom $u_i$ satisfies Condition 1 on $\Delta(\{l,k\})$.
\end{prop}

The logic behind the result remains close to that in the baseline model. Any signalling of private information that does occur via choice of experiments will, as private information is bounded, induce bounded beliefs for the receiver. Hence a sender can `overpower' information provided through signalling by generating extreme enough interim beliefs, just as she can overpower interior interim beliefs induced by her opponents' experiments. Condition 1 is needed because of inability to discipline off-path beliefs in equilibrium; we leave details and the proof to Supplementary Appendix B.

\textbf{Sequential moving senders.} Consider a sequential version of our model, in which senders $1,...,M$ move in order. Senders observe all previous experiment choices (but not realizations).\footnote{If senders could observe the realizations of upstream senders' experiments, then a downstream senders could correlate her experiment with upstream senders' by conditioning her experiment choice on upstream signal realizations. Our results in this section would also hold for such a model.} Such a setup may be applicable in modelling firms competing for consumers by designing advertisements. A firm may observe the advertising campaign (or experiment) its competitor chooses, but because it cannot observe the sales/marketing data associated with the advertising campaign,\footnote{e.g. how many consumers clicked on an online ad or how many bought products after viewing an ad.} cannot observe the realization of the experiment.

We are interested in subgame perfect Nash Equilibria (SPNE) of this game. Note that for each simultaneous game there are multiple corresponding sequential games, one for each ordering of senders. The following result helps clarify the relationship between the our baseline model and a sequential version.

\begin{prop} \label{prop_seq}
If for $u_1,...,u_M$ there is full revelation in every SPNE of the sequential game with the senders moving in some order, then there is full revelation in every equilibrium of the simultaneous game.
\end{prop}

We prove the result Supplementary Appendix B and also show that the converse does not hold: it is possible for there be to full revelation in every equilibrium of the simultaneous game but non-fully revealing SPNE in the sequential game for every ordering of senders. The result shows that we are guaranteed full revelation in equilibrium for a (weakly) larger set of sender preferences with simultaneity than with sequentiality. This is in line with Norman and Li (2018a) and Wu (2017), which show that simultaneous persuasion cannot generate less information than sequential.\footnote{However both papers allow senders to correlate experiments arbitrarily. Wu (2017) also considers zero-sum games, but only shows existence of a fully revealing equilibrium.}

In their supplementary appendix, \citet{dworczak2020robust} consider a model similar to the sequential model described above with two senders (see Section 1 for a description of their model). Their paper however asks very different questions than ours and obtains weaker full revelation results. Differences in timing (sequentiality vs simultaneity) are one reason for our stronger full revelation results. Another is that while our results concern the total information revealed by multiple senders, their results emphasize the information revealed by a single persuader (less so her opponent \textemdash nature). Formally, in the two sender version of our model, our condition for any states $\Omega'$ to be not pooled in every equilibrium is equivalent to the persuader in \citep{dworczak2020robust}'s model having a unique optimal strategy of not pooling $\Omega'$ \emph{or} nature minimizing her payoff by not pooling $\Omega'$. 


\section{Applications}

In this section we consider two applications, one with a single receiver and one with multiple receivers. 

\subsection{Lobbying}

We first work through an example of persuading a single receiver. This example is intended to unpack the sender preferences over the receiver's posterior, which are the primitives of our baseline model. 

Consider the problem of two competing lobbyists attempting to persuade a single politician. Suppose $\omega$ represents type of threat climate change poses; the threat can either be low ($\omega=L$), medium ($\omega=M$), or high ($\omega=H$). There are three possible actions the politician can take to combat climate change. If she chooses $a=N$ no action is taken, choosing $a=W$ would pass a weak law (maybe with non-binding measures), and choosing $a=S$ would enact a strong law (maybe mandating use of alternative energy).\footnote{For instance, if the politician is the president of the United States $W$ and $S$ could be potential executive orders. Alternatively, the politician could be a legislator casting a pivotal vote on two bills.}


At the start of the game, an oil lobby firm and an environmental lobbying group commission experts of their choice to write reports on climate change (i.e. the two senders choose experiments on $\omega$).\footnote{As noted by \citet{kamenica2011bayesian}, lobbying groups spend large amounts of money on such reports. For example, the tobacco lobby heavily funds reports on the health effects of smoking cigarettes \citep{barnoya2006tobacco}.} The politician then reads the reports, updates her belief on $\omega$, and chooses her action $a$.

All three agrents receive payoffs that depend on the politician's action, $a$. The oil lobbyist prefers the weakest possible action: $u_o(N,\omega)=1$, $u_o(W,\omega)=0.5$, $u_o(S,\omega)=0$ for all $\omega$.\footnote{Both laws could reduce demand for oil: $S$ directly and $W$ by affecting consumers' choices.} The environmental lobbyist, meanwhile gets payoffs: $u_e(N,\omega)=0$, $u_e(W,\omega)=0.5$, $u_e(S,\omega)=1$ for all $\omega$. Note that payoffs for lobbyists are constant-sum at every action and state pair. The politician's preferences over actions depend on the state. Suppose the politician strictly prefers to take no action when she is certain $\omega = L$ and strictly prefers $a=S$ at state $\omega=H$. When $\omega=M$, she has preferences: $u_r(N,M) = 0$, $u_r(W,M) = 0.5$ and $u_r(S,M) = \rho$, where $\rho \in [0,1]$ may measure her willingness to fight the oil lobby. Assume that when indifferent between actions, the politician breaks indifferences favoring $N$ over $W$ over $S$.

The solution concept is PBE. Let $R_N, R_W, R_S \subseteq \Delta(\Omega)$ be the sets of posteriors at which the politician optimally chooses actions $N$, $W$, and $S$ respectively. These sets are disjoint, partition $\Delta(\Omega)$, and are convex (given our tie-breaking rule). Figure \ref{figure_lobbyist} shows an example of these sets. In any PBE, the lobbyists' expected payoff given the politician's posterior, $u_o(\beta)$ and $u_{e}(\beta)$, can be written: $u_i(\beta) = \sum_{\omega \in \Omega} u_i(a^*(\beta),\omega) \beta_{\omega}$, where $a^*(\beta)$ is the politician's action choice at $\beta$. Each $u_i(\beta)$ is linear on each of $R_N,R_W,R_S$ and hence piecewise linear on $\Delta(\Omega)$. As $u_c$ and $u_e$ are also constant-sum, they meet the conditions for our analysis to apply (for exposition, we do not normalize $u_i(\delta_l)=0$ for all $l$).                                                                                                                                                                                                                                                                                                                                                                                                                                                                                                                                                                                                                                                                                                                                                                                                                                                                                                                                                                                                                                                                                                                                                                                                                                                                                                                                                                                                                                                                                                                                                                                                                                                                                                                                                                                                                                                                                                                                                                                                                                                                                                                                                                                                                                                                                                                                                                                                                                                                                                                                                                                                                                                                                                                                                                                                                                                                                                                                                                                                                                                 

First suppose $\rho < 0.5$. Then: $\delta_L \in R_N$, $\delta_M \in R_W$, and $\delta_H \in R_S$ \textemdash at each degenerate belief the politician takes a different action. Then $u_o(\beta)$ must have a discontinuity along $\Delta(\{L,H\})$ as $u_o(\beta)$ jumps from $1$ to $0$ when crossing from $R_N$ to $R_S$. Similarly, there are discontinuities in $u_o(\beta)$ along $\Delta(\{L,M\})$ and along $\Delta(\{M,H\})$. As $u_o(\beta),u_e(\beta)$ are nonlinear on every edge, the state is fully revealed in all equilibria (Theorem 1).


Now suppose $\rho > 0.5$. Then $\delta_L \in R_N$ and $\delta_M, \delta_H \in R_S$. Along $\Delta(\{L,H\})$, $u_o(\beta)$ and $u_e(\beta)$ are the same as when $\rho < 0.5$; hence Proposition 2 implies $L$ and $H$ cannot be pooled in equilibrium. Along $\Delta(\{L,M\})$, both $u_i(\beta)$ are discontinuous and so $\{L,M\}$ also cannot be pooled.  However, as $\delta_M,\delta_H \in R_S$ and $R_S$ is convex, $\Delta(\{M,H\}) \subseteq R_S$; hence at every belief along this edge the politician approves the strong bill and the  environmental lobbyist gets a payoff of $1$ while the oil lobbyist gets payoff $0$. Sender preferences are linear along $\Delta(\{L,H\})$ and these states \emph{are} pooled in some equilibria.

\begin{center} 
\includegraphics[scale=0.7]{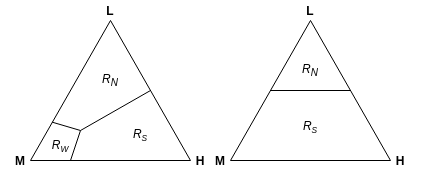}
\captionof{figure}{An example of how $R_N,R_W,R_S$ may look when $\rho=0$ (left) and $\rho=1$ (right).}
\label{figure_lobbyist}
\end{center}

When $\rho < 0.5$, the politician prefers a different action in each state. Hence any information provided about the relative probabilities of the three states could be valuable to her in decision making. As the politician learns the state in every equilibrium, she is always able to take her first-best action. Meanwhile, when $\rho>0.5$, the politician prefers the strong action in both states $M$ and $H$ and so conditional on learning $\omega \in \{M,H\}$, no further information could benefit her decision making. Although $M$ and $H$ can be pooled in equilibrium, neither can be pooled with $L$ and so the politician still learns enough to take her first best action. Because lobbyists disagree \textemdash in a zero-sum sense \textemdash over the politician's action, they never withhold information that could influence this action. However, they may withhold information which will not affect the action.

\subsection{Persuading voters}

As an example of a game with multiple receivers, we adapt \cite{alonso2016persuading}'s model of a single politician persuading voters to a competitive setting with two politicians.

There are $V$ voters $\{v_1,...,v_V\}$ ($V$ odd). Each voter $v_i$ will cast a vote $a_i$ for either politician 1 or 2. The politician who receives the majority of the votes, $O \in \{1,2\}$, wins the election and receives a payoff of $1$ while the losing politician receives payoff $0$. We call $O$ the outcome of the election. Voters have preferences over politicians that depend on an underlying state $\omega \in \Omega$: $\{u_{v_i}(O,\omega)\}_{i=1}^V$. 
%

We interpret this setup as follows. Politician 1 and 2 have platforms specifying policies they would enact if elected. Voters have preferences over the policy decisions that would be taken by the politicians, and hence the politicians themselves, that depend on $\omega$. At the start of the game, the politicians simulateneously choose experiments on $\omega$; we think of this as politicians commissioning experts to write reports. The voters commonly observe these experiments and their realizations and then simultenously cast their votes.

In any PBE, voters' actions constitute a Nash Equilibrium at every posterior belief induced. We select equilibria in which voters play undominated stategies and vote for politician 1 when indifferent. This implies that voter $v_i$ votes for politician 1 at posterior $\beta$ if and only if $\sum_{k \in \Omega }u_{v_i}(1,k) \geq \sum_{k \in \Omega }u_{v_i}(2,k)$. Note that the set of posterior beliefs at which voter $v_i$ votes for politician 1, $O_1(v_i) \subseteq \Delta(\Omega)$, is convex. 

To apply our results, we need to show that in any equilibrium, poiliticians' preferences over posterior beliefs fit our assumptions. In an equilibrium abiding by our selection, let $O_1 \subseteq \Delta(\Omega)$ be the set of beliefs at which the voters elect politician 1 and $O_2$ be the set of beliefs they elect politician 2. At beliefs in $O_1$, politician 1 (2) gets payoffs $1$ ($0$), while in $O_2$ politician 1 (2) gets payoff $0$ ($1$). It is easy to show that $O_1$ and $O_2$ can each be written as the union of a finite number of disjoint convex sets;\footnote{The regions in which an individual voter votes for politicians 1,2 are defined by a single hyperplane. Drawing all $V$ such hyperplanes partitions the simplex into finitely many convex cells. All voters' actions are unchanged when moving within a cell. Hence each cell is either in $O_1$ or $O_2$.} hence politicians' payoffs are real analytic and zero-sum. 

For any $\Omega' \subseteq \Omega$, senders have linear preferences on $\Delta(\Omega')$ if and only if $\Delta(\Omega') \subseteq O_1$ or $\Delta(\Omega') \subseteq O_2$. If $\Delta(\Omega')$ intersects both $O_1$ and $O_2$, then sender payoffs jump discontinuously at some point in $\Delta(\Omega')$. Applying Proposition 2: $\Omega'$ is pooled in some equilibrium if and only if $\Delta(\Omega') \subseteq O_1$ or $\Delta(\Omega') \subseteq O_2$. This implies that in equilibrium, at every posterior no additional information would affect the election outcome. By Theorem 1, the state is fully revealed in every equilibrium if and only if for every pair of states $l,k$, $\Delta(\{l,k\})$ intersects $O_1$ and $O_2$. 

There are a couple interesting features of this setting which we will discuss via the following example. Suppose $\Omega=\{0,1\}$ with a flat prior, $V=3$, and $u_{v_1},u_{v_2},u_{v_3}$ are such that the following is true. If voter 1 knew the state, she would prefer politician 1 when $\omega=0$ and politician 2 when $\omega=1$; however, her preference for politician 2 when $\omega=1$ is stronger than her preference for politician 1 when $\omega=0$.\footnote{i.e. $u_{v_1}(2,1) - u_{v_1}(1,1)>u_{v_1}(1,0) - u_{v_1}(2,0)$.} Meanwhile voter 2 prefers politician 1 at $\omega=1$ and politician 2 at $\omega=0$ but her preference at $\omega=0$ is stronger than that at $\omega=1$. Finally voter 3 always prefers politician 1. Figure \ref{fig_persuading_voters} shows what $O_1(v_i),O_2(v_i)$ look like for all voters.

\begin{center}
\includegraphics[scale=0.8]{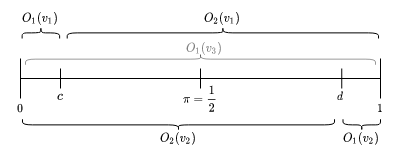}
\captionof{figure}{}
\label{fig_persuading_voters}
\end{center}

Note that while $O_1(v_i)$ and $O_2(v_i)$ are convex for each $v_i$, $O_1$ and $O_2$ need not be convex. In the example, $O_2 = (c,d)$ while $O_1 = [0,c) \cup (d,1]$. Hence despite the fact that the voters would elect politician 1 at both states $0$ and $1$, politicians must still fully reveal the state in all equilibria as there are interior beliefs at which voters would make a different decision. This is in contrast to a game with a single receiver where the receiver takes the same action on convex sets.

\cite{alonso2016persuading} find that relative to no information, a single persuader can make a majority of voters strictly worse off. Although competition guarantees all equilibria of our example provide voters with full information, this still makes a majority of voters worse off. In equilibrium, when $\omega=0$, voter 1 gets an outcome ($O=1$) she marginally prefers given the state. With equal probability, $\omega=1$ and voter 1 gets an outcome she strongly disprefers. Meanwhile, under no information, politician 2 wins the election; this is the marginally dispreferred outcome for voter 1 when $\omega=0$ and the strongly preferred outcome when $\omega=1$. In net, voter 1 does better under no information; the same is true for voter 2 by identical logic. Voter 3, meanwhile, does better under full information.\footnote{\cite{alonso2016persuading} show that a majority of voters can be made worse off by a single persuader even if at every state all voters agree on the best outcome. With competing persuaders, such agreement between voters guarantees all voters do better off in equilibrium than under no information.} If she was the sole persuader, sender 2 would like to reveal no information; hence voters 1 and 2 do worse under competition than under unilateral persuasion by sender 2. Sender 1, meanwhile, would like to fully reveal the state even in the abscence of sender 2.\footnote{Senders 2 and 1 each have multiple Bayesian Persuasion solutions, but each are outcome equivalent to no information and full revelation respectively.}

The example adds some nuance to our results on information provision in equilbrium. While the effects of zero-sum competition are unambiguously positive for a single receiver, the welfare implications for multiple receivers will depend on the environment.

\section{Discussion of Assumptions}

Our two most substantive assumptions in this paper are that preferences are zero-sum and senders have access to only conditionally independent experiments. In this section we examine to what extent we can relax each of these.

\subsection{Zero-sum preferences}

Zero-sum preferences correspond to maximal disagreement between senders,\footnote{Formally, whenever $U_i(\Gamma_1,...,\Gamma_M)>U_i(\Gamma_1',...\Gamma_M')$ then there exists a sender $j$ with $U_j(\Gamma_1,...,\Gamma_M)<U_j(\Gamma_1',...\Gamma_M')$. Hence zero-sum preferences maximize the set of strategy profiles pairs for which senders do not unanimously agree on the ranking.} and hence are an important benchmark to consider for the applications we had in mind: those in which senders have strongly opposing interests. Under this benchmark, we were able to characterize properties of all equilibria\footnote{Proposition \ref{prop_robustness} tells us that these results apply close to this benchmark as well.} and show that competition cannot decrease information provision. In this section we demonstrate  that neither of these findings apply away from zero-sum games.

It is easy to see why our full revelation results collapse in a more general setting: if we allow for arbitrary sender preferences then it is possible for all senders to have identical preferences; in this case any sender's Bayesian Persuasion solution will be an equilibrium and, of course, need not be fully revealing. More interestingly, when the game is not zero-sum and senders employ conditionally independent experiments, the effect of competition on information provision is ambiguous. We can find examples of nonzero-sum piecewise real analytic preferences such that adding a sender can create an equilibrium strictly less informative than all equilibria without the additional sender. In the one parameter model of disagreement between two senders studied in \citet{gentzkow2016competition} we can show that increasing disagreement between senders can produce strictly less informative equilibria.\footnote{At one limit of the parameter is zero-sum preferences.} This suggests that while in the limit of maximal disagreement competition cannot decrease information provision, information provision is not monotonic in competition. Hence another justification for our limiting the scope of sender preferences is an inability to answer our motivating question outside this scope.

\subsection{Limitations in technology}

The main results of this paper show that when senders have zero-sum preferences and have access to all conditionally independent experiments, then competition cannot decrease equilibrium information. This happens starkly \textemdash typically any such competition results in all information being revealed. 

The arguments we use have relied on our assumptions on the information technology senders have access to. In particular, Lemma 1 property (1) relied on each sender being able to fully reveal the state and Lemma 1 propery (2) relied on a sender being able to construct any experiment using $N$ signals that is conditionally independent of her opponents'. In this section, we first show that these two ingredients are sufficient for our results; in particular if other (potentially non-conditionally independent) experiments are additionally available to senders then Propositions 1 and 2 and Theorem 1 go through. We then argue that while the fact that zero-sum competition results in nondecreasing information may not be shocking, it is not obvious; if senders do not have access to a rich enough set of conditionally independent experiments then competition may in fact decrease the amount of information provided in equilibrium. Finally, we consider a few examples which demonstrate that our baseline assumptions on technology are not necessary for our results \textemdash even when senders are limited to picking experiments from some simple restrictive classes of conditionally experiments, we still can maintain our main results.

\subsubsection{Additional experiments}

The case of conditionally independent (CI) experiments is realistic in many real-world settings. However, in some environments it is possible that senders may be able, to an extent, to correlate the results of their experiments. For instance, in criminal cases tried in United States federal court, the prosecutor may be bound to disclose (or place in discovery) certain types of evidence at various deadlines before the trial. The defense may hence (partially) condition what exonerating evidence they seek on what information the prosecutor discloses. In the case of politicians competing for votes, campaigns are dynamic processes and a candidate may commission experts to write reports informative about her agenda after her opponent has already done so. Importantly, while a sender in these examples is able to partially condition what information she seeks on what her opponent has revealed, she need not do so. A defense attorney is free to ignore the evidence put into disclosure by the prosecution and a politician is free to conduct her campaign ignoring information revealed by her competitors. 

Our results also apply in these types of settings. Conditional independence is not crucial to our results in the sense that they still hold if senders have access to additional (partially or arbitrarily correlated) strategies. Fix a set of signals $S$, $|S| \geq N$. Suppose each sender $i$ has access to a set of experiments $E_i$. We say sender $i$ has \emph{access to all CI experiments} if for any finite subset of signals $S_i \subseteq S$ $i$ chooses, $E_i$ contains every conditionally independent experiment $\Pi_i: \Omega \rightarrow \Delta(S_i)$. Sender $i$ may have access to all CI experiments and also have access to other (finite signal) experiments which realize to signals in $S_i$.

\begin{corollary} \label{cor_add_exp}
If every sender has access to all CI experiments then Propositions 1,2 and Theorem 1 hold.
\end{corollary}

The important part of Corollary \ref{cor_add_exp} is showing that Proposition 2 holds in such an environment; Proposition 1 is then implied and Theorem 1 follows by the same argument as in the baseline model. First note that if all senders have access to all CI experiments then Lemma 1 property (1) holds (for the same reason). Fixing an experiment $\Gamma_{-i}$, we can define $W_i(x)$, as before, as $i$'s expected payoff from generating interim belief $x$ from an experiment conditionally independent to $\Gamma_{-i}$. Next note that Lemma 1 propery (2) holds here as well; just as in the baseline model, whenever $\Gamma_{-i}$ is such that $W_i(x)>0$ for some $x$, $i$ can find a conditionally independent $\Gamma_i$ such that $U_i(\Gamma_i,\Gamma_{-i})>0$ (violating Lemma 1 property (1)). Extending the `if' direction of Proposition 2 to this setting then follows by and identical argument as in the baseline model; the `only if' direction holds by an identical construction of pooling equilibria.

It is worth noting that adding correlated experiments to senders' strategy spaces augments the set of deviations they can play and hence it could be \emph{easier} to support non-fully revealing equilibria when only CI experiments are available. Our analysis shows that CI deviations are sufficient to eliminate non-fully revealing equilibria whenever there is global nonlinearity in preferences.

\citet{gentzkow2017bayesian} consider a multi-sender Bayesian persuasion game in which senders are allowed to arbitrarily correlate their signal realizations and obtain a sufficient condition for full revelation in all equilibria of zero-sum games almost identical to ours. Corollary \ref{cor_add_exp}, combined with Theorem 1, nests this result of GK (2017) with the caveat that we assume piecewise analytic utilities. 

It is important to note the forces that deliver the full revelation results in GK (2017) and this paper are different. In GK (2017), ability to correlate experiments gives each sender much more control over the posterior. Given any $\Gamma_{-i}$ played by her opponents, a sender $i$ can play a different experiment for \emph{each} realization of $\Gamma_{-i}$. Senders' ability to manipulate the receiver's posterior belief by belief makes Proposition 2, and hence Theorem 1, much easier to prove. In our setting senders have less control over posteriors; the strongest tool a sender has is using extreme interim beliefs to ensure poteriors are similarly extreme. The mechanisms by which senders can take advantage of the state not being fully revealed in our model require less complexity. As we discuss later in the section, a key contribution of this paper is to show zero-sum competition in persuasion typically generates full information even when senders have access to lower complexity technology.

\subsubsection{Limited conditionally independent experiments}

Ex-ante, one may think that zero-sum competition is most favorable to information production and hence that our results are not surprising. The following example demonstrates that these results are not obvious and depend on the set of strategies senders have access to. 

\begin{example} \label{ex_0sum_dec}
Suppose $N=M=2$ and let $\pi = \frac{1}{2}$. Sender preferences are zero-sum and are defined as follows: 
\begin{itemize}
\item $u_1(\beta)=u_2(\beta)=0$ for all $\beta \not \in \{0.25,0.26,0.74,0.75\}$
\item $u_1(0.25)=u_1(0.75)=u_2(0.26)=u_2(0.74)=1$
\end{itemize}   

Define experiments $\Gamma^1,\Gamma^2,\Gamma^N$ as follows: \begin{itemize}
\item $Pr(\Gamma^1=0.25)=Pr(\Gamma^1=0.75)=\frac{1}{2}$
\item $Pr(\Gamma^2=0.26)=Pr(\Gamma^2=0.74)=\frac{1}{2}$
\item $Pr(\Gamma^N=\frac{37}{76})=Pr(\Gamma^N=\frac{39}{76})=\frac{1}{2}$
\end{itemize}Suppose $E_1=\{\Gamma^U,\Gamma^1,\Gamma^N\}$ and $E_2 = \{\Gamma^U,\Gamma^2,\Gamma^N\}$.
\end{example}
\begin{center}
\includegraphics[scale=0.6]{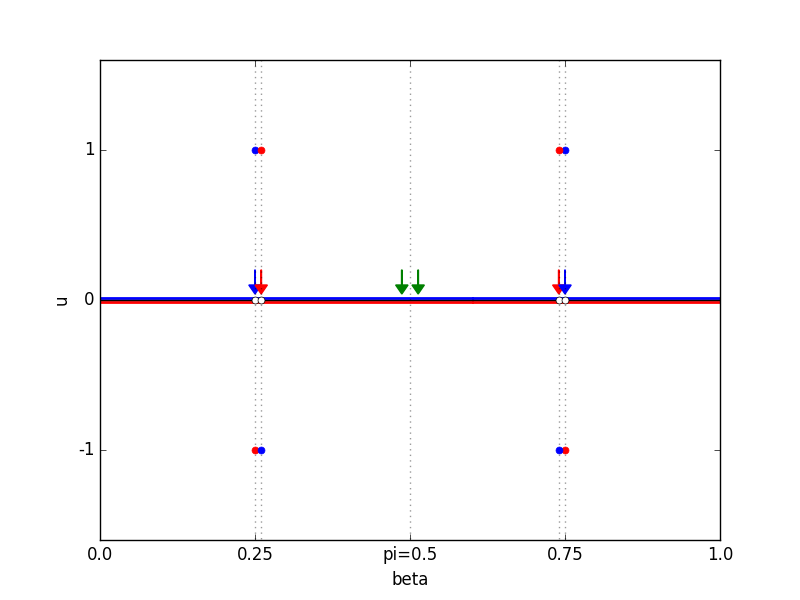}
\captionof{figure}{Plots $u_1$ (blue) and $u_2$ (red) in Example \ref{ex_0sum_dec}. The blue arrows show the possible interim belief realizations for $\Gamma^1$. The red and green arrows show the same for $\Gamma^2$ and $\Gamma^N$ respectively.}
\end{center}

$\Gamma^1$ and $\Gamma^2$ here are experiments that realize to sender 1 and sender 2's points of advantage respectively. $\Gamma^N$ provides extremely little information.

Note that $\beta(0.74,\frac{39}{76}) = 0.75$, $\beta(0.75,\frac{37}{76})=0.74$, $\beta(0.26,\frac{37}{76})=0.25$, and $\beta(0.25,\frac{39}{76})=0.26$. Given this, both players have best response $\Gamma_i=\Gamma^N$ whenever $\Gamma_{-i}=\Gamma^N$ or $\Gamma_{-i} =\Gamma^{-i}$. Sender $i$'s best response to $\Gamma_{-i}=\Gamma^i$ is $\Gamma_i=\Gamma^U$. With two senders, there is a unique equilibrium $(\Gamma^N,\Gamma^N)$. When only sender $i \in \{1,2\}$ is playing, the unique `equilibrium' (Bayesian Persuasion solution) is $\Gamma_i = \Gamma^i$. Note that $\Gamma^1$ and $\Gamma^2$ are both strictly more informative than the strategy profile $(\Gamma^N,\Gamma^N)$.\footnote{To see this note that $\Gamma^1$ and $\Gamma^2$ both only induce posteriors outside interval $(0.26,0.74)$. $(\Gamma^N,\Gamma^N)$ induces posteriors inside $(0.26,0.74)$; as we are in a binary-state setting, this means that $\Gamma^1$ and $\Gamma^2$ are both mean preserving spreads of the $(\Gamma^N,\Gamma^N)$.}

Consider the game with 2 senders. As no feasible strategy profile fully reveals the state, it is not surprising that our full revelation results collapse in this example. What is more interesting is that zero-sum competition here results in less information in equilibrium  than the receiver(s) would obtain with any one sender being the sole persuader.\footnote{Note that if both senders were able to fully reveal the state, then while there would be a fully revealing equilibrium, $(\Gamma^N,\Gamma^N)$ would remain an equilibrium. Hence going from 1 to 2 senders would still create a strictly less informative equilibrium while also creating a strictly more informative one.}\footnote{This fact does not depend on senders having different strategy sets. We can construct a similar but more cumbersome example where $E_1=E_2$.} Hence if senders have limited sets of experiments available, it is possible for competition to decrease equilibrium information. In this example, this occurs because both senders play low information experiments, $\Gamma^N$, which given the available experiments make it impossible for their opponent to make use of her advantage points. Though one may think zero-sum competition should intutively lead to more information, this is only true when senders can choose from rich enough sets of experiments.

The following examples and results show that while very limited sets of CI experiments will break our results, if senders have access to some simple classes of experiments but not all CI experiments, our results go through.

\begin{example} \label{ex_chemical}
Two lobbyists, $1,2$, persuade a politician to vote yes/no ($Y$/$N$) on a bill. They do so by commissioning scientists to conduct studies revealing information about a state $\omega \in \{0,1\}$. The politician will vote yes on the bill if at her posterior $Pr(\omega =1) \geq 0.5$ and no otherwise. Lobbyist $i$ receives payoff, $u_i(a,\omega)$ where $a \in \{Y,N\}$ is the politician's vote. As in Section 4.1, we assume all senders have strict preferences over actions at both states and that payoffs are zero-sum for each $(a,\omega)$. 

As an example, suppose the bill would place restrictions on chemicals manufacturing firms can use in producing household goods. If $\omega=1$ these chemicals are harmful and if $\omega=0$ they are not.  Lobbyist $1$ represents the firms and wants the bill to fail, say $u_1(N,\omega)>u_1(Y,\omega)$ for all $\omega$, whereas lobbyist $2$ is a consumer protection advocate and has the opposite preferences. Each lobbyist has a scientist (scientists $1,2$) who will study the chemical's effects on humans. Each scientist $1,2$ has access to a standard test they can use on a human subject to determine the chemical's effect on the subject; the only dimension the lobbying firms can control is how many subjects their scientist tests. 

We assume both scientists can employ the same standard test, represented by experiment $\Gamma$; assume $\Gamma \not \in \{\Gamma^U,\Gamma^{FR}\}$ \textemdash testing an individual subject is not totally uninformative or fully informative about the effects on the population. Let experiment $\Gamma^K$ represent the experiment induced by repeating $\Gamma$ $K$ times conditionally independently. Lobbyist $i$ can play experiment $\Gamma^K$ by asking scientist $i$ to test $K$ subjects. We assume scientist $i$ can repeat $\Gamma$ an arbitrary number of times. As $K \rightarrow \infty$, note that $\Gamma^K$ will fully reveal the state; this is because $\Gamma$ is not totally uninformative and hence an arbitrarily large number of copies will reveal the state with arbitrary certainty (we discuss this below). As the limit of infinite repetitions, we allow each sender $i$ to fully reveal the state as well. Hence sender $i$ has access to experiments $E_i = \{\Gamma^K\}_{K=1}^{\infty} \cup \{\Gamma^{FR}\}$.

Each sender has access to a much smaller set of experiments than the set of all CI experiments. However, each can fully reveal the state; this implies Lemma 1 property (1) must hold. Further, by choosing large enough $K$, a sender can ensure the state is almost fully revealed with high probabilility and that posterior beliefs concentrate close to $0$ and $1$. If $\Gamma_{-i}$ is not fully revealing the state and sender $i$ has advantages close to $0$ and $1$, sender $i$ can obtain a strictly positive payoff by playing $\Gamma^K$ for large enough $K$. If $i$ has an advantage close to $1$ but a disadvantage close to $0$, this sort of deviation may still give $i$ a strictly positive payoff depending on the relative sizes of the advantage and disadvantage. It turns out, that for generic payoffs $\{u_i(a,\omega)\}$, some sender will be able to obtain a strictly positive payoff when her opponent is not fully revealing the state. This implies:

\textbf{Result.} For generic payoffs the state is fully revealed in every equilibrium.
\end{example}

The intuition from Example \ref{ex_chemical} extends more generally. We say an experiment $\Gamma$ is \emph{asymptotically sufficient} if: (1) $\Gamma \neq \Gamma^{FR}$ and (2) $\Gamma^K$ converges in distribution to $\Gamma^{FR}$ as $K \rightarrow \infty$.  In the binary-state case, $\Gamma$ is asymptotically sufficient if and only if $\Gamma \neq \Gamma^U$. More generally, $\Gamma$ is assymptotically sufficient if the convex hull of its support has dimension $N-1$; when this is true, then repeating $\Gamma$ generates sufficient information about the relative probabilities of states to guarantee convergence to $\Gamma^{FR}$.\footnote{Versions of this result are known, but we were unable to find a suitable version of the result in the literature and hence prove it in Supplementary Appendix A.} If each player has access to $\Gamma^{FR}$ and an asymptotically sufficient experiment which she can repeat arbitrarily, then we can find a sufficient condition for the state to be fully revealed in all equilibria (satisfied in all but a knife-edge case).\footnote{If no sender has access to the fully revealing experiment then there can be issues with equilibrium existence. $\Gamma^{FR}$ needs to be included as a limiting case of infinite repetitions.} Here, as the set of pure strategies is itself quite coarse, we consider the class of all mixed strategy Nash Equilibria with support on a finite number of pure strategies. The result is easiest to state in the case of a single receiver with finite actions:

\begin{prop} \label{prop_repeat}
For each sender $i$ let $E_i = \{\Gamma^K\}_{K=1}^{\infty} \cup \{\Gamma^{FR}\}$ for some asymptotically sufficient $\Gamma$. There is a single receiver with a finite action set who breaks indifferences in favor of higher actions. Generically, the state is fully revealed in every mixed strategy Nash Equilibrium with support on finite pure strategies if the receiver has a different best action at some two states.
\end{prop} 

Note that the result is slightly stronger in implication than \ref{cor_finite_action}, as we can have get full revelation in all equilibria even if the receiver prefers the same action at some two states. Substantively, the implication is the same: the receiver will always learn enough to take her first-best action. As a sender $i$ repeats an asymptotically sufficient informative experiment, interim beliefs $\Gamma^K$ converge to full revelation. For any fixed $\Gamma_{-i}$, posterior beliefs do as well. Depending on $i$'s relative advantage close to each $\delta_l$, this may be good for $i$ (for instance if $i$ has an advantage close to all $\delta_l$). We show that generically if information ever affects the receiver's action (i.e. she prefers a different action at some two states), some sender $i$ can obtain a strictly positive payoff from $\Gamma^K$ as $K \rightarrow \infty$;  this deviation is always available to sender $i$ and hence Lemma 1 property (1) (extended to this setting) implies the result. 

Proposition \ref{prop_repeat} gives us a simple class of conditionally independent experiments that is sufficient for the receiver to learn adequately (or, enough that no further information would help) under zero-sum competition. This class is in a sense coarse, but provides senders with the necessary flexibility: ability to force posteriors into extreme regions of advantage. We give one other example of a set of experiments that does the job.

Suppose for all $i$ $E_i = \{\Gamma : \exists \alpha \in [0,1] \text{ s.t. } \forall k\in \Omega \text{, } Pr(\Gamma=\alpha \delta_k + (1-\alpha) \pi) = \pi_k \}$. Each $E_i$ contains all convex combination of $\Gamma^U$ and $\Gamma^{FR}$; this can be interpretted as senders having access to only these two experiments and privately randomizing which experiment they employ. With these strategies, in the single receiver model, again, the state is fully revealed in every equilibrium if the receiver has a different best action at some two states. The argument is similar to that in the previous example: again senders as able to fully reveal the state (so Lemma 1 property (1) holds), and are able to almost fully reveal the state.

\section{Conclusion}

We study a multi-sender Bayesian Persuasion game. The substantive assumption is that senders are maximally competitive and have zero-sum preferences over the the receiver's posterior belief. In our baseline model senders employ conditionally independent experiments and we show that for typical sender preferences, the state is fully revealed in every equilibrium. Further, we show that zero-sum competition cannot decrease equilibrium information provision. Our results do not critically rely on conditional independence: they apply when senders have access to technology to that is richer than the set of all conditionally independent experiments as well as, in some cases, when sender technology is coarser. 

In the paper we consider various real-world applications of our model. Many of these applications involve games with a single receiver who chooses from a finite actions set. We show that in these settings our results take a clean form: the receiver always learns enough to attain her first-best payoff. While our results are very positive for a single receiver, the consequences of full revelation (or of a lot of information being revealed in equilibria) on receiver welfare need not be positive in games with multiple receivers with conflicting interests. While our results indicate that persuaders with opposing interests will tend to produce a lot of information in equilibrium, whether this is `good' or not will depend on the setting.

\bibliographystyle{abbrvnat}
\bibliography{ref.bib}{}

\begin{thebibliography}{23}
\providecommand{\natexlab}[1]{#1}
\providecommand{\url}[1]{\texttt{#1}}
\expandafter\ifx\csname urlstyle\endcsname\relax
  \providecommand{\doi}[1]{doi: #1}\else
  \providecommand{\doi}{doi: \begingroup \urlstyle{rm}\Url}\fi

\bibitem[Alonso and C{\^a}mara(2016)]{alonso2016persuading}
R.~Alonso and O.~C{\^a}mara.
\newblock Persuading voters.
\newblock \emph{American Economic Review}, 106\penalty0 (11):\penalty0
  3590--3605, 2016.

\bibitem[Au and Kawai(2020)]{au2020competitive}
P.~H. Au and K.~Kawai.
\newblock Competitive information disclosure by multiple senders.
\newblock \emph{Games and Economic Behavior}, 119:\penalty0 56--78, 2020.

\bibitem[Barnoya and Glantz(2006)]{barnoya2006tobacco}
J.~Barnoya and S.~A. Glantz.
\newblock The tobacco industry's worldwide ets consultants project: European
  and asian components.
\newblock \emph{The European Journal of Public Health}, 16\penalty0
  (1):\penalty0 69--77, 2006.

\bibitem[Battaglini(2002)]{battaglini2002multiple}
M.~Battaglini.
\newblock Multiple referrals and multidimensional cheap talk.
\newblock \emph{Econometrica}, 70\penalty0 (4):\penalty0 1379--1401, 2002.

\bibitem[Blackwell(1953)]{blackwell1953equivalent}
D.~Blackwell.
\newblock Equivalent comparisons of experiments.
\newblock \emph{The annals of mathematical statistics}, pages 265--272, 1953.

\bibitem[Boleslavsky and Cotton(2018)]{boleslavsky2018limited}
R.~Boleslavsky and C.~Cotton.
\newblock Limited capacity in project selection: Competition through evidence
  production.
\newblock \emph{Economic Theory}, 65\penalty0 (2):\penalty0 385--421, 2018.

\bibitem[Che et~al.(2019)Che, Kim, and Kojima]{che2019weak}
Y.-K. Che, J.~Kim, and F.~Kojima.
\newblock Weak monotone comparative statics.
\newblock \emph{Available at SSRN 3486620}, 2019.

\bibitem[Dworczak and Pavan(2020)]{dworczak2020robust}
P.~Dworczak and A.~Pavan.
\newblock Robust (bayesian) persuasion.
\newblock \emph{Available at SSRN 3523114}, 2020.

\bibitem[Emons and Fluet(2019)]{emons2019strategic}
W.~Emons and C.~Fluet.
\newblock Strategic communication with reporting costs.
\newblock \emph{Theory and Decision}, 87\penalty0 (3):\penalty0 341--363, 2019.

\bibitem[Gentzkow and Kamenica(2017a)]{gentzkow2016competition}
M.~Gentzkow and E.~Kamenica.
\newblock Competition in persuasion.
\newblock \emph{The Review of Economic Studies}, 84\penalty0 (1):\penalty0
  300--322, 2017a.

\bibitem[Gentzkow and Kamenica(2017b)]{gentzkow2017bayesian}
M.~Gentzkow and E.~Kamenica.
\newblock Bayesian persuasion with multiple senders and rich signal spaces.
\newblock \emph{Games and Economic Behavior}, 104:\penalty0 411--429, 2017b.

\bibitem[Kamenica and Gentzkow(2011)]{kamenica2011bayesian}
E.~Kamenica and M.~Gentzkow.
\newblock Bayesian persuasion.
\newblock \emph{American Economic Review}, 101\penalty0 (6):\penalty0
  2590--2615, 2011.

\bibitem[Kartik et~al.(2017)Kartik, Lee, and Suen]{kartik2017investment}
N.~Kartik, F.~X. Lee, and W.~Suen.
\newblock Investment in concealable information by biased experts.
\newblock \emph{The RAND Journal of Economics}, 48\penalty0 (1):\penalty0
  24--43, 2017.

\bibitem[Kartik et~al.(2020)Kartik, Lee, and Suen]{kartik2020information}
N.~Kartik, F.~Lee, and W.~Suen.
\newblock Information validates the prior: A theorem on bayesian updating and
  applications.
\newblock \emph{arXiv preprint arXiv:2005.05714}, 2020.

\bibitem[Koessler et~al.(2022)Koessler, Laclau, and
  Tomala]{koessler2022interactive}
F.~Koessler, M.~Laclau, and T.~Tomala.
\newblock Interactive information design.
\newblock \emph{Mathematics of Operations Research}, 47\penalty0 (1):\penalty0
  153--175, 2022.

\bibitem[Krishna and Morgan(2001)]{krishna2001model}
V.~Krishna and J.~Morgan.
\newblock A model of expertise.
\newblock \emph{The Quarterly Journal of Economics}, 116\penalty0 (2):\penalty0
  747--775, 2001.

\bibitem[Li and Norman(2018{\natexlab{a}})]{li2018bayesian}
F.~Li and P.~Norman.
\newblock On bayesian persuasion with multiple senders.
\newblock \emph{Economics Letters}, 170:\penalty0 66--70, 2018{\natexlab{a}}.

\bibitem[Li and Norman(2018{\natexlab{b}})]{li2018sequential}
F.~Li and P.~Norman.
\newblock Sequential persuasion.
\newblock \emph{Available at SSRN 2952606}, 2018{\natexlab{b}}.

\bibitem[Milgrom and Roberts(1986)]{milgrom1986relying}
P.~Milgrom and J.~Roberts.
\newblock Relying on the information of interested parties.
\newblock \emph{The RAND Journal of Economics}, pages 18--32, 1986.

\bibitem[Milgrom and Shannon(1994)]{milgrom1994monotone}
P.~Milgrom and C.~Shannon.
\newblock Monotone comparative statics.
\newblock \emph{Econometrica: Journal of the Econometric Society}, pages
  157--180, 1994.

\bibitem[Shin(1998)]{shin1998adversarial}
H.~S. Shin.
\newblock Adversarial and inquisitorial procedures in arbitration.
\newblock \emph{The RAND Journal of Economics}, pages 378--405, 1998.

\bibitem[Veinott(1989)]{veinott1989lattice}
A.~F. Veinott.
\newblock Lattice programming.
\newblock \emph{Unpublished notes from lectures delivered at Johns Hopkins
  University}, 1989.

\bibitem[Wu(2017)]{wu2017coordinated}
W.~Wu.
\newblock Coordinated sequential bayesian persuasion.
\newblock Technical report, Working Paper, University of Arizona, 17-18, 2017.

\end{thebibliography}

\appendix
\section{Appendix A: Proofs}

\textbf{Definitions and facts.} 
The following definitions and facts are used in both Appendix A and Supplementary Appendix B.

Let $P$ be the set of elements of $\Delta(\Delta(\Omega))$ which are Bayes-plausible, have finite support, and have support on at most $|S|$ beliefs. It will be convenient to talk about a strategy for sender $i$ as a choice of interim belief $\Gamma_i$ with probability mass function $p_i \in P$ (in the text of the paper we did not introduce notation for the distribution of $\Gamma_i$).

For any strategy profile $(\Gamma_1,...,\Gamma_M)$ and subset of senders $S \subseteq \{1,...,M\}$, let the random variable $\Gamma_S$  be the receiver's belief after observing realizations of $\{\Gamma_j\}_{j \in S}$ but not the realizations of $\{\Gamma_j \}_{j \not \in S}$; let $p_S$ be it's probability mass function and $p_S(\cdot|\omega=k)$ be its probability mass function conditional on the state being $k$. Let $\Gamma_{-S}$ and $p_{-S}$ be the same objects for the complementary set of senders.

For any disjoint subsets of senders $S,S' \subset \{1,...,M\}$ and any fixed strategy profile $(\Gamma_1,...,\Gamma_M)$, let $p_{S'}(\cdot|x)$ be the probability mass function of $\Gamma_{S'}$ conditional on $\Gamma_S=x$.

\begin{equation} \label{p_S}
\begin{split}
& p_{S'}(y|x) = \sum_{k \in \Omega} p_{S'}(y|\omega=k,x)Pr(\omega=k|\Gamma_S=x) = \sum_{k \in \Omega} p_{S'}(y|\omega=k) x_k \\
&  = \sum_{k \in \Omega} \frac{Pr(\omega=k|y)p_{S'}(y)}{Pr(\omega=k)} x_k = \sum_{k \in \Omega} \frac{x_k y_k p_{S'}(y)}{\pi_k}
\end{split}
\end{equation}

Where the second equality comes from conditional independence of $\Gamma_S$ and $\Gamma_{S'}$. Claim \ref{claim_p_S'} tells us that conditional on $\Gamma_S$, with probability $1$ $\Gamma_{S'}$ assigns positive probability to at least one state that $\Gamma_{S}$ assigns positive probability to (i.e. $\Gamma_{S'}$ cannot contradict $\Gamma_{S}$). This is a simple implication of Bayesian updating.

\begin{claim} \label{claim_p_S'}
For any disjoint subset of senders $S,S'$, $\Gamma_S$,$\Gamma_{S'}$, and $x \in \Delta(\Omega)$: $p_{S'}(y|x) = 0$ for all $y$ s.t. $y_l =0 $ for all $l \in \Omega$ for which $x_l >0$. Further, there exists $y \in supp[\Gamma_{S'}]$ such that $p_{s'}(y|x)>0$.
\end{claim}

\begin{proof}
The first statement, that $p_{S'}(y|x) = 0 $ for all $y$ such that $y_l =0$ for all $l$ for which $x_l >0$, follows immediately from Equation \ref{p_S}. The second statement follows from Bayes-plausibility of $\Gamma_{S'}$. For every $l$ such that $x_l >0$, as $\pi_l >0$, there exists $y \in supp[\Gamma_{S'}]$ with $y_l \geq \pi_l > 0 $; by Equation \ref{p_S}, $p(y|x)>0$ for such a $y$.
\end{proof}

Let $\beta_l(x^1,...,x^M)=Pr(\omega=l|x^1,...,x^M)$ be the receiver's posterior belief that $\omega=l$ after observing experiment realizations $\Gamma_1 = x^1,...,\Gamma_M = x^M$. By Bayes rule:

\begin{equation} \label{beta_def}
\begin{split}
& \beta_l(x^1,...,x^M)= \frac{Pr(\Gamma_1=x^1,...,\Gamma_M=x^M|\omega=l)Pr(\omega=l)}{Pr(\Gamma_1=x^1,...,\Gamma_M=x^M)} \\
& = \frac{[\Pi_{i=1}^M p_i(x^i|\omega=l)] \pi_l}{\sum_{k=1}^N [\Pi_{i=1}^M p_i(x^i|\omega=k)]Pr(\omega=k)} = \frac{[\Pi_{i=1}^M \frac{Pr(\omega=l|x^i) p_i(x^i)}{Pr(\omega=l)}] \pi_l}{\sum_{k=1}^N [\Pi_{i=1}^M \frac{Pr(\omega=k|x^i) p_i(x^i)}{Pr(\omega=k)}] \pi_k} \\
& =\frac{[\Pi_{i=1}^M x^i_l] / \pi_l^{M-1}}{\sum_{k=1}^N [\Pi_{i=1}^M x^i_k] / \pi_k^{M-1}} 
\end{split}
\end{equation}

Where the second equality uses the conditional independence of $\Gamma_1,...,\Gamma_M$. Note that $\beta_l$ is not well defined when for each state $k \in \Omega$ there exists sender $j$ with $x^j_k=0$. However it is straightforward to see by applying Claim \ref{claim_p_S'} that such a realization of $(\Gamma_1,...,\Gamma_M)$ occurs with zero probability; after viewing the realizations of any number of experiments, the Bayesian receiver will have a well defined posterior w.p. $1$.

For any strategy profile $(\Gamma_1,...,\Gamma_M)$ and disjoint sets of senders $S_1,...,S_T$, we similarly define define the receiver's posterior as a function of interim belief realizations from each experiment: $\{\Gamma_{S_s}=y^{S_s}\}_{s=1,...,T}$. 

\begin{equation} \label{beta_def_sets}
\beta_l(y^{S_1},...,y^{S_T}) = \frac{\Pi_{s=1}^T y^{S_s}_l / \pi_l^{T-1}}{\sum_{k=1}^N \Pi_{s=1}^T y^{S_s}_k / \pi_k^{T-1}}
\end{equation}

Note: $\Gamma_{S_1 \cup ... \cup S_T} = \beta(\Gamma_{S_1},...,\Gamma_{S_T})$, as both define the receiver's belief after observing realizations of $\Gamma_{S_1},...,\Gamma_{S_T}$.

Claim \ref{claim_projection} shows that if any subset of experiments in a strategy profile generate an interim belief in $\Delta(\Omega')$ then the posterior will fall in $\Delta(\Omega')$ w.p. $1$.

\begin{claim} \label{claim_projection}
For any strategy profile $(\Gamma_1,...,\Gamma_M)$, disjoint subsets of senders $S_1,...,S_T$, and states $\Omega' \subseteq \Omega$, if $\Gamma_{S_1} \in \Delta(\Omega')$ then $\beta(\Gamma_{S_1},...,\Gamma_{S_T}) \in \Delta(\Omega')$ w.p. $1$. 
\end{claim}

\begin{proof}
This can be seen from the definition of $\beta(y^{S_1},...,y^{S_T})$ which implies $\beta_l(y^{S_1},...,y^{S_T}) = 0$ for all $l \not \in \Omega'$. After observing $\Gamma_{S_1} \in \Delta(\Omega')$, the receiver updates to an interim belief assigning $0$ probability to all states outside of $\Omega'$. No additional information can change this.
\end{proof}

\begin{claim} \label{not_pool}
For any strategy profile $(\Gamma_1,...,\Gamma_M)$, $\Omega' \subseteq \Omega$, and any subsets of senders $S$: If $\Gamma_S$ does not pool $\Omega'$ then $(\Gamma_1,...,\Gamma_M)$ does not either.
\end{claim}

\begin{proof}
Let $S' = \{1,...,M\} \setminus \{S\}$. As $\Gamma_S$ does not pool $\Omega'$, then $Pr(\Gamma_S=y \text{ : s.t. } y_l>0 \text{ } \forall l \in \Omega')=0$. If $\Gamma_S=y$, $\Gamma_{S'}=y'$, then by Equation \ref{beta_def_sets}, if $y_l =0$ then $\beta_l(y,y')=0$. Hence as w.p. $1$ $\Gamma_S$ assigns $0$ probability to at least one state in $\Omega'$, $\beta(\Gamma_S,\Gamma_{S'})$ does as well and so $(\Gamma_1,...,\Gamma_M)$ does not pool $\Omega'$.
\end{proof}

\subsection{Section 2}

\textbf{Normalization of utility functions.} Here we show that we can normalize $u_i(\delta_l)=0$ for all $i=1,...,N$, $l=1,...,M$ without changing senders' preferences over strategy profiles or the zero-sumness of the game.

Suppose senders have utility functions $u_1',...,u_M'$ with $u_1'+...+u_M'=0$ for all $\beta$. For $i=1,..,M$ let $\alpha_i : \Delta(\Omega) \rightarrow \mathbb{R}$ be the affine function $\alpha_i(\beta) = - \sum_l \beta_l u_i(\delta_l)$. For each $i$, define the function $u_i : \Delta(\Omega) \rightarrow \mathbb{R}$ as $u_i = u_i' + \alpha_i$. Then $u_i(\delta_l)=0$ for all $i$, $l=1,...,N$. Note that utility function $u_i$ preserves the same preferences over strategy profiles as $u_i'$, as for any strategy profile $(\Gamma_1,...,.\Gamma_M)$, $\mathbb{E}_{p_1,...,p_M} [u_i(\beta(\Gamma_1,...,\Gamma_M))] = \mathbb{E}_{p_1,...,p_M} [u_i'(\beta(\Gamma_1,...,\Gamma_M))] - \sum_l \pi_l u_i(\delta_l)$ - the latter term is a constant. Finally note that $\alpha_1(\beta)+...+\alpha_M(\beta)=0$ for all $\beta \in \Delta(\Omega)$, so $u_1+....+u_M=0$.

\subsection{Section 3}

\textbf{Lemma 1. General case:} In any equilibrium $(\Gamma_1,...,\Gamma_M)$: (1) $U_i(\Gamma_1,...,\Gamma_M)=0$ for $i=1,...,M$ and (2) $W_i(x) \leq 0$ for all $x \in \Delta(\Omega)$ and $i=1,...,M$.

\begin{proof}
We prove (1) first. First note that as the functions $\{u_i\}_{i=1,...,M}$ are zero-sum, so are $\{U_i\}_i$. To see this, fix any $(\Gamma_1',...,\Gamma_M')$ and let $\Gamma'$ be the random variable representing the receiver's posterior after viewing all $M$ experiment realizations and $p'$ be its pmf. Then $\sum_{i=1}^M U_i(\Gamma_1',...,\Gamma_M') = \sum_i \sum_{\beta \in supp[\Gamma']} u_i(\beta) p'(\beta) = \sum_{\beta \in supp[\Gamma']} p'(\beta) \sum_i u_i(\beta)  = 0$. Next note that any sender $i$ choosing $\Gamma_i=\Gamma^{FR}$ yields $U_i(\Gamma^{FR},\Gamma_{-i})=0$ for all $\Gamma_{-i}$. Hence in any equilibrium $(\Gamma_1,...,\Gamma_M)$, each sender gets $U_i(\Gamma_1,...,\Gamma_M) \geq 0$. Finally no sender can have $U_i(\Gamma_1,...,\Gamma_M)>0$ as this would imply $U_j(\Gamma_1,...,\Gamma_M)<0$ for some $j \neq i$. 

For (2) we prove the contrapositive. Fix any sender $i$ and opponents' strategy profile $\Gamma_{-i}$ such that $W_i(x) >0 $ for some $x \in \Delta(\Omega)$; we will show $\Gamma_{-i}$ cannot be played in equilibrium. Consider the strategy $\Gamma_i'$ with distribution $p_i'$ and support only on $x$ and $\{\delta_l\}_{l=1,...,N}$. Set $p_i'(x) > 0$ small enough such that $\pi_l - x_l p_i'(x)>0$ for all $l$ (such a value exists as $\pi_l >0$ for all $l$). Bayes-plausibility implies we must have: $p_i'(\delta_l) = \pi_l - x_l p_i'(x) > 0 $ for all states $l$ (as the support of $\Gamma_i'$ is $\{x,\delta_1,...,\delta_M\}$). Then $U_i(\Gamma_i',\Gamma_{-i}) = W_i(x) p_i'(x) + \sum_l u_i(\delta_l) p_i'(\delta_l) >0$. Property (1) of the lemma implies $\Gamma_{-i}$ cannot be played in equilibrium; hence $W_i(x)\leq 0$ for all $i$ in any equilibrium.
\end{proof}

\textbf{Proposition 1.} 
Proposition 1 is implied by Proposition 2, proven in the next section. However, as the proving Proposition 1 is much simpler than Proposition 2, we provide a proof here for exposition.

\begin{proof}
The `only if' direction is trivial. If $u_1$ is linear so is $u_2$. Under our normalization of $u_1(0)=u_1(1)=0$, this implies that $u_1(\beta)=u_2(\beta)=0$ for all $\beta \in [0,1]$. Hence both senders are indifferent across all strategy profiles and any $(\Gamma_1,\Gamma_2)$ is an equilibrium.

Now for the `if' direction. Suppose $u_1$ (and hence $u_2$) are nonlinear. Let $q = \sup \{\beta \in [0,1]: u_1(\beta) \neq 0\}$ be the supremum of posteriors at which $u_1,u_2$ are nonlinear. We prove the result in two cases.

\emph{Case 1: $q=1$.} If $q=1$, then by the piecewise analycity of $u_1,u_2$, there exists $r<1$ such that either $u_1(\beta)>0$ or $u_1(\beta)<0$ for all $\beta \in [r,1)$. If $u_1(\beta)<0$ then $u_2(\beta)>0$, and so WLOG (we can always relabel senders) we assume $u_1(\beta)>0$ for all $\beta \in [r,1)$.  Suppose for contradiction that sender 2 plays a non-fully revealing strategy $\Gamma_2$ in some equilibrium. As $\Gamma_2 \neq \Gamma^{FR}$, $\Pr(0<\Gamma_2<1)>0$; let $\underbar{y} = \min supp[\Gamma_2]\setminus \{0,1\} \in (0,1)$ be in the smallest interior belief in the support of $\Gamma_2$. Using the definition of $\beta(x,y)$, define $\underbar{x}$ by $\beta(\underbar{x},\underbar{y}) = r$. Conditional on $\Gamma_1=x \in [\underbar{x},1)$, $\beta(x,y) \in [r,1)$ for all interior $y$ in $\Gamma_2$'s support. But then for all $x \in [\underbar{x},1)$ we have:

\begin{align*}
& W_1(x) = \underbrace{u_1(\beta(x,0))}_{=0} Pr(\Gamma_2=0|\Gamma_1=x) + \underbrace{u_1(\beta(x,1))}_{=0} Pr(\Gamma_2=1|\Gamma_1=x) + \\
& \sum_{y \in supp[\Gamma_2] \setminus \{0,1\}} \underbrace{u_1(\beta(x,y))}_{>0} \underbrace{Pr(\Gamma_2=y|\Gamma_1=x)}_{>0} > 0.\\
\end{align*}

This contradicts Lemma 1 property (2) and hence $\Gamma_2=\Gamma^{FR}$ in all equilibria.

\emph{Case 2: $q<1$.} We break this case into two subcases.

First suppose $u_1(q) \neq 0$. WLOG assume $u_1(q)>0$ (if not then $u_2(q)>0$). Suppose for contradiction $\Gamma_2 \neq \Gamma^{FR}$ in some equilibrium. Then let $r=q$ and define $\underbar{y},\underbar{x}$ as before. Again Lemma 1 prpoerty (2) is violated as:

\begin{align*}
& W_1(\underbar{x}) = \underbrace{u_1(\beta(x,0))}_{=0} Pr(\Gamma_2=0|\Gamma_1=\underbar{x}) + \underbrace{u_1(\beta(x,1))}_{=0} Pr(\Gamma_2=1|\Gamma_1=\underbar{x}) + \\
& \sum_{y \in supp[\Gamma_2] \setminus \{0,1,\underbar{y}\}} \underbrace{u_1(\beta(\underbar{x},y))}_{=0} \underbrace{Pr(\Gamma_2=y|\Gamma_1=\underbar{x})}_{>0} > 0 + \underbrace{u_1(\beta(\underbar{x},\underbar{y}))}_{=u_1(r)>0} \underbrace{Pr(\Gamma_2=\underbar{y}|\Gamma_1=\underbar{x})}_{>0} > 0.\\
\end{align*}

Now suppose $u_1(q)=u_2(q)=0$. Then by piecewise analycity of utilities either $u_1(q^-)>0$ or $u_2(q^-)>0$. WLOG assume $u_1(q^-)>0$ and suppose for contradiction $\Gamma_2 \neq \Gamma^{FR}$ in some equilibrium. Define $\underbar{y}$ as before. There exists $r<q$ and $\underbar{x}$ such that $u_1>0$ on interval $[r,q)$ and $\beta(\underbar{x},\underbar{y})=r$. Then we have $W_1(\underbar{x})>0$, violating Lemma 1 property (2).

\end{proof}

\subsection{Proof of Proposition 2}

\subsubsection{`Only if' direction.}

Suppose for some $\Omega' \subseteq \Omega$ all senders have linear utilities on $\Delta(\Omega')$. Let $x' \in \Delta(\Omega')$ with $x_k' = \frac{\pi_k}{\sum_{n \in \Omega'} \pi_n}$ $\forall k \in \Omega'$. Consider the experiment $\Gamma'$ with $Pr(\Gamma'=\delta_l) = \pi_l$ for all $l \not \in \Omega'$ and $Pr(\Gamma'=x') = \sum_{n \in \Omega'} \pi_n$. $\Gamma'$ is Bayes-plausible and has finite support. The strategy profile $(\Gamma',...,\Gamma')$ is an non-fully revealing equilibrium. To see this consider a sender $i$'s incentive to deviate. If $\omega \in \Omega'$ then $\Gamma_{-i} \in \Delta(\Omega') \implies \beta(\Gamma_1,...,\Gamma_M) \in \Delta(\Omega')$ w.p. $1$ (Claim 2); as $u_i$ is linear on $\Delta(\Omega')$, $i$ has no profitable deviation conditional on $\omega \in \Omega'$. Conditional on $\omega \not \in \Omega'$, $\Gamma_{-i}$ fully reveals the state and no deviation from $i$ can change this. 

\subsubsection{`If' direction.}

Let $\Delta^{int}(\Omega') = \{\gamma \in \Delta(\Omega'): \gamma_l > 0 \text{ } \forall l \in \Omega'\}$; this is the set of beliefs in $\Delta(\Omega')$ whose support is $\Omega'$.

We first prove the result for the case of $|\Omega'|=2$. This case is simpler than the case of $|\Omega'|>2$ and is of particular interest because Theorem 1 only relies on Proposition 2 with $|\Omega'|=2$. 

\textbf{Proof for $|\Omega'|=2$.}

\begin{proof}
WLOG let $\Omega' = \{1,2\}$. Suppose some sender $i$ has $u_i$ nonlinear on $\Delta(\{1,2\})$. For each sender $j'$ let $r'_{j'} = \sup \{t \in [0,1] :  u_{j'}(t \delta_2 + (1-t) \delta_1)>0 \}$. Let $r' = \max_{j'=1,...,M} r'_{j'}$ and $j \in$ argmax$_{j'=1,...,M} r'_{j'}$.  As $u_i$ is nonlinear there exists $\gamma \in \Delta^{int}(\{1,2\})$ with $u_i(\gamma) \neq 0$. If $u_i(\gamma)<0$ then $u_{i'}(\gamma)>0$ for some sender $i'$ (zero-sumness); otherwise $u_i(\gamma)>0$. Regardless, we have that $r'$ exists and is $>0$. 

WLOG let $j = 1$. We prove the `if' direction in 2 cases.

\emph{Case 1: $r'=1$.} $u_1$ is piecewise real analytic and so $\Delta(\{1,2\})$ can be partitioned into intervals each of which $u_1$ is real analytic on. Each $\gamma \in \Delta(\{1,2\})$ can be represented by scalar $\gamma_2$ \textemdash how close it is to $\delta_2$. For some $a \in [0,1)$, $u_1$ is real analytic on an interval $\{\gamma \in \Delta(\{1,2\}): \gamma_2 \in (a,1)\}$ (this $a \in [0,1)$ is not unique; any selection will do). This implies that there are a finite number of points (possibly zero) on $\{\gamma \in \Delta(\{1,2\}): \gamma_2 \in (a,1)\}$ at which $u_1 = 0$. As $r'=1$, there exists $r \in (a,1)$ such that $u_1(\gamma)>0$ for all $\gamma \in \Delta(\{1,2\})$ s.t. $\gamma_2 \in [r,1)$ (again, this $r$ will not be unique; any selection will do).  

Suppose, for contradiction, that some equilibrium $(\Gamma_1,...,\Gamma_M)$ pools $\{1,2\}$. Then we must must have $Pr(\Gamma_{-1}=y \text{ s.t. } y_1,y_2>0)>0$, i.e. $\Gamma_{-1}$ pools $\{1,2\}$, by Claim \ref{not_pool}.

Let $Z = \{y \in supp[\Gamma_{-1}] : y_1,y_2>0\}$; $Z$ is nonempty. For any $x \in \Delta^{int}(\{1,2\})$ and $y  \in Z$, we have $p_{-1}(y|x)>0$ and $\beta(x,y) \in \Delta^{int}(\{1,2\})$ (by Claim \ref{claim_projection} and equation \ref{beta_def}). Note that as $x \rightarrow \delta_1$ we have $\beta(x,y) \rightarrow \delta_1 \implies \beta_2(x,y) \rightarrow 0$ and as $x \rightarrow \delta_2$ we have $\beta(x,y) \rightarrow \delta_2 \implies \beta_2(x,y) \rightarrow 1$. As $Z$ is finite, this implies $\min_{y \in Z} \beta_2(x,y)$ goes to $0$ as $x \rightarrow \delta_1$ and $\min_{y \in Z} \beta_2(x,y)$ goes to $1$ as $x \rightarrow \delta_2$. As for all $y \in Z$ $\beta_2(x,y)$ is continuous in $x$ for $x \in \Delta^{int}(\{1,2\})$, $\min_{y \in Z} \beta_2(x,y)$ is also continuous in $x$ for $x \in \Delta^{int}(\{1,2\})$. By the intermediate value theorem there exists $\underbar{x} \in \Delta^{int}(\{1,2\})$ such that $\min_{y \in Z} \beta_2(\underbar{x},y) = r$.

Note that by equation \ref{beta_def}, $\beta(\underbar{x},y) =\delta_1$ for all $y \in supp[\Gamma_{-1}]$ with $y_1>0$ and $y_2=0$; similarly $\beta(\underbar{x},y) =\delta_2$ for all $y \in supp[\Gamma_{-1}]$ with $y_2>0$ and $y_1=0$. Finally by Claim \ref{claim_p_S'}, $p_{-1}(y|\underbar{x}) = 0$ for all $y \in supp[\Gamma_{-1}]$ with $y_1=y_2=0$. 

Putting this together:

\begin{align*}
& W_1(\underbar{x}) = \sum_{\substack{y \in supp[\Gamma_{-1}] \\ y_1=0,y_2>0}} \underbrace{u_1(\delta_2)}_{=0} p_{-1}(y|\underbar{x}) + \sum_{\substack{y \in supp[\Gamma_{-1}] \\ y_2=0,y_1>0}} \underbrace{u_1(\delta_1)}_{=0} p_{-1}(y|\underbar{x}) \\
& + \sum_{y \in Z} \underbrace{u_1(\beta(\underbar{x},y))}_{>0} \underbrace{p_{-1}(y|\underbar{x})}_{>0} > 0\\
\end{align*}

This contradicts Lemma 1 property (2). Hence no equilibrium can pool $\{1,2\}$.

\emph{Case 2: $r'<1$.} First, if $u_1(r')>0$, then set $r=r'$ and derive $\underbar{x}$ just as in Case 1. As in Case 1, we have $W_1(\underbar{x})>0$, violating Lemma 1 property (2). Next if $u_1(r')<0$, then some sender $i \neq 1$ must have $u_i(r')>0$ (zero-sumness); we can relabel sender $i$ to $1$ and repeat the same argument.

Next assume $u_1(r')=0$. Now for some $a \in [0,r')$, $u_1$ is real analytic on an interval $\{\gamma \in \Delta(\{1,2\}): \gamma_2 \in (a,r')\}$ (this $a \in [0,r')$ is not unique; any selection will do). This implies that there are a finite number of points (possibly zero) on $\{\gamma \in \Delta(\{1,2\}): \gamma_2 \in (a,r')\}$ at which $u_1 = 0$. This implies that there exists $r \in (a,r')$ such that $u_1(\gamma)>0$ for all $\gamma \in \Delta(\{1,2\})$ with $\gamma_2 \in [r,r')$ (again, this $r$ will not be unique; any selection will do). Note that $u_1(\gamma) \geq 0$ for all $\gamma \in \Delta(\{1,2\})$ with $\gamma_2 \geq r$. 

Suppose, for contradiction, that in some equilibrium $(\Gamma_1,...,\Gamma_M)$ pools $\{1,2\}$. We follow identical steps in defining $Z$ and $\underbar{x}$. Note that for all $y \in Z$, $\beta_k(\underbar{x},y) \in [r,1)$ $\implies$ $u_1(\beta(\underbar{x},y))\geq 0$. By the definition of $\underbar{x}$, there exists $\underbar{y} \in Z$ such that $\beta(\underbar{x},\underbar{y}) = r$ $\implies$ $u_1(\beta(\underbar{x},\underbar{y}))>0$. For all $y \in supp[\Gamma_{-1}] \setminus Z$ either $\beta(\underbar{x},y) \in \{\delta_1,\delta_2\}$ or $p_{-1}(y|\underbar{x})>0$. Hence $W_1(\underbar{x})>0$, violating Lemma 1 property (2).
\end{proof}

Now we proceed with the analysis for $|\Omega'| \geq 2$.

\textbf{General analysis.}

Suppose some $u_j$ is nonlinear on $\Delta(\Omega')$. Fix a strategy profile $(\Gamma_1,...,\Gamma_M)$ that pools $\Omega'$; let $\Gamma$ be the experiment induced by observing the realizations of all $M$ experiments and $p$ be its probability mass function. We show, via violation of Lemma 1 property (1), that $(\Gamma_1,...,\Gamma_M)$ is not an equilibrium. To do this it sufficies to identify a sender $j$ and an interim belief $\underbar{x}$ such that when $\Gamma_1,...,\Gamma_M$ are played, conditional on generating interim belief $\underbar{x}$ (from an experiment played additionally and conditionally independently to $\Gamma_j$) $j$ gets a strictly positive expected payoff: : $\mathbb{E}[u_i(\beta(\underbar{x},\Gamma))]>0$. As in the proof of Lemma 1 property (2), $j$ can then construct an experiment $\Gamma_j'$ with support on $\{\delta_1,...,\delta_N,\underbar{x}\}$ and obtain a strictly positive payoff by playing $\Gamma_j'$ in addition to, conditionally independently, $\Gamma_j$. The remainder of the proof shows that such a sender $j$ and $\underbar{x}$ exist.

For $i=1,...,M$ let $A_i = \{\gamma \in \Delta(\Omega): u_i(\gamma) > 0\}$ and $D_i = \{\gamma \in \Delta(\Omega): u_i(\gamma)<0\}$ be the sets of posteriors at which $i$ has an advantage and disadvantage respectively. Let $A = \cup_{i} A_i$ be the union of these advantage sets (also equal to the union of disadvantage sets as utilities are zero-sum) and $cl(A)$ be its closure. 

We say a subset of states $\Theta \subseteq \Omega$ ($|\Theta|>1$) is \emph{minimal} if $A \cap \Theta \neq \emptyset$ and $\Theta' \cap A = \emptyset$ for all $\Theta' \subset \Theta$. Note that if $A$ is empty, then there are no minimal subsets. Meanwhile if $A$ is nonempty, any subset of states that intersects $A$ (i.e. any set $\Theta$ for which some $u_i$ is nonlinear on $\Delta(\Theta)$) is either minimal or has a minimal subset:

\begin{claim}  \label{claim_minimal}
Every subset $\Theta \subseteq \Omega$ for which  $u_i$ (for some $i$) is nonlinear on $\Delta(\Theta)$ is either minimal or has a subset $\Theta'$ that is minimal.
\end{claim}

\begin{proof}
If for some $i$ $u_i$ is nonlinear on $\Delta(\Theta)$, then $A \cap \Delta(\Theta) \neq \emptyset$. Either $\Theta$ is minimal, or there exists a subset $\Theta' \subset \Theta$ that intersects $A$. Now set $\Theta = \Theta'$ and repeat this process until $\Theta$ is minimal; it must be minimal at some point because $\Omega$ is finite and states are removed from $\Theta$ each iteration. 
\end{proof}

By Claim \ref{claim_minimal}, it is sufficient to prove Proposition 2 for minimal $\Omega'$ alone. If all minimal sets cannot be pooled in equilbirium, then any set on which there are nonlinear sender preferences cannot be pooled, as all such sets have a minimal subset. \emph{Henceforth we assume $\Omega'$ is minimal.}

Let $|\Omega'|=K \leq N$ and WLOG let $\Omega' = \{1,...,K\}$. 

It is convenient for us to represent any belief $\gamma \in \Delta(\Omega')$ by the ratios $(r_1(\gamma),...,r_{K-1}(\gamma)) \in (\mathbb{R_+} \cup \{\infty\})^{K-1}$, where for $k=1,...,K-1$: (1) $r_k(\gamma) = \frac{\gamma_k}{1 - \sum_{l=1}^k \gamma_l}$ when $1 - \sum_{l=1}^k \gamma_l$ is nonzero, (2) when this doesn't hold $r_k(\gamma)=\infty$ if $\gamma_k>0$ and $r_k(\gamma)=0$ if $\gamma_k=0$. We call this the ratio representation of $\gamma$. The ratio $r_k(\gamma)$ tells us the ratio of probability mass assigned to state $k$ by $\gamma$ to the mass assigned to states $k+1,...,K$.

\begin{lemma} \label{prop2lemma1}
Note for any $\gamma, \gamma' \in \Delta(\Omega')$ we have $r_k(\gamma) = r_k(\gamma')$ for all $k=1,...,K-1$ if and only if $\gamma = \gamma'$; that is, ratio representations for beliefs in $\Delta(\Omega')$ are unique.
\end{lemma}

\begin{proof}
The `if' direction is trivial; we prove the `only if' direction as follows. First suppose $r_k(\gamma) < \infty$ for all $k=1,...,K-1$. This implies that $1 - \sum_{l=1}^k \gamma_l$ is nonzero for all $k$ (or else, let $k'$ be the minimum $k$ for which  $1 - \sum_{l=1}^k \gamma_l = 0$; but then we must have $\gamma_{k'}>0 \implies r_{k'}(\gamma) = \infty$ \textemdash contradiction). But then from its definition, $r_1$ uniquely pins down $\gamma_1$ ($\gamma_1 = \frac{r_1(\gamma)}{1+ r_1(\gamma)}$), after which $r_2$ pins down $\gamma_2$, ..., $r_{K-1}$ pins down $\gamma_{K-1}$, and $\gamma_K$ is pinned down by $1 = \sum_{l=1}^K \gamma_l$). Now suppose $r_{k''}(\gamma) = \infty$ for some $k''$. Note that this implies $\gamma_{k} = 0 $ for all $k>k''$; further, $1 - \sum_{l=1}^k \gamma_l > 0$ for all $k<k''$ and hence $r_{k}(\gamma)< \infty$ for all $k<k''$. Then $\gamma_1,...,\gamma_{k'' - 1}$ are uniquely pinned down by using the definitions of $r_1(\gamma),...,r_{k''-1}(\gamma)$ (just as in the previous case). $\gamma_{k''}$ is pinned down by $1 = \sum_{l=1}^K \gamma_l$.
\end{proof}

The the continuity of $r_k(\gamma)$ on part of $\Delta(\Omega')$ will be useful later:

\begin{claim} \label{claim_rcont}
For $k=1...,K-1$, $r_k(\gamma)$ is continuous in $\gamma$ for $\gamma \in \Delta(\{k,...,K\}) \setminus \{\delta_k\}$.
\end{claim} 

\begin{proof}
$r_k(\gamma) = \frac{\gamma_k}{1- \sum_{l=1}^k} \gamma_k$. As $\gamma \in \Delta(\{k,...,K\})$, the denominator is strictly positive when $\gamma_k<1$ and so $r_k(\gamma)$ is continuous in $\gamma$ on this domain. 
\end{proof}

The following simple results will be useful.

\begin{lemma} \label{lemma_operation}
Suppose $K>2$. For any $1<L<K$, let $x \in \Delta(\{L,...,K\})$, $x' \in \Delta(\{1,...,L-1\})$ and $y \in \Delta(\Omega)$. If $\beta(x,y)$ is well defined,\footnote{From its definition, one can see $\beta(x,y)$ is only not well defined when $y$ assigns probability $0$ to every state that $x$ assigns strictly positive probability to.} then $r_k(\beta(\lambda x' + (1-\lambda)x,y)) = r_k(\beta(x,y))$ for all $k=L,...,K-1$ and $\lambda \in [0,1)$.
\end{lemma}

\begin{proof} For any $n \in \Omega$,

\begin{align*}
& \beta_n(\lambda x' + (1-\lambda)x,y) = \frac{\frac{(\lambda x_n' + (1-\lambda)x_n)y_n}{\pi_n}}{\sum_{n' =1}^N \frac{(\lambda x_{n'}' + (1-\lambda)x_{n'})y_{n'}}{\pi_{n'}}}
\end{align*} 

For $k \geq L$:

\begin{align*}
&r_k(\beta((\lambda x' + (1-\lambda)x),y)) = \frac{\beta_k((\lambda x' + (1-\lambda)x),y)}{\sum_{n=k+1}^N \beta_n((\lambda x' + (1-\lambda)x),y)} = \frac{(\lambda x_{k}' + (1-\lambda)x_{k})y_k/\pi_k}{\sum_{n=k+1}^N (\lambda x_{n}' + (1-\lambda)x_{n})y_n/\pi_n} \\
& = \frac{(1-\lambda)x_k y_k /\pi_k}{(1-\lambda)\sum_{n=k+1}^N x_n y_n /\pi_n} =  \frac{x_k y_k /\pi_k}{\sum_{n=k+1}^N x_n y_n /\pi_n}
\end{align*}

whenever the denominator is nonzero; when the denominator is nonzero, this expression is equal to $r_k(\beta(x,y))$. When the denominator and numerator are zero, $r_k(\beta(\lambda x' + (1-\lambda)x,y)) = r_k(\beta(x,y))=0$ and when the denominator is zero and the numerator is nonzero, $r_k(\beta(\lambda x' + (1-\lambda)x,y)) = r_k(\beta(x,y))=\infty$.

\end{proof}

\begin{claim} \label{claim_convex_comb}
Suppose $K>2$. For any $1<L<K$, let $x \in \Delta(\{L,...,K\})$, $x' \in \Delta(\{1,...,L-1\})$ and $y \in \Delta(\Omega)$. If $\beta(x,y)$ and $\beta(x',y)$ are well defined then $\beta_k(\lambda x' + (1-\lambda)x,y) = \lambda \beta_k(x',y) + (1-\lambda)\beta_k(x,y)$ for all $\lambda \in [0,1]$, $k=1,...,N$.
\end{claim}

\begin{proof}
Simple algebra.
\end{proof}

Let $Z = \{z \in supp[\Gamma]  \text{ : } z_n > 0\ \text{ for all } n \in \Omega'\}$. Note $Z$ is nonempty as $(\Gamma_1,...,\Gamma_M)$ pools $\Omega'$. For $k=1,...,K-1$ and $x \in \Delta^{int}(\Omega')$ define $M_x(k)$ recursively starting with $k=K-1$:

\begin{equation} \label{M_N-1}
M_x(K-1) = \text{argmin}_{z \in Z} r_{K-1}(\beta(x,z))
\end{equation} 

$M_x(K-1)$ is nonempty as $Z$ is. For $k<K-1$, let $M_x(k) =$ argmin$_{z \in M_x(k+1)} r_k(\beta(x,z))$; these sets are nonempty for all $k$. For $k=1,...,K-1$ define $m_x(k)$ by: pick $z \in M_x(k)$ and let $m_x(k) = r_k(\beta(x,z))$. $m_x(k)$ is well defined for all $k$.

$M_x(K-1)$ gives the set of realizations of $\Gamma$ that, conditional on interim belief $x$ being realized from a different experiment, would induce the lowest $r_{K-1}$ ratio of posteriors among those in $Z$. $m_{x}(K-1)$ gives the value of this lowest $r_{K-1}$ ratio. $M_{x}(K-2)$ gives the subset of $M_{x}(K-1)$ that would result in lowest $r_{K-2}$ ratio of posteriors conditional on $x$ being realized and $m_{x}(K-2)$ gives this value, etc.

Note any $z \in M_{x}(1)$ must satisfy:

\begin{equation} \label{M1_singleton}
r_k(\beta(x,z))=m_{x}(k) \text{ for all } k=1,...,K-1
\end{equation}

As $\beta(x,y) \in \Delta(\Omega')$ ($x \in \Delta^{int}(\Omega')$ and Claim \ref{claim_projection}), by Lemma \ref{prop2lemma1}, ratios $m_{x}(1),...,m_{x}(K-1)$ uniquely pin down the value of $\beta(x,z)$ for all $z \in M_{x}(1)$. If we have $|M_{x}(1)|>1$, this means that multiple realizations of $\Gamma$, $z \neq z'$, produce the same posterior conditional on $x$. This is possible when $x$ assigns probability $0$ to states $z,z'$ do not \textemdash $z$ and $z'$ differing on these states may not affect the posterior. 

Using the objects introduced above, we finish proving Proposition 2 in two cases. In the first case, $cl(A) \cap \Omega'' = \emptyset$ for all $\Omega'' \subsetneq \Omega'$. This case include the example in the main Appendix in the text of paper; the same logic generalizes. The second case to consider is $cl(A) \cap \Omega'' \neq \emptyset$ for some $\Omega'' \subsetneq \Omega'$.

\textbf{Case 1: $cl(A) \cap \Omega'' = \emptyset$ for all $\Omega'' \subsetneq \Omega'$.}

\begin{lemma} \label{prop2_lemma_interior}
Suppose $cl(A) \cap \Delta(\Omega'') = \emptyset$ for all $\Omega'' \subsetneq \Omega'$. Then there exists $x^* \in \Delta(\Omega')$ and $\bar{\beta} \in cl(A)$ such that for all $y \in Z$ either: (1) $\beta(x^*,y) \not \in cl(A)$ or (2) $\beta(x^*,y) = \bar{\beta}$.
\end{lemma}

\begin{proof}
Define the point $\bar{\beta} \in cl(A) \cap \Delta(\Omega')$ as follows. 
Let $E(K-1)=\argmax_{\gamma \in cl(A) \cap \Delta(\Omega')} r_{K-1}(\gamma)$ and $e(K-1)=  \max_{\gamma \in cl(A) \cap \Delta(\Omega')} r_{K-1}(\gamma)$. For $k=1,...,K-2$, let $E(k)=\argmax_{\gamma \in E(k+1)} r_{k}(\gamma)$ and $e(k)=  \max_{\gamma \in E(k+1)} r_{k}(\gamma)$. Note that as $cl(A) \cap \Delta(\Omega'') = \emptyset$ for all $\Omega'' \subsetneq \Omega'$, we have $0<e(k)<\infty$ for all $k=1,...,K-1$. Further, by Lemma \ref{prop2lemma1}, $|E(1)|=1$ as $\gamma \in E(1)$ must satisfy $r_k(\gamma)=e(k)$ for all $k=1,...,K-1$. Let $\bar{\beta}$ be the unique element in $E(1)$.

Consider $x \in \Delta^{int}(\{K-1,K\})$. Note that as $x \rightarrow \delta_{K}$, $\beta(x,y) \rightarrow \delta_K \implies r_{K-1}(\beta(x,y)) \rightarrow 0$ for all $y \in Z$. Similarly as $x \rightarrow \delta_{K-1}$, $\beta(x,y) \rightarrow \delta_{K-1} \implies r_{K-1}(\beta(x,y)) \rightarrow \infty$ for all $y \in Z$. By the finiteness of $Z$, $x \rightarrow \delta_{K} \implies \min_{y \in Z} r_{K-1}(\beta(x,y)) \rightarrow 0$ and $x \rightarrow \delta_{K-1} \implies \min_{y \in Z} r_{K-1}(\beta(x,y)) \rightarrow \infty$. The continiuity of $\beta(x,y)$ in $x$, continuity of $r_{K-1}$ in $\beta(x,y)$ (Claim \ref{claim_rcont}), and finiteness of $Z$ together imply the continuity of $\min_{y \in Z} r_{K-1}(\beta(x,y))$ in $x$. By the intermediate value theorem, there exists $x' \in \Delta^{int}(\{K-1,K\})$ with $\min_{y \in Z} r_{K-1}(\beta(x',y)) = e(K-1)$, or $m_{x'}(K-1)=e(K-1)$.

We prove the result inductively, with the previous paragraph being the base case. Suppose we have found $x' \in \Delta(\{k'+1,...,K\})$ such that for all $k=k'+1,...,K-1$, $m_{x'}(k) = e(k)$. We find $x'' \in \Delta(\{k',...,K\})$ with $m_{x''}(k) = e(k)$ for all $k=k',...,K-1$. Consider $x'(\lambda) = \lambda \delta_{k'} + (1-\lambda) x'$ for $\lambda \in (0,1)$. As $\lambda \rightarrow 1$, $\beta(x'(\lambda),y) \rightarrow \delta_{k'} \implies r_{k'}(\beta(x'(\lambda),y)) \rightarrow \infty$ for all $y \in Z$ and as $\lambda \rightarrow 0$, $\beta(x'(\lambda),y) \rightarrow x' \implies r_{k'}(\beta(x'(\lambda),y)) \rightarrow 0$ for all $y \in Z$. For all $\lambda \in [0,1)$, $y \in Z$, $k=k'+1,...,K-1$, $r_{k}(\beta(x'(\lambda),y)) = r_k(\beta(x',y))$ by Lemma \ref{lemma_operation}; hence changing $\lambda$ will leave $m_{x'}(k)=e(k)$ for $k=k'+1,...,K-1$. By finiteness of $M_{x'}(k'+1)$, continuity of $\beta(x'(\lambda),y)$ for all $y \in M_{x'}(k'+1)$, continuity of $r_{k'}$ in $\beta(x'(\lambda),y)$ for all $y \in M_{x'}(k'+1)$, and in the intermediate value theorem, there exists $\lambda^* \in (0,1)$ and $x'' = \lambda^* \delta_{k'} + (1-\lambda^*)x' \in \Delta(\{k',...,K\})$ such that $m_{x''}(k') = e(k')$. Then $m_{x''}(k) = e(k)$ for all $k=k',...,K-1$.

Carrying this inductive process through until $k'=1$, by equation \ref{M1_singleton} we find $x^* \in \Delta(\Omega')$ with, for all $y \in M_{x^*}(1)$ and $k=1,...,K-1$: $r_k(\beta(x^*,y)) = e(k)$. Hence for all $y \in M_{x^*}(1)$, $\beta(x^*,y) = \bar{\beta}$. Meanwhile for all $y \not \in M_{x^*}(1)$, there exists $1 \leq k' \leq K-1$ such that $r_{k}(\beta(x^*,y))=e(k)$ for all $k>k'$ and $r_{k'}(\beta(x^*),y)>e(k')$; this implies (by definition of $e(k')$) that $\beta(x^*,y) \not \in cl(A)$.

\end{proof}

The following result follows from Lemma \ref{prop2_lemma_interior} almost immediately.

\begin{lemma} \label{prop2_lemma_interior2}
Suppose $cl(A) \cap \Delta(\Omega'') = \emptyset$ for all $\Omega'' \subsetneq \Omega'$. Then there exists $\underbar{x} \in \Delta(\Omega')$ and $\beta' \in A$ such that for all $y \in Z$ either: (1) $\beta(\underbar{x},y) \not \in cl(A)$ or (2) $\beta(\underbar{x},y) = \beta'$.
\end{lemma}

\begin{proof}
Find $x^*$ and $\bar{\beta}$ as per Lemma \ref{prop2_lemma_interior}. If $\bar{\beta} \in A$ then set $\underbar{x}=x^*$ and $\beta' = \bar{\beta}$ and we're done.

Assume now $\bar{\beta} \not \in A$. As $Z$ finite, $cl(A)$ is closed, and $\beta(x,y)$ is continuous in $x$ for all $y \in Z$, there exists $\epsilon>0$ and a set $N_\epsilon=\{x \in \Delta(\Omega'): |x-x^*|<\epsilon\}$ such that for all $x \in N_\epsilon$ and $y \in Z$, $\beta(x^*,y) \not \in cl(A) \implies \beta(x,y) \not \in cl(A)$. 

By equation \ref{beta_def}, for any $\gamma \in \Delta(\Omega')$ and $y \in Z$ there exists $x \in \Delta(\Omega')$ such that $\beta(x,y) = \gamma$. Also by equation \ref{beta_def}, if $\beta(x,y)=\beta(x,y')$ for some $x \in \Delta^{int}(\Omega')$ and $y,y' \in Z$, then $\beta(x',y)= \beta(x',y')$ for all $x' \in \Delta^{int}(\Omega')$.

As $\bar{\beta} \in cl(A)$, there is a sequence of beliefs in $A$ converging to $\bar{\beta}$. By continuity of $\beta(x,y)$ in $x$ for all $y \in Z$ and the facts in the previous paragraph, there exists $\beta' \in A$ close to $\bar{\beta}$ and $\underbar{x} \in N_{\epsilon}$ such that for all $y \in Z$, $\beta(x^*,y) = \bar{\beta} \implies \beta(\underbar{x},y) =\beta'$. 

Hence we have either $\beta(\underbar{x},y) = \beta' \in A$ or $\beta(\underbar{x},y) \not \in cl(A)$ for all $y \in Z$.
 
\end{proof}

This next Lemma wraps up the proof of the `if' direction of Proposition 2 for the case that $cl(A) \cap \Delta(\Omega'') = \emptyset$ for all $\Omega'' \subsetneq \Omega'$:

\begin{lemma} \label{prop2_interior_case}
Suppose $cl(A) \cap \Delta(\Omega'') = \emptyset$ for all $\Omega'' \subsetneq \Omega'$. Then there exists a sender $j$ and $\underbar{x} \in \Delta^{int}(\Omega')$ with $\mathbb{E}[u_j(\beta(\underbar{x},\Gamma))]>0$.
\end{lemma}

\begin{proof}
Find $\beta'$ and $\underbar{x}$ as per Lemma \ref{prop2_lemma_interior2}.

Note that for all $y \in supp[\Gamma] \setminus Z$, one of the following holds. (1) $y_l=0$ for all $l \in \Omega'$. (2) $y_l =0$ for some $l \in \Omega'$ but $y_k>0$ for some $k \in \Omega'$. For $y$ in case (1), by Claim \ref{claim_p_S'} $p(y|\underbar{x}) = 0$. For $y$ in case (2), equation \ref{beta_def} implies $\beta(\underbar{x},y) \in \Omega''$ for some $\Omega'' \subsetneq \Omega'$; as $\Omega'$ is minimal we have $u_i(\beta(\underbar{x},y))=0$ for all $i$.

For all $y \in Z$ such that $\beta(\underbar{x},y) \neq \beta'$, $u_i(\beta(\underbar{x},y)) = 0 $ for all $i$.

Let $j$ be some sender with $u_j(\beta')>0$. Then $\mathbb{E}[u_j(\beta(\underbar{x},\Gamma))] = u_j(\beta')Pr(\Gamma \in \{y \in Z: \beta(\underbar{x},y)=\beta'\}|\underbar{x}) > 0$.

\end{proof}

\textbf{Case 2: $cl(A) \cap \Omega'' \neq \emptyset$ for some $\Omega'' \subsetneq \Omega'$.}

Note that it is still the case that $\Omega'$ is minimal, so $A \cap \Omega'' = \emptyset$ for all $\Omega'' \subsetneq \Omega'$.

Let $\Omega'' \in$ argmin$_{\{\Omega''' \subset \Omega': cl(A) \cap \Omega''' \neq \emptyset\}} |\Omega'''|$ be one of the smallest subsets of $\Omega'$ that intersects $cl(A)$. Note that $1 \leq |\Omega''| < |\Omega'|$. Define $Z$, and all other necessary objects, as in the previous case. Then note that $cl(A) \cap \Theta = \emptyset$ for all $\Theta \subsetneq \Omega''$ and so by an identical argument to Lemma \ref{prop2_lemma_interior} we can find $\bar{\beta} \in cl(A) \cap \Delta^{int}(\Omega'')$ and $x^* \in \Delta^{int}(\Omega'')$ such that for all $y \in Z$ either $\beta(x^*,y) = \bar{\beta}$ or $\beta(x^*,y) \not \in cl(A)$. Then by a similar argument to Lemma \ref{prop2_lemma_interior2}, we can find $\underbar{x} \in \Delta^{int}(\Omega')$ close to $x^*$ and $\beta' \in A$ close to $\bar{\beta}$ such that either $\beta(\underbar{x},y)=\beta'$ or $\beta(\underbar{x},y) \not \in cl(A)$ for all $y \in Z$. 

Following the same argument as in Lemma \ref{prop2_interior_case}, all $y \in supp[\Gamma] \setminus Z$ either occur with probability $0$ conditional on $\underbar{x}$ or result in a posterior outside of $\Delta^{int}(\Omega')$ conditional on $\underbar{x}$ (yielding $0$ utility by minimality of $\Omega'$). Letting $j$ be some sender with $u_j(\beta')>0$, we have $\mathbb{E}[u_j(\beta(\underbar{x},\Gamma))]>0$ and we're done.

\subsubsection{Additional claim}

In the proof of Proposition 2 given, when utilities are nonlinear on some $\Delta(\Omega')$ and $(\Gamma_1,...,\Gamma_M)$ pooled $\Omega'$, we were able to find a sender $j$ who could take advantage of $\Omega'$ being pooled and find $\underbar{x}$ conditonal on which she gets strictly positive expected utility. Identifying this sender $j$ did not depend on the strategy profile $(\Gamma_1,...,\Gamma_M)$. Hence, the claim below (which is useful for further results in Supplementary Appendix B) holds.
\begin{claim}\label{strong_claim_prop2}
Suppose for some sender $i$ and $\Omega' \subseteq \Omega$ that $u_i$ is nonlinear on $\Delta(\Omega')$. Then there exists a sender $j$ such that for any $(\Gamma_1,...,\Gamma_M)$ that pools $\Omega'$ $j$ can find some $\underbar{x}$ such that $\mathbb{E}[u_j(\beta(\underbar{x},\Gamma_1,...,\Gamma_M))]>0$.
\end{claim}

\subsection{Proof of Theorem 1}
\textbf{Theorem 1.} 
\begin{proof}
`If' direction. Proposition 2 implies that if for every $l,k$ some $u_i$ is nonlinear on $\Delta(\{l,k\})$, then in any equilibrium $(\Gamma_1,...,\Gamma_M)$, w.p. $1$ $\beta_n(\Gamma_1,...,\Gamma_M)>0$ for only one $n \in \Omega$. Hence w.p. $1$ we must have $\beta_n(\Gamma_1,...,\Gamma_M)=1$ for some $n \in \Omega$; the state is fully revealed.

`Only if' direction. Suppose for some $l,k \in \Omega$, $u_i$ is linear along $\Delta(\{l,k\})$ for all $i$. Then the equilibrium construction in the Proposition 2 `only if' direction is a non-fully revealing equilibrium.

\end{proof}

\subsection{Proofs of Section 4.3}

\textbf{Proposition \ref{prop_comp}.}

Before proving Proposition \ref{prop_comp} we prove the following simple Lemma.

\begin{lemma} \label{lemma_pooling_informativeness}
Suppose an experiment $\Gamma$ not pools $\Omega' \subseteq \Omega$. If $\Gamma'$ is more informative than $\Gamma$ then $\Gamma'$ also not pools $\Omega'$.
\end{lemma}

\begin{proof}
At every posterior induced by $\Gamma$, the receiver has ruled out at least one state in $\Omega'$; hence at any belief induced by $\Gamma$, no further information could induce a posterior which pools $\Omega'$. Thus $\Gamma'$ cannot pool $\Omega'$.
\end{proof}

Now we prove Proposition \ref{prop_comp}.

\begin{proof}
We first prove property (1). Take any equilibrium of the $M$ sender game $(\Gamma_1,...,\Gamma_M)$. If this equilibrium is fully revealing the result is trivial; suppose it is not. We prove the result in two cases.

First suppose $M'=1$; equilibria of the $M'$ sender game are sender 1's Bayesian Persuasion solutions. First suppose that sender 1's payoff from every experiment yields her ex-ante expected utility $\leq 0$. Then the equilibrium of the $M$ sender game, which yields sender 1 $0$ utility (Lemma 1) is a solution to sender 1's Bayesian Persuasion problem, and we're done. Now suppose sender 1 can attain a strictly positive payoff from some experiment (when she is the sole persuader). Then all of sender 1's Bayesian Persuasion solutions yields her strictly positive utility and hence must pool some set of states that cannot be pooled in an equilibrium of the $M$ sender game (any experiment that only pools states that can be pooled in an equilibrium of the $M$ sender game yields $0$ utility for sender 1). By Lemma \ref{lemma_pooling_informativeness}, none of sender 1's Bayesian Persuasion solutions is more informative than $(\Gamma_1,...\Gamma_M)$. 

Now suppose $M'>1$. We again split the proof into two cases. First suppose that for every $\Omega'$ and $\Omega''$ pooled by $(\Gamma_1,...,\Gamma_M)$, whenever $\Omega' \cap \Omega'' \neq \emptyset$, all $u_i$ are linear on $\Delta(\Omega' \cup \Omega'')$. This implies that $\Omega$ can be partitioned into some $\Omega_1,...,\Omega_B$ such that: (1) w.p. $1$ $\beta(\Gamma_1,...,\Gamma_M) \in \Delta(\Omega_b)$ for some $b$, and (2) whenever $|\Omega_b|>1$, all $u_i$ are linear on $\Delta(\Omega_b)$. Consider the strategy profile where all $M'$ senders reveal which partition element $\omega$ is in, but nothing else. This is an $M'$ sender equilibrium (no sender $i$ can profitably deviate as the receiver will learn the partition element from her opponents and conditional on this, $i$ cannot strictly improve her payoff by providing additional information). It is also weakly less informative than $(\Gamma_1,...,\Gamma_M)$, as $(\Gamma_1,...,\Gamma_M)$ reveals which partition element $\omega$ is in and, potentially, more information.

Now suppose there exists $\Omega'$ and $\Omega''$ pooled by $(\Gamma_1,...,\Gamma_M)$ such that $\Omega' \cap \Omega'' \neq \emptyset$ and some $u_i$ is nonlinear on $\Delta(\Omega' \cup \Omega'')$. Consider the strategy profile in the $M'$ sender game $(\Gamma,...,\Gamma)$, where $Pr(\Gamma=\delta_n) = \pi_n$ for all $n \not \in \Omega'$ and $Pr(\Gamma = y \text{ s.t. } y_l =\frac{\pi_l}{\sum_{k \in \Omega'} \pi_k} \text{ } \forall l \in \Omega')= \sum_{k \in \Omega'} \pi_k$. This is an $M'$ sender equilibrium (see proof of the `only if' direction of Proposition 2). We now show it is not strictly more informative than $(\Gamma,...,\Gamma)$. Note that the strategy profile $(\Gamma,...,\Gamma)$ induces the same distribution over posteriors as the experiment $\Gamma$; hence we will deal with $\Gamma$ instead of $(\Gamma,...,\Gamma)$. Let $conv[supp[\Gamma]]$ be the convex hull of the support of $\Gamma$. $\Gamma$ induces $N - |\Omega'|+1$ distinct posteriors: the $N-|\Omega'|$ degenerate beliefs in $\Delta(\Omega \setminus \Omega')$, and a single belief in $\Delta(\Omega')$. These posteriors are all affinely independent (or, alternatively, their convex hull has dimension $N-|\Omega'|$). By Theorem 5 of \citet{wu2017coordinated}, $\Gamma$ is more informative than $(\Gamma_1,...,\Gamma_M)$ if and only if $supp[\beta(\Gamma_1,...,\Gamma_M)] \subseteq conv[supp[\Gamma]]$. Note that $conv[supp[\Gamma]$ only contains beliefs that that satisfy at least one of the following three properties. Either they assign probability one to $\Omega \setminus \Omega'$, assign probability one to $\Omega'$, or assign strictly positive probability to all states in $\Omega$. With positive probability, $(\Gamma_1,...,\Gamma_M)$ results in posteriors which pool $\Omega''$. Such posteriors cannot assign probability one to $\Omega \setminus \Omega'$ (as $\Omega'' \cap \Omega' \neq \emptyset$), cannot assign probability one to $\Omega'$ ($\Omega'' \not \subseteq \Omega'$ as some utilities are nonlinear on $\Delta(\Omega' \cup \Omega'')$), and cannot assign strictly positive probability to all states (some utilities are nonlinear on $\Delta(\Omega' \cup \Omega'')$ $\implies$ $\Omega \supseteq \Omega'' \cup \Omega'$ cannot be pooled in an $M$ sender equilibrium). Hence $(\Gamma,...,\Gamma)$ is no more informative than $(\Gamma_1,...,\Gamma_M)$\emph{•}

Now we prove property (2) of Proposition \ref{prop_comp}. Suppose $M'$ sender equilibrium $(\Gamma'_1,...,\Gamma'_{M'})$ is more informative than $(\Gamma_1,...,\Gamma_M)$. Consider the $M$ sender strategy profile $(\Gamma'_1,...,\Gamma'_{M'},\Gamma^U,...,\Gamma^U)$. This is an $M$ sender equilibrium for the following reasons. No sender $i=1,...,M'$ has an incentivie to deviate as $(\Gamma'_1,...,\Gamma'_{M'})$ is an $M'$ sender equilibrium. No sender $i=M'+1,...M$ has an incentive to deviate as, by Lemma \ref{lemma_pooling_informativeness}, every non-fully revealing posterior of $(\Gamma'_1,...,\Gamma'_{M'},\Gamma^U,...,\Gamma^U)$ falls in some $\Delta(\Omega')$ on which all senders have linear utilities. 
\end{proof}

\subsection{Proof of Proposition \ref{prop_repeat}}

First we prove a result about asymptotic sufficiency.

\begin{lemma}
An experiment $\Gamma$ is asymptotically sufficient if $dim(co(supp[\Gamma]))=N-1$.
\end{lemma}

\begin{proof}
Let $p$ be the probability mass function of $\Gamma$. First we show that if $dim(co(supp[\Gamma]))=N-1$ then $\forall l \neq k \in \Omega$, there exists $x \in supp[\Gamma]$ s.t. $\frac{x_k}{\pi_k} \neq \frac{x_l}{\pi_l}$. To see this, suppose not for contradiction. Then $\exists l,k$ s.t. for all $y \in co(supp[\Gamma])$, $y_k = \frac{\pi_k}{\pi_l} y_l$. But then $co(supp[\Gamma])$ cannot have full dimension if all $y_k$ values are linearly dependent on $y_l$ values; contradiction.

Now we prove the result. Note that for all $x \in supp[\Gamma]$ and $k \in \Omega$, $Pr(\Gamma = x|\omega=k) = \frac{x_k}{\pi_k} p(x)$. Let $\{y^t\}_{t=1}^T$ be a sample of $T$ conditionally independent realizations of $\Gamma$. Conditional on any realized state $\omega$, for any $x \in supp[\Gamma]$, by the Weak Law of Large Numbers, $\frac{\sum_{t=1}^T \mathbbm{1}_{y^t=x}}{T} \xrightarrow{p} \frac{x_{\omega}}{\pi_{\omega}} p(x)$. 

If $dim(co(supp[\Gamma]))=N-1$, then $\forall l \neq k \in \Omega$, there exists $x \in supp[\Gamma]$ s.t. $\frac{x_k}{\pi_k} \neq \frac{x_l}{\pi_l}$ and hence a decision maker will be able to differentiate between events $\omega=l$ and $\omega=k$ using the frequency of occurance of $\Gamma=x$ with arbitrarily high posterior probability as $T \rightarrow \infty$. This is true for all $l,k$ and so $\Gamma^T \xrightarrow[p]{T \rightarrow \infty} \Gamma^{FR}$.
\end{proof}

Now we prove the Proposition \ref{prop_repeat}. Sender and receiver preferences are determined by expected utility over payoffs $\{u_i(a,\omega)\}_{a\in A, \omega \in \Omega}$. Assume the receiver breaks indifferences in favor of higher labelled actions. 

\begin{proof}
For each action $a \in A$, let $\Delta_a \subset \Delta(\Omega)$ be the posteriors at which the receiver takes action $a$. Each $\Delta_a$ is convex and sender's payoffs are linear on $\Delta_a$ (see Appendix B Section B.1). Under our assumption that the receiver is not indifferent between any two actions at any state, each $\delta_l$ is contained in exactly one set $\Delta_a$. 

At belief $\beta \in \Delta_a$, sender $i$ gets payoff $u_i(\beta) = u_a^i \cdot \beta$ where $u_a^i=(u_i(a,1),...,u(a,N))^T$. For any $n \in \mathcal{N}$ and action $a$, define:

\begin{equation} 
p^n(a) = Pr(\Gamma^n \in \Delta_a)
\end{equation}

\begin{equation}
\gamma^n(a) = \mathbb{E}[\Gamma^n | \Gamma^n \in \Delta_a]
\end{equation}

These are the probability action $a$ is induced under $\Gamma^n$ and the expected interim belief generated by $\Gamma^n$ conditional on this. If $p^n(a)=0$, we define $\gamma^n(a)$ to be an arbitrary fixed belief in $\Delta_a$. As all senders have linear payoffs on $\Delta_a$, a sender's utility from any $\Gamma^n$ can be written as:

\begin{equation}
U_i(\Gamma^n) = \sum_{a \in A} u^i_a \cdot (\gamma^n(a) p^n(a))
\end{equation}

We look for Nash equilibria in (finite support) mixed strategies $(\sigma_1,...,\sigma_M)$ where $\sigma_i$ is a probability distribution over $\{\Gamma^n\}_{n=1}^{\infty}$. A strategy profile $\sigma = (\sigma_1,...,\sigma_M)$ induces a probability distribution over the number of repetitions of $\Gamma$ that are collectively played by all senders; we denote the probability that $n$ repetitions are played under strategy profile $\sigma$ as $\sigma(n)$. A senders payoff from deviating from $\sigma$ by playing an additional $n$ repetitions of $\Gamma$ is:

\begin{equation} \label{msne_requirement}
U_i(\sigma,n) = \sum_{n' = 0}^{\infty} \sigma(n') \sum_{a \in A} u^i_a \cdot (\gamma^{n+n'}(a) p^{n+n'}(a)) = 0
\end{equation}

In equilibrium such deviations must yield the equilibrium payoff; if $U_i(\sigma,n)-U_i(\sigma,0)>0$, sender $i$ should deviate to play $n$ extra copies and if $U_i(\sigma,n)-U_i(\sigma,0)<0$, zero-sumness implies some sender can profitably make this deviation.

First we argue that no non-fully revealing pure strategy equilibrium can exist and then extend this to mixed strategies. Suppose for contradiction a pure strategy equilibrium existed which resulted in equilibrium aggregate information $\Gamma^{n'}$. Equation \ref{msne_requirement} implies that it must be that for all $n \geq n'$ and senders $i$:

\begin{equation} \label{psne_requirement}
U_i(\Gamma^{n+1})-U_i(\Gamma^n) = \sum_{a \in A} u^i_a \cdot [\gamma^{n+1}(a) p^{n+1}(a) - \gamma^n(a) p^n(a)] = 0
\end{equation}

Each sender must be indifferent between adding any number of copies to $\Gamma^{n'}$. We consider the limit of this expression as $n \rightarrow \infty$. As $\Gamma$ is asymptotically sufficient, as the number of copies played goes to infinity, information converges to full relevation. As a consequence, if we pick any $a \in A$, if $\Delta_a \cap \{\delta_1,...,\delta_N\} \neq \emptyset$, then: $\lim_{n \rightarrow \infty} p^n(a) = \sum_{l \in \Omega: \delta_l \in \Delta_a} \pi_l$ and $\lim_{n \rightarrow \infty} \gamma^n(a) = \sum_{l \in \Omega: \delta_l \in \Delta_a} \pi_l/p^n(a)$ (as the posterior converges to fully revealing). For $a$ s.t. $\Delta_a \cap \{\delta_1,...,\delta_N\} = \emptyset$, $\lim_{n \rightarrow \infty} p^n(a) = 0$; meanwhile $\gamma^n(a)$ has a convergent subsequence converging to some $\gamma(a) \in cl(\Delta_a)$ as $cl(\Delta_a)$ is closed and bounded. We select a subsequence such that each $\gamma^n(a)$ converges; henceforth we assume each has a limit $\gamma(a)$. This implies that $[\gamma^{n+1}(a) p^{n+1}(a) - \gamma^n(a) p^n(a)] \rightarrow 0$ for all $a \in A$, $i=1,...,M$. 

Pick a sender $i$ and let $t^i_a =  |\gamma^{n+1}(a) p^{n+1}(a) - \gamma^n(a) p^n(a)|$ for $a=1,...,A$. Without loss of generality, assume $t^i_1$ converges to $0$ at least as slowly as $t^i_2,...,t^i_A$. For $a=1,...,A$ let $l_a^i = \lim_{n \rightarrow \infty} \frac{ [\gamma^{n+1}(a) p^{n+1}(a) - \gamma^n(a) p^n(a)]}{t^i_1}$. Each $l_a^i$ is a vector; note that as $t^i_1 \rightarrow 0$ the slowest, each $l^i_a$ exists. Furthermore, $l^i_1$ must have at least 2 nonzero terms; this is because

Note that from equation \ref{psne_requirement}, $\frac{U_i(\Gamma^{n+1})-U_i(\Gamma^n)}{t^i_1}=0$ for all $n \neq n'$. Hence we must have:

\begin{equation}
0 = \lim_{n \rightarrow \infty} \frac{U_i(\Gamma^{n+1})-U_i(\Gamma^n)}{t^i_1}= \sum_{a \in A} u^i_a \cdot l^i_a 
\end{equation}

As some terms in $l_a^i$ are nonzero $u^i_1,...,u^i_A$ must satisfy a nontrivial linear equation. Note the $l_a^i$'s depend purely on the receiver's preferences and $\Gamma$. Generic $\{u_a^i\}_{a=1,...,A}$ will not satisfy this; generically we will have $\lim_{n \rightarrow \infty} \frac{U_i(\Gamma^{n+1})-U_i(\Gamma^n)}{t^i_1} \neq 0$ implying equation \ref{psne_requirement} is not satisfied for large enough $n$. Note that for any $i$ considering deviation, the limits $\{l_a^i\}_{a=1,...,A}$ do not depend on the pure strategy profile considered; hence the limit is the same when deviating from any mixed-strategy profile with support on finite pure strategies.  
\end{proof}

\section{Supplementary Appendix B: Extensions}

\newtheorem{theorem2}{Theorem}[subsection]
\newtheorem{lemma2}{Lemma}[subsection]
\newtheorem{prop2}{Proposition}[subsection]
\counterwithin{figure}{section}

\subsection{Single receiver with finite actions}
\emph{Setup.} There is a single receiver. Let $\mathcal{A} = \{a_1,...,a_A\}$ be a finite set of actions; the game is as in the baseline model except after observing realizations of $\Gamma_1,...,\Gamma_M$ the receiver picks an action $a \in \mathcal{A}$ after which the players get payoffs. The reciever's utility function is $u_r: \mathcal{A} \times \Omega \rightarrow \mathbb{R}$ and senders' utility functions are $u_i: \mathcal{A} \times \Omega \rightarrow \mathbb{R}$ for $i=1,...,M$. Sender's preferences are zero-sum: for all $a \in \mathcal{A}$, $\omega \in \Omega$, $\sum_{i=1}^M u_i(a,\omega)=0$. We make the assumption that for any $a \neq a' \in \mathcal{A}$, $\omega \in \Omega$, $u_r(a,\omega) \neq u_r(a',\omega)$ and $u_i(a,\omega) \neq u_i(a',\omega)$ for all $i \in \{1,...,M\}$; this is generically true. The equilibrium concept is Perfect Bayesian Equilibrium (PBE). 

For any $\beta \in \Delta(\Omega)$ let $A^*(\beta) = argmax_{a \in \mathcal{A}} \sum_{l=1}^N u_r(a,l) \beta_l$. In any PBE, after signals are realized and the receiver updates to posterior belief $\beta \in \Delta(\Omega)$, the receiver takes an action $a \in A^*(\beta)$ that maximizes expected utility (sequential rationality). Let $R_{\text{indiff}} = \{ \beta \in \Delta(\Omega) : |A^*(\beta)|>1 \}$ be the set of posteriors at which the receiver is indifferent between multiple best actions. We assume that at posteriors in $\beta \in R_{\text{indiff}}$ the receiver breaks ties by choosing the lowest indexed action in $A^*(\beta)$. Hence in equilibrium the reciever takes action $a^*(\beta) = argmin_{a_b \in A^*(\beta)}$ $b$; this is well defined and single valued for each $\beta \in \Delta(\Omega)$. 

For $i=1,...,M$ define $v_i : \Delta(\Omega) \rightarrow \mathbb{R}$ as  sender $i$'s expected utility from any posterior $\beta \in \Delta(\Omega)$ in a PBE following the specified tie breaking rule: $v_i(\beta) = \sum_{l=1}^N u_i(a^*(\beta),l) \beta_l$. First we show that in any PBE satisfying this tie-breaking rule, $v_i$'s are zero-sum and piecewise analytic. This will allow us to directly apply our previous results to them.

\begin{claim} \label{finite_action_claim}
In any PBE in which the receiver chooses action $a^*(\beta) = argmin_{a_b \in A^*(\beta)}$ $b$ after signals are realized, $v_1,...,v_M$ are each piecewise analytic and are zero-sum. 
\end{claim}
\begin{proof}

Zero-sumness is trivial. For all $\beta$, $\sum_{i=1}^M v_i(\beta) = \sum_i \sum_{l=1}^N u_i(a^*(\beta),l) \beta_l$

$= \sum_{l=1}^N \beta_l \sum_i u_i(a^*(\beta),l) = 0$.

Now we show piecewise analyicity. For each $a_b \in \mathcal{A}$, let $\Delta_{a_b} = \{\beta \in \Delta(\Omega) : a_b = a^*(\beta) \}$. Note that for every $a_b \in \mathcal{A}$ the set $\Delta_{a_b}$ is convex; for any $\beta, \beta' \in \Delta_{a_b}$ we can see $a^*(\lambda \beta + (1-\lambda) \beta') =a_b$ for all $\lambda \in (0,1)$ as follows. For any $b'<b$ we have:

\begin{align*}
& \sum_{l=1}^N u_r(a_b,l) \beta_l > \sum_{l=1}^N u_r(a_{b'},l) \beta_l \\
& \sum_{l=1}^N u_r(a_b,l) \beta'_l > \sum_{l=1}^N u_r(a_{b'},l) \beta'_l \\
& \implies \sum_{l=1}^N u_r(a_b,l) (\lambda \beta_l + (1-\lambda) \beta'_l) > \sum_{l=1}^N u_r(a_{b'},l) (\lambda \beta_l + (1-\lambda) \beta'_l) \\
\end{align*}

For any $b'>b$ the first two inequalities must hold weakly which implies the third holds weakly. So at $(\lambda \beta_l + (1-\lambda) \beta'_l)$ $a_b$ yields weakly higher utility for the receiver than all higher indexed actions. Together this implies $a^*(\lambda \beta + (1-\lambda) \beta') =a_b$ and hence $\Delta_{a_b}$ is convex. Note that $\{\Delta_{a_b}\}_{b=1,...,A}$ partition $\Delta(\Omega)$; hence they form a partition of $\Delta(\Omega)$ into convex sets. On each $\Delta_{a_b}$, each $v_i = \sum_{l=1}^N u_i(a_b,l) \beta_l$ ($i=1,...,M$) is linear in $\beta$. Hence each $v_i$ is piecewise linear on $\Delta(\Omega)$. So $\{v_i\}_{i=1,...,M}$ are piecewise analytic.
\end{proof}

Lemma \ref{finite_action_lemma} below shows that all $v_i$'s are linear on $\Delta(\Omega')$ if and only if the receiver prefers the same action at all states in $\Omega'$.

\begin{lemma2} \label{finite_action_lemma}
Fix $\Omega' \subseteq \Omega$, $|\Omega'|>1$. $v_1,...,v_M$ are all linear on $\Delta(\Omega')$ if and only if $a^*(\delta_l) = a^*(\delta_k)$ for all $l,k \in \Omega'$.
\end{lemma2}

\begin{proof}
We prove the `only if' direction by proving the contrapositive. Consider any states $l \neq k \in \Omega'$. Suppose $a^*(\delta_l) \neq a^*(\delta_k)$. For $\beta \in \Delta(\{l,k\}) \subseteq \Delta(\Omega')$, the receiver has expected utility $u_r(a,l) \beta_l + u_r(a,k) \beta_k$. Note by our assumption that no agent is indifferent between any two actions at any state, we must have that $u_r(a^*(\delta_l),l)>u_r(a',l)$ for all $a' \neq a^*(\delta_l)$ and $u_r(a',k)<u_r(a^*(\delta_k),k)$ for all $a' \neq a^*(\delta_k)$. By continuity, this implies that there exists $1 >\beta_k^u >  \beta_k^l>0$ such that for $\beta \in \Delta(\{l,k\})$ if: (1) $\beta_k < \beta_k^l$ then $a^*(\beta) = a^*(\delta_l)$ and (2) $\beta_k > \beta_k^u$ then $a^*(\beta) = a^*(\delta_k)$. For $\beta \in \Delta(\{l,k\})$ and $i \in \{1,...,M\}$, $v_i(\beta) = \beta_k u_i(a^*(\beta),k) + (1- \beta_k) u_i(a^*(\beta),l) = u_i(a^*(\beta),l) + \beta_k(u_i(a^*(\beta),k) - u_i(a^*(\beta),l))$. When $\beta_k > \beta_k^u$, $v_i$ is linear in $\beta_k$ with slope $(u_i(a^*(\delta_k),k) - u_i(a^*(\delta_k),l))$. When $\beta_k < \beta_k^l$, $v_i$ is linear in $\beta_k$ with slope $(u_i(a^*(\delta_l),k) - u_i(a^*(\delta_l),l))$. These slopes are different because:

\begin{align*}
& u_i(a^*(\delta_k),k) - u_i(a^*(\delta_l),k)> 0> u_i(a^*(\delta_k),l) - u_i(a^*(\delta_l),l) \\
& \implies u_i(a^*(\delta_k),k) - u_i(a^*(\delta_k),l) \neq u_i(a^*(\delta_l),k) - u_i(a^*(\delta_l),l)
\end{align*}
 
Hence $v_i$ is nonlinear along $\Delta(\{l,k\})$ and hence nonlinear on $\Delta(\Omega') \supseteq \Delta(\{l,k\})$. 

`If' direction. Suppose for all $l \in \Omega'$, $a^*(\delta_l)=a$. Note that $\Delta(\Omega')$ is the convex hull of $\{\delta_l: l \in \Omega'\}$. Then, by convexity of the set $\Delta_a$ (see proof of Claim \ref{finite_action_claim}), we must have $\Delta(\Omega') \subseteq \Delta_a$. For all $i \in \{1,...,M\}$ $v_i$ is linear on $\Delta_a$ (proof of Claim \ref{finite_action_claim}) and hence $\Delta(\Omega')$.
\end{proof}

Finally we prove our main results for the finite action model. Proposition \ref{finite_action_prop} is important as it says that even when the receiver does not fully learn the state, they learn adequately \textemdash that is, enough that further learning would not influence their action.

\begin{prop} \label{finite_action_prop}
In any PBE in which the receiver chooses action $a^*(\beta) = argmin_{a_b \in A^*(\beta)}$ $b$, the receiver takes their first best action w.p. $1$.
\end{prop}
\begin{proof}
Fix a PBE and consider all posteriors induced by the equilibrium experiments with positive probability. At any posterior $\beta \in \{\delta_1,...,\delta_N\}$, the receiver clearly takes their first best action. Now consider any posterior that occurs with positive probability that does not fully reveal the state. At such a $\beta$, there exists $\Omega' \subseteq \Omega$ with $|\Omega'|>1$ and $\beta_l > 0$ for all $l \in \Omega'$. But then $\Omega'$ is being pooled in equilibrium which implies (Proposition 2) that $v_1,...,v_M$ are all linear on $\Delta(\Omega')$ which implies (Lemma \ref{finite_action_lemma}) that there exists $a \in A$ such that $a^*(\delta_l) = a$ for all $l \in \Omega'$. Note that $\omega \in \Omega'$ (all other states are ruled out by $\beta$). By convexity of $\Delta_a$ (proof of Claim \ref{finite_action_claim}), $a^*(\beta)=a=a^*(\delta_{\omega})$. Hence, the receiver's ex-post payoff is always their first best payoff: $u_r(a^*(\omega),\omega)$. 
\end{proof}

The Corollary \ref{cor_finite_action}, finite action characterization of all equilibria being fully revealing follows immediately from Theorem 1 and Lemma \ref{finite_action_lemma}.

\textbf{Corollary \ref{cor_finite_action} proof.}
\begin{proof}
`If' direction. If for a pair of states $l$ and $k$, $a^*(\delta_l) \neq a^*(\delta_k)$, then by Lemma \ref{finite_action_lemma} some $v_i$ is nonlinear (and hence nonlinear) on $\Delta(\{l,k\})$. If this is true for every pair of states, then along every edge of the simplex we have nonlinearity of some sender's utility function which implies (Theorem 1) full revelation in every equilibrium.

`Only if' direction. If $a^*(\delta_l)=a^*(\delta_k)$ for any pair of states $l$ and $k$, then Lemma \ref{finite_action_lemma} implies that all $v_i$'s are linear on $\Delta(\{l,k\})$. This implies (Theorem 1) that there are non-fully revealing equilibria.
\end{proof}

\subsection{Robustness}
Here we consider the robustness of our results to the assumption that preferences are zero-sum. Consider a game identical to the baseline model but with utility functions $u_1,...,u_M$ ($u_i: \Delta(\Omega) \rightarrow \mathbb{R}$) that need not be zero-sum. We assume that all utilities are piecewise analytic and make the normalization $u_i(\delta_l)=0$ for all senders $i$ and states $l$. We adopt notation from the baseline model whenever it obviously carries over. 

Before presenting the robustness results, which concern the information revealed in equilibrium as preferenes approach zero-sum, we discuss what we can say in this more general setting. As in the baseline model, we have no issues with equilibrium existence; for any $u_1,...u_M$ there is a fully revealing equilibrium $(\Gamma^{FR},...,\Gamma^{FR})$.\footnote{The reason this is an equilibrium is the same \textemdash no sender's experiment is pivotal when others are full revealing the state.} 

Of course our results in the paper, starting with Lemma 1, rely on the zero-sumness of preferences and do not generalize to this setting. While Lemma 1 says senders must get their full revelation payoff in \emph{every} equilibrium of a zero-sum game, this is no longer true when preferences may be nonzero-sum. As an example, suppose all senders have the same preferences; then there will be an equilibrium in which one sender plays a single sender optimal experiment (a Bayesian Persuasion, or BP, solution) and all others play $\Gamma^U$. This BP solution may not fully reveal the state and may yield the senders strictly larger payoffs than the $0$ payoff from full revelation.\footnote{In fact whenever there exists a posterior $\gamma'$ that yields the senders utility larger than $0$, then all BP solutions yield utility larger than $0$. To see this, note that there exists an experiment that puts support only on $\gamma'$ and $\delta_1,...,\delta_N$ (such a construction is shown in the proof of Lemma 1). This experiment yields all senders utility strictly greater than $0$.} More generally, when preferences are not zero-sum, agreement between senders (even if it isn't complete agreement) may allow them to play experiments that yield them strictly positive payoffs in equilibrium.

For any $\gamma \in \Delta(\Omega)$ and $u_1,...,u_M$, let $\sum_i u_i(\gamma)$ be the total surplus shared among senders when the receiver's posterior is $\gamma$. Note this surplus is $0$ when $\gamma \in \{\delta_1,...,\delta_N\}$. Let $MS(u_1,...,u_M) = \sup_{\gamma \in \Delta(\Omega)} \sum_{i=1}^M u_i(\gamma)$ be the supremum of the total surplus senders get at any posterior. Note $MS(u_1,...,u_M) =0$ when the game is zero-sum. While we cannot pin down equilibrium payoffs as we did in Lemma 1, we can upperbound them for each sender. While senders may get above utility $0$ in equilibrium, none can attain utility higher than the maximum surplus:

\begin{lemma2} \label{nonzerosum_lemma}
In any equilibrium $(\Gamma_1,...,\Gamma_M)$, for all senders $i=1,...,M$: $0 \leq U_i(\Gamma_1,...,\Gamma_M) \leq MS(u_1,...,u_M)$.
\end{lemma2}

\begin{proof}
The first inequality, $0 \leq U_i(\Gamma_1,...,\Gamma_M)$, follows from the fact that each sender can always fully reveal the state and obtain payoff $0$.

Note that for all strategy profiles $(\Gamma_1',...,\Gamma_M')$ (with distributions $p_1',...,p_M'$), we have: 

$\sum_i U_i(\Gamma_1,...,\Gamma_M) = \sum_i \mathbb{E}_{p_1',...,p_M'}[u_i(\beta(\Gamma_1,...,\Gamma_M))] = \mathbb{E}_{p_1',...,p_M'} [\sum_i u_i(\beta(\Gamma_1,...,\Gamma_M))] \leq \mathbb{E}_{p_1',...,p_M'} [MS(u_1,...,u_M)] = MS(u_1,...,u_M)$.

Suppose for contradiction that for equilibrium $(\Gamma_1,...,\Gamma_M)$ and sender $i$, $U_i(\Gamma_1,...,\Gamma_M) > MS(u_1,...,u_M)$. Then the previous paragraph implies that there exists sender $j$ with $U_j(\Gamma_1,...,\Gamma_M)<0$. But then $j$ can profitably deviate to $\Gamma_j = \Gamma^{FR}$. Contradiction.
\end{proof}

Lemma \ref{nonzerosum_lemma} implies that when $MS(u_1,...,u_M)=0$, all senders get payoff $0$ in every equilibrium. In one special case, when \emph{only} fully revealing posteriors generate this maximum surplus, all equilibria must be fully revealing:

\begin{corollary} \label{nonzerosum_corollary}
If for all $\gamma' \not \in \{\delta_1,...,\delta_M\}$ $\sum_i u_i(\gamma') < 0$, then the state is fully revealed in every equilibrium.
\end{corollary}

\begin{proof}
Note that $MS(u_1,...,u_M) = 0$. For any strategy profile $(\Gamma_1,...,\Gamma_M)$ that is fully revealing, $U_i(\Gamma_1,...,\Gamma_M) = 0 $ for all $i$. For $(\Gamma_1,...,\Gamma_M)$ that is not fully revealing, $\sum_i U_i(\Gamma_1,...,\Gamma_M) <0$, as with strictly positive probability the surplus at the posterior is strictly less than $0$ (when the state is not fully revealed). This implies that for each non-fully revealing $(\Gamma_1,...,\Gamma_M)$, there is a sender $j$ such that $U_j(\Gamma_1,...,\Gamma_M)<0$. Hence, by Lemma \ref{nonzerosum_lemma}, the state is fully revealed in equilibrium.
\end{proof}

The logic of Corollary \ref{nonzerosum_corollary} is similar to that of Proposition 1 of \cite{gentzkow2016competition}. If sender surplus is uniquely maximized at fully revealing posteriors, any non-fully revealing strategy profile leaves at least one sender strictly worse off than full revelation, which is always an available strategy. 

We now turn our attention to what we can say as preferences approach zero-sum. We consider convergence of utilities under the $\sup$ norm. A sequence of utility functions $\{u^k\}_{k=1}^{\infty}$ converges to a function $u$, or $u_k \rightarrow u$, if $\lim_{k \rightarrow \infty} \sup_{\gamma \in \Delta(\Omega)} |u^k(\gamma) - u(\gamma)| = 0$. For a sequence of profiles of utility functions $\{(u_1^k,...,u_M^k)\}_{k=1}^{\infty}$ (for notational convenience we will drop the limits), $(u_1^k,...,u_M^k) \rightarrow (u_1,...,u_M)$ if $u_i^k\rightarrow u_i$ for $i=1,...,M$. 

For a sequence of strategies/experiments $\{\Gamma^k\}_{k=1}^{\infty}$, with each $\Gamma^k$ distributed according to pmf $p^k$, we say $\Gamma^k \rightarrow \Gamma$ (where $\Gamma$ has distribution $p$), or $\Gamma^k$ converges in distribution to $\Gamma$,  if for all $\gamma \in \Delta(\Omega)$, $\lim_{k \rightarrow \infty} p_k(\gamma)=p(\gamma)$. For any strategy profile $(\Gamma_1,...,\Gamma_M)$, let the random variable $\Gamma(\Gamma_1,...,\Gamma_M)$ denote the receiver's posterior after observing realizations of all $\Gamma_1,...,\Gamma_M$ (i.e. the experiment induced by combining all $M$ senders' experiments).

The following result says that as utilities converge to zero-sum globally nonlinear functions, the information revealed along any sequence of equilibria, whenever convergent, converges to full revelation. This is the same Proposition \ref{prop_robustness} from Section 4.4 stated more formally.

\textbf{Proposition \ref{prop_robustness}.} \emph{Fix a sequence of games with utilities $\{(u_1^k,...,u_M^k)\}$ with $(u_1^k,...,u_M^k) \rightarrow (u_1,...,u_M)$ and $\sum_i u_i(\gamma)=0$ for all $\gamma \in \Delta(\Omega)$. For each $k$ let $(\Gamma_1^k,...,\Gamma_M^k)$ be an equilibrium of game $(u_1^k,...,u_M^k)$. Suppose for every pair of states $l,k \in \Omega$ there exists an $i$ with $u_i$ nonlinear on $\Delta(\{l,k\})$. Then if  $\Gamma(\Gamma_1^k,...,\Gamma_M^k) \rightarrow \Gamma$, $\Gamma=\Gamma^{FR}$.}

Proposition \ref{prop_robustness} is important, as indicates that Theorem 1 does not qualitatively rely on the knife-edge assumption of zero-sum preferences. As utilities get close to zero-sum and globally nonlinear, the information revealed in every equilibrium (if it converges) gets close to full revelation. Before proving the result, we prove a lemma which shows Proposition 2 is similarly robust.

\begin{lemma2} \label{robust_lemma}
Fix a sequence of games with utilities $\{(u_1^k,...,u_M^k)\}$ with $(u_1^k,...,u_M^k) \rightarrow (u_1,...,u_M)$ and $\sum_i u_i(\gamma)=0$ for all $\gamma \in \Delta(\Omega)$. For each $k$ let $(\Gamma_1^k,...,\Gamma_M^k)$ be an equilibrium for $(u_1^k,...,u_M^k)$. Suppose for some $\Omega' \subseteq \Omega$ and $i$, $u_i$ is nonlinear on $\Delta(\Omega')$. Then if $\Gamma(\Gamma_1^k,...,\Gamma_M^k) \rightarrow \Gamma$, $\Gamma$ does not pool $\Omega'$.
\end{lemma2}

\begin{proof}
For strategy profile $(\Gamma_1',...,\Gamma_M')$ (with distributions $p_1',...,p_M'$), let $U_i^k(\Gamma_1',...,\Gamma_M') = \mathbb{E}_{p_1',...,p_M'}[u_i^k(\beta)]$ be sender $i$'s expected utility when she has preferences $u_i^k$. 

Let the distribution of $\Gamma$ be $p$ and for any $k$ let the distribution of $\Gamma(\Gamma^k_1,...,\Gamma^k_M)$ be $p^k$. Note that as $\Gamma(\Gamma_1^k,...,\Gamma_M^k)$ has finite support for every $k$, $\Gamma$ must also finite support (by the definition of convergence in distribution). Define $T = supp[\Gamma] \cup (\cup_{k=1}^{\infty} supp[\Gamma(\Gamma_1^k,...\Gamma_M^k)])$; note by the finiteness of all terms in the union, $T$ is countable.

Suppose for contradiction that some $u_i$ is nonlinear on $\Omega'$ and $\Gamma$ pools $\Omega'$. By Claim \ref{strong_claim_prop2}, there exists a sender $j$ such who can find experiment $\Gamma_j$ (with distribution $p_j$) such that in the limiting (zero-sum) game $(u_1,...,u_M)$, $U_j(\Gamma_j,\Gamma)=c>0$.

For any $k$, $MS(u_1^k,...,u_M^k) = \sup_{\gamma} \sum_i u_i^k(\gamma) = \sup_{\gamma} \sum_i u_i(\gamma) + (u_i^k(\gamma)-u_i(\gamma)) \leq  \sup_{\gamma} \sum_i u_i(\gamma) + |u_i^k(\gamma)-u_i(\gamma)|$. As  $(u^k_1,...,u^k_M) \rightarrow (u_1,...,u_M)$, the second term goes to $0$ for all $i$ as $k \rightarrow \infty$ and hence $MS(u_1^k,...,u_M^k) \rightarrow MS(u_1,...,u_M) = 0$. This implies there exists $K$ s.t. $\forall k>K$, $MS(u_1^k,...,u_M^k) < c/2$. By Lemma \ref{nonzerosum_lemma}, for all $k>K$, $U^k_j(\Gamma^k_j,\Gamma^k_{-j})<c/2$. 

Consider $j$ playing experiment $\Gamma^k_j$ as well as (conditionally independently) playing $\Gamma_j$, while her opponents' play $\Gamma^k_{-j}$. The expected payoff that $j$ gets from this can be written $U^k_j(\Gamma(\Gamma_j,\Gamma^k_j),\Gamma^k_{-j}) = U^k_j(\Gamma_j,\Gamma(\Gamma^k_1,...,\Gamma^k_M))$, as it does not affect $j$'s payoff if she plays $\Gamma^k_{j}$ or her opponents' do.  As $T \supset supp[\Gamma(\Gamma^k_1,...,\Gamma^k_M)]$ for each $k$, we can write $U^k_j(\Gamma_j,\Gamma(\Gamma^k_1,...,\Gamma^k_M)) = \sum_{x \in supp[\Gamma_j]} \sum_{y \in supp[T]} u^k_j(\beta(x,y))p^k(y|x)p_j(x)$. Then:

\begin{align*}
& \lim_{k \rightarrow \infty} U^k_j(\Gamma_j,\Gamma(\Gamma^k_1,...,\Gamma^k_M)) = \sum_{x \in supp[\Gamma_j]} \sum_{y \in supp[T]} \lim_{k \rightarrow \infty} u^k_j(\beta(x,y))p^k(y|x)p_j(x)
\end{align*}

For any $y \in T$, $\lim_{k \rightarrow \infty} u^k_j(\beta(x,y))= u_j(\beta(x,y))$ and $\lim_{k \rightarrow \infty} p^k(y|x) = p(y|x)$ (by definition of $p^k(y|x)$ and convergence in distribution). Hence:

\begin{align*}
& \lim_{k \rightarrow \infty} U^k_j(\Gamma_j,\Gamma(\Gamma^k_1,...,\Gamma^k_M)) = \sum_{x \in supp[\Gamma_j]} \sum_{y \in supp[T]} u_j(\beta(x,y))p(y|x)p_j(x) = U_j(\Gamma_j,\Gamma) = c 
\end{align*}

This implies that there exists $K'$ such that $\forall k>K'$, $U^k_j(\Gamma_j,\Gamma(\Gamma^k_1,...,\Gamma^k_M))>c/2$. Take $K'' = \max\{K,K'\}$. For all $k > K''$ we have: $U^k_j(\Gamma^k_j,\Gamma^k_{-j})<c/2$ and $U^k_j(\Gamma_j,\Gamma(\Gamma^k_1,...,\Gamma^k_M))>c/2$. But for all $k>K''$, this contradicts that $(\Gamma_1^k,...,\Gamma_M^k)$ is an equilibrium as $j$ has a profitable deviation of $\Gamma(\Gamma_j,\Gamma^k_j)$.
\end{proof}

We now prove Proposition \ref{prop_robustness}. 

\begin{proof}
Suppose not. Then $\Gamma(\Gamma^k_1,...,\Gamma^k_M) \rightarrow \Gamma \neq \Gamma^{FR}$. Then $\Gamma$ pools some set of states $\Omega'$ with $|\Omega'| \geq 2$; let $l,k \in \Omega'$. There is a sender with $u_i$ nonlinear on $\Delta(\{l,k\})$; hence $u_i$ is nonlinear on $\Delta(\Omega')$. But then by Lemma \ref{robust_lemma} $\Gamma$ must not pool $\Omega'$. Contradiction.
\end{proof}

Proposition \ref{prop_robustness} relates to the standard results on the upper hemicontinuity of the set of equilibria (although here we are not concerned with the set of equilibrium actions themselves but instead the set of information that could be revealed in equilibrium). As is also standard, we do not have the corresponding lower hemicontinuity properties. In particular, it is possible for there to be non-fully revealing equilibria in the limit, but only fully revealing equilibria along the sequence. It is not hard to come up with examples of this; we provide a one here.

\begin{example} \label{no_lhc_robust}
Suppose $\Omega=\{0,1\}$ and there are two senders $1$ and $2$. Consider the sequence of utility functions $\{(u_1^k,u_2^k)\}$ with $u^k_1(\beta) = \frac{-1}{k}$ for all $\beta \in [0,1] \setminus \{0,1\}$, $u_1^k(0)=u_1^k(1)=0$, and $u_2(\beta)=0$ for all $\beta \in [0,1]$. Define utility function $u$ as $u(\beta)=0$ for all $\beta$. Then note $(u_1^k,u_2^k) \rightarrow (u,u)$. Proposition 1 states that in the game $(u,u)$, there are non-fully revealing equilibria (any strategy profile is an equilibrium). Note for any $k$, the game $(u_1^k,u_2^k)$, $u_1^k(\beta)+u_2^k(\beta)<0$ for all $\beta \not \in \{0,1\}$. Hence by Corollary \ref{nonzerosum_corollary}, all equilibria of $(u_1^k,u_2^k)$ are fully revealing for all $k$.
\end{example}

\subsection{Infinite signal experiments} \label{section_inf}

So far we have restricted senders to choosing experiments with a finite number of signals, or equivalently interim belief distributions $p_i \in P$ that have finite support. In this section we demonstrate that our takeaway from the finite signal results \textemdash that for typical sender preferences we have full revelation in every equilibrium \textemdash extends when senders can choose from a more general set of experiments. Senders now choose any experiments $\Pi: \Omega \rightarrow \Delta(S)$ (with no restrictions on the signal space). Again, we recast choices of experiments as choices of interim belief distributions. For technical convenience we restrict our attention to senders choosing interim belief distributions that can be written as the sum of absolutely continuous and discrete distributions.\footnote{As $\Delta(\Omega) \subset \mathbb{R}^N$, by the Lebesgue Decomposition Theorem this only rules senders choosing distributions with singular continuous components.} We will call the space of pure strategies, or the set of Bayes-plausible distributions that satisfies this requirement $G$. Formally $G = \{g \in \Delta(\Delta(\Omega)): \mathbb{E}_g[\Gamma_i]= \pi,$ $g=g_c+g_d$ for some abs. cont. and discrete (respectively) measures $g_c,g_d \in \Delta(\Delta(\Omega))\}$. Each $g \in G$ is a generalized density.\footnote{$g$ is the density function $g_c$ on intervals where the distribution is absolutely continuous and the probability mass function $g_d$ everywhere else.} A strategy for sender $i$ is a choice of random variable $\Gamma_i$ with generalized density $g_i \in G$. Preferences over strategy profiles for a sender $i$ are given by $U_i(\Gamma_1,...,\Gamma_M) = \mathbb{E}_{g_1,...,g_M}[u_i(\beta)]$. The game and equilibrium concept are otherwise identical to the finite signal model (including our normalization of the $u_i$'s).

\textbf{Remark:} Note that our equilibrium analysis under both the finite signal restriction and under the technical restriction above can be seen as equilibrium selections. Any finite signal equilibrium will also be an equilibrium when senders are allowed to pick strategies from $G$ and any equilibria with strategies selected from $G$ will be equilibria in a game where senders can pick any distribution in $\Delta(\Delta(\Omega))$. 

We make one additional technical assumption on utility functions:

\begin{assp} \label{technical_assp}
For each $l \in \Omega$ and $i \in \{1,...,M\}$, $u_i$ is real analytic in some neighborhood of $\delta_l$. 
\end{assp}

Note that Assumption \ref{technical_assp} only rules out piecewise analytic utility functions for which some $\delta_l$ lies on the boundary between different pieces. For any states $l,k$ let $v^{l,k} \in \mathbb{R}^{N}$ be the vector from $\delta_l$ to $\delta_k$.\footnote{$v^{l,k}_k=1$, $v^{l,k}_l = - 1$, $v^{l,k}_n = 0$ for all $n \neq l,k$.} For any sender $i$ let $\nabla_{v^{l,k}} u_i(\cdot)$ be the directional derivative of $u_i$ moving along $v^{l,k}$. Note for some $i$ $\nabla_{v^{l,k}} u_i(\cdot)$ may not be well defined at some points on $\Delta(\Omega)$; but under Assumption 7.1  for all $i$ and all $l$ it is a well defined continuous function in some neighborhood of $\delta_l$. With this assumption we can define Condition \ref{condition1}. Condition \ref{condition1} concerns the shape of utility functions on  an edge of the simplex and will be sufficient for a pair of states $l$ and $k$ to not be pooled in every equilibrium. 

\begin{definition} \label{condition1}
For any states $l,k$ and sender $i$ we say that $u_i$ satisfies Condition \ref{condition1} on $\Delta(\{l,k\})$ if either $\nabla_{v^{l,k}} u_i(\delta_l) \neq u_i(\delta_k)-u_i(\delta_l)$, $\nabla_{v^{l,k}} u_i(\delta_k) \neq u_i(\delta_k)-u_i(\delta_l)$, or both.\footnote{Under the normalization we made, $u_i(\delta_k)-u_i(\delta_l)=0$.}
\end{definition}

For example, utility functions that look like Figure 7.2 along edge $\Delta(\{l,k\})$ do not satisfy Condition \ref{condition1}; nor does a utility function that is linear along that edge. Functions that look this those in Figure 1 do satisfy Condition \ref{condition1}.

\begin{center}
\includegraphics[scale=0.5]{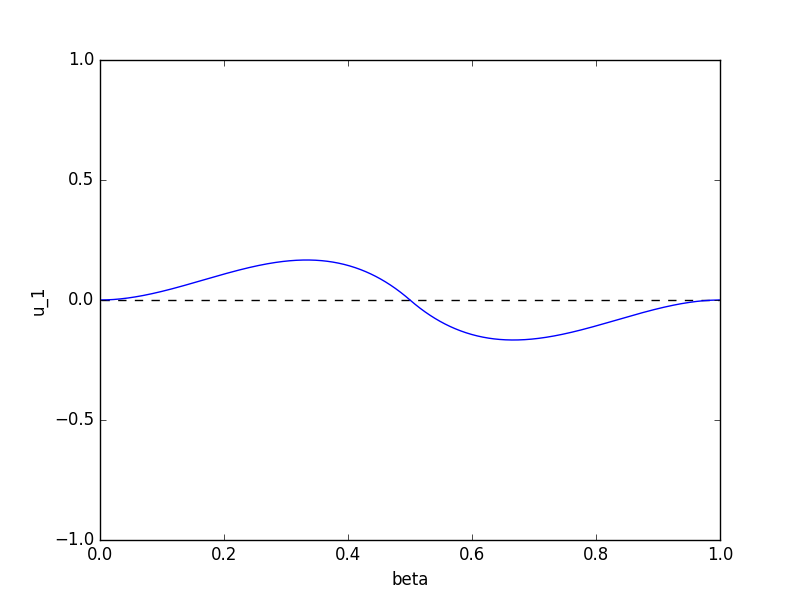}
\captionof{figure}{}
\end{center}

\begin{prop2} \label{prop_pooling_inf}
For any pair of states $l$ and $k$, if there exists a sender $i$ with $u_i$ satisfying Condition \ref{condition1} on $\Delta(\{l,k\})$, then $l$ and $k$ are not pooled in every equilibrium.
\end{prop2}

We prove Proposition \ref{prop_pooling_inf} after redefining some objects for this setting. For any sender $i$ and strategy profile for all opponents $\{\Gamma_j\}_{j \neq i}$, we define $W_i(x)$ for $x \in \Delta(\Omega)$ analogously to the finite signal case.

\begin{equation}
W_i(x) = \int_{\Delta(\Omega)} u_i(\beta(x,y)) p_{-i}(y|x) dy 
\end{equation}

For any vector $v \in \mathbb{R}^N$ and $y \in \Delta(\Omega)$, let $\nabla_v p_{-i}(y|x)$ and $\nabla_v \beta_k(x,y)$ (for any state $k$) be the directional derivates of these two functions with respect to $x$ along vector $v$. Note the following. 
 
\begin{equation} \label{d_p(y|x)}
\nabla_v p_{-i}(y|x) = \nabla_v \sum_{l=1}^N \frac{x_l y_l p_{-i}(y)}{\pi_l} = \sum_{l=1}^N \frac{v_l y_l p_{-i}(y)}{\pi_l}
\end{equation}
  
\begin{equation} \label{d_beta}
\nabla_v \beta_k(x,y) = \nabla_v \frac{\frac{x_k y_k}{\pi_k}}{\sum_{l=1}^N \frac{x_l y_l}{\pi_l}} = - \sum_{l=1}^N \frac{\frac{v_l y_l}{\pi_l} \frac{x_k y_k}{\pi_k}}{(\sum_{n=1}^N \frac{x_n y_n}{\pi_n})^2} + \frac{\frac{v_k y_k}{\pi_k}}{\sum_{l=1}^N x_l y_l / \pi_l}
\end{equation} 

For any states $l,k \in \{1,...,N\}$ let $\Delta^{int}(\{l,k\}) = \Delta(\{l,k\}) \setminus \{\delta_l,\delta_k\}$ be the nondegenerate beliefs in $\Delta(\{l,k\})$. Let $\Delta^0(\{l,k\}) = \{\beta \in \Delta(\Omega): \beta_l, \beta_k=0\}$ and $\Delta^1(\{l,k\}) = \{\beta \in \Delta(\Omega): \beta_l+\beta_k<1$; $\beta_l$ and/or $\beta_k \neq 0 \}$. Note $\{\delta_l,\delta_k\} \cup \Delta^{int}(\{l,k\}) \cup \Delta^0(\{l,k\}) \cup \Delta^1(\{l,k\}) = \Delta(\Omega)$ and all four sets are disjoint.

We now prove Proposition \ref{prop_pooling_inf}

\begin{proof}
First note Lemma 1 still holds in this context by an identical proof. In any equilibrium $(\Gamma_1,...,\Gamma_M)$ all senders $i$ have: $U_i(\Gamma_1,...,\Gamma_M)=0$ and $W_i(x) \leq 0$ for all $x \in \Delta(\Omega)$.

WLOG let $l=1$, $k=2$. For notational convenience let $v = v^{1,2}$. Suppose $u_i$ satisfies Condition \ref{condition1} on $\Delta(\{1,2\})$ for some $i$. WLOG we consider the case $\nabla_v u_i(\delta_2) \neq u_i(\delta_2) - u_i(\delta_1)=0$. Then by zero-sumness (derivatives must also be zero-sum where they exist for all senders) there exists a sender $j$ with $\nabla_v u_j(\delta_2)=c<0$. We will prove that $\Gamma_{-j}$ must not pool $\{l,k\}$; this clearly implies the Proposition \ref{prop_pooling_inf} as by Claim \ref{not_pool}, $(\Gamma_1,...,\Gamma_M)$ also will not pool $\{l,k\}$.

Suppose for contradiction there is an equilibrium $(\Gamma_1,...,\Gamma_M)$ such that $\Gamma_{-j}$ pools $\{l,k\}$. For $x \in \Delta(\Omega)$, by the product rule and the partition of $\Delta(\Omega)$ into $\{\delta_1,\delta_2\}, \Delta^{int}(\{1,2\})$, $\Delta^0(\{1,2\}), \Delta^1(\{1,2\})$:

\begin{align*}
& \nabla_v W_j(x) = \nabla_v (u_j(\beta(x,\delta_2))p_{-j}(\delta_2|x)+ u_j(\beta(x,\delta_2))  \nabla_v p_{-j}(\delta_2|x)\\
& + \nabla_v (u_j(\beta(x,\delta_1))p_{-j}(\delta_1|x) + u_j(\beta(x,\delta_1)) \nabla_v p_{-j}(\delta_1|x) \\
& + \int_{\Delta^{int}(\{1,2\})} \nabla_v u_j(\beta(x,y))p_{-j}(y|x)dy + \int_{\Delta^{int}(\{1,2\})} u_j(\beta(x,y)) \nabla_v p_{-j}(y|x)dy \\
& + \int_{\Delta^{0}(\{1,2\})} \nabla_v u_j(\beta(x,y))p_{-j}(y|x)dy + \int_{\Delta^{0}(\{1,2\})} u_j(\beta(x,y)) \nabla_v p_{-j}(y|x)dy \\
& + \int_{\Delta^{1}(\{1,2\})} \nabla_v u_j(\beta(x,y))p_{-j}(y|x)dy + \int_{\Delta^{1}(\{1,2\})} u_j(\beta(x,y)) \nabla_v p_{-j}(y|x)dy
\end{align*}

though this may not be well defined for some $x$. Note:

\begin{align*}
& \nabla_v(u_j(\beta(x,y))p_{-j}(\delta_2|x)) = \sum_{k=1}^N \frac{\partial u_j(\beta(x,y))}{\partial \beta_k} \nabla_v \beta_k(x,y) p_{-j}(y|x) \\
\end{align*}

(again this may not be well defined for some $x$). Now consider $x \in \Delta(\{1,2\})$. For such $x$, with some algebra:

\begin{equation} \label{derivative}
\nabla_v \beta_k(x,y) p_{-j}(y|x) = \frac{(y_1/\pi_1 - y_2/\pi_2) \frac{x_k y_k}{\pi_k}}{\frac{x_1y_1}{\pi_1}+ \frac{x_2y_2}{\pi_2}} + \frac{v_k y_k}{\pi_k}
\end{equation}

One can check that for $y \in \Delta^0(\{1,2\})$ and for $y \in \{\delta_1,\delta_2\}$ this expression is $0$ for all $k=1,...,N$. This tells us that the 1st, 3rd, and 7th terms of $\nabla_v W_j(x)$ are $0$ for $x \in \Delta(\{1,2\})$. Note that for $y \in \Delta(\{1,2\})^0$, $\nabla_v p_{-i}(y|x) = 0$, and so the 8th term is also $0$.

We consider the limit of $\nabla_v W_j(x)$ for $x \in \Delta(\{1,2\})$ as $x \rightarrow \delta_2$. We show this limit exists and is negative.

\begin{align*}
& \lim_{x \rightarrow \delta_2} \nabla_v W_j(x) = \lim_{x \rightarrow \delta_2}  u_j(\beta(x,\delta_2))  \nabla_v p_{-j}(\delta_2|x) + \lim_{x \rightarrow \delta_2}  u_j(\beta(x,\delta_1))  \nabla_v p_{-j}(\delta_1|x) \\ 
& + \int_{\Delta^{int}(\{1,2\})} \sum_{k=1}^N \lim_{x \rightarrow \delta_2} \frac{\partial u_j(\beta(x,y))}{\partial \beta_k}  \nabla_v \beta_k(x,y) p_{-j}(y|x) dy + \int_{\Delta^{int}(\{1,2\})} \lim_{x \rightarrow \delta_2} u_j(\beta(x,y)) \nabla_v p_{-j}(y|x)dy \\
& + \int_{\Delta^{1}(\{1,2\})} \sum_{k=1}^N \lim_{x \rightarrow \delta_2} \frac{\partial u_j(\beta(x,y))}{\partial \beta_k} \nabla_v \beta_k(x,y) p_{-j}(y|x) dy + \int_{\Delta^{1}(\{1,2\})} \lim_{x \rightarrow \delta_2} u_j(\beta(x,y)) \nabla_v p_{-j}(y|x)dy
\end{align*}

We evaluate this term by term. The first term is $0$ as $\nabla_v p_{-j}(\delta_2|x)$ is finite and does not depend on $x$ and $\lim_{x \rightarrow \delta_2}  \beta(x,\delta_2) = \delta_2$ which implies (by continuity of $u_j$ in a neighborhood of $\delta_2$) $\lim_{x \rightarrow \delta_2} u_j(\beta(x,\delta_2)) = 0$. The second term is also $0$ as $\nabla_v p_{-j}(\delta_1|x)$ is finite and does not depend on $x$ and $\lim_{x \rightarrow \delta_2}  \beta(x,\delta_1) = \delta_1$ (by L'Hopital's rule) which implies (by continuity of $u_j$ in a neighborhood of $\delta_1$) $\lim_{x \rightarrow \delta_2} u_j(\beta(x,\delta_1)) = 0$. 

The fourth term is also $0$ as for all $y \in \Delta^{int}(\{1,2\})$, $\lim_{x \rightarrow \delta_2}  \beta(x,y) = \delta_2$ and $u_j(\delta_2) = 0 $ while $\nabla_v p_{-j}(y|x)$ is finite and does not depend on $x$. The sixth term is also $0$ for the following reason. For $y \in \Delta^1(\{1,2\})$ with $y_2 > 0$, $\lim_{x \rightarrow \delta_2} \beta(x,y) = \delta_2$; as $u_j(\delta_2)=0$, the terms inside the integral are $0$ when $y_2>0$. For $y \in \Delta^1(\{1,2\})$ with $y_2=0$, we must have $y_1>0$; then by L'Hopital's rule $\lim_{x \rightarrow \delta_2} \beta(x,y) = \delta_1$ and $u_j(\delta_1)=0$ meaning terms inside the interal are $0$ when $y_2=0$. 

Using \eqref{derivative} note that $\nabla_v \beta_k(x,y) p_{-j}(y|x) = 0$ for all $k \neq 1,2$ as for such $k$ $x_k=v_k=0$. For $y \in \Delta(\{1,2\})^{int}$ we have: $\lim_{x \rightarrow \delta_2} \nabla_v \beta_1(x,y) p_{-j}(y|x) = \frac{-y_1}{\pi_1}p_{-j}(y)$ and $\lim_{x \rightarrow \delta_2} \nabla_v \beta_2(x,y) p_{-j}(y|x) = \frac{y_1}{\pi_1}p_{-j}(y)$. The same holds for $y \in \Delta^1(\{1,2\})$ with $y_2>0$. For $y \in \Delta^1(\{1,2\})$ with $y_2=0$, we have: $\lim_{x \rightarrow \delta_2} \nabla_v \beta_1(x,y) p_{-j}(y|x) =\lim_{x \rightarrow \delta_2} \nabla_v \beta_2(x,y) p_{-j}(y|x)=0$ (applying L'Hopital's rule). 

Putting this together:
\begin{equation} \label{delta_W_final}
\begin{split}
& \lim_{x \rightarrow \delta_2} \nabla_v W_j(x) = \int_{\Delta^{int}(\{1,2\}) \cup \{y \in \Delta^1(\{1,2\}): y_2>0\}} [\frac{\partial u_j(\delta_2)}{\partial \beta_2}  - \frac{\partial u_j(\delta_2)}{\partial \beta_1} ] \frac{y_1}{\pi_1} p_{-j}(y) dy \\
& = \nabla_v u_j(\delta_2) \int_{\Delta^{int}(\{1,2\}) \cup \{y \in \Delta^1(\{1,2\}): y_2>0\}} \frac{y_1}{\pi_1} p_{-j}(y) dy \\
& = c \int_{\Delta^{int}(\{1,2\}) \cup \{y \in \Delta^1(\{1,2\}): y_2>0\}} \frac{y_1}{\pi_1} p_{-j}(y) dy \\
\end{split}
\end{equation}

As $\Gamma_{-j}$ pools $\Omega'$, there exists $y \in supp[\Gamma]$ for which $y_1,y_2>0$. Such a $y$ must fall inside the set $(\Delta^{int}(\{1,2\}) \cup \{y \in \Delta^1(\{1,2\}): y_2>0\})$ (any point in the completement of this set assigns probability $0$ to at least one of states $1,2$.). This implies that there are $y \in (\Delta^{int}(\{1,2\}) \cup \{y \in \Delta^1(\{1,2\}): y_2>0\})$ for which $p_j(y)>0$ and $y_1>0$. As $y_1' \geq 0$ for all $y' \in \Delta(\Omega)$, the integral on the righthand side of equation \ref{delta_W_final} is strictly positive. As $c<0$, we have $\lim_{x \rightarrow \delta_2}\nabla_v W_j(x) <0$.

But as $W_j(\delta_2)=0$, this implies that for some $x^* \in \Delta(\{1,2\})$ close enough to $\delta_2$, we must have $W_j(x^*)>0$, contradicting Lemma 1.

\end{proof}

Condition \ref{condition1} holding on each edge for some sender is a sufficient condition for full revelation in any equilibrium:

\textbf{Proposition \ref{prop_inf}.} \emph{If for every pair of states $l$ and $k$ there exists a sender $i$ such that $u_i$ satisfies Condition \ref{condition1} on $\Delta(\{l,k\})$ then the state is fully revealed in every equilibrium in which senders choose experiments from $G$.}

\begin{proof}
If for every pair $l,k$ there is a sender with $u_i$ satisfying Condition \ref{condition1} on $\Delta(\{l,k\})$, then by Proposition \ref{prop_pooling_inf}, no pair of states is pooled in any equilibrium. Hence w.p. $1$ the posterior assigns positive probability to only $1$ state and hence must fully reveal that state.
\end{proof}

Note that a function not satisfying Condition \ref{condition1} along a given edge is knife-edge \textemdash it requires a particular directional derivative to take a certain value at two points. If no sender has a utility function satisfying Condition \ref{condition1} along an edge, this is even more particular. Hence we `typically' expect Condition \ref{condition1} to be satified on each edge for some $u_i$ and so Proposition \ref{prop_inf} says we should typically expect fully revelation in every equilibrium.

\subsection{Privately informed senders}

Consider our baseline model with one modification: each sender receives a private signal before the game. For simplicity, senders' private signals are realizations of finite signal experiments that are conditionally (on $\omega$) independent across senders.\footnote{These assumptions are not necessary.} We think of the experiments in terms of the beliefs they induce. Formally each sender $i$ draws a private belief $b_i \in \Delta(\Omega)$ with $b_i \sim B_i \in P$, $|supp[b_i]| < \infty$, and $\mathbb{E}[b_i]=\pi$ (Bayes-plausibility). The distributions $\{B_i\}_{i=1}^N$ are conditionally independent. We make one more assumption: that for each $\Omega' \subsetneq \Omega$, $supp[b_i]\cap \Delta(\Omega') = \emptyset$ for all $i$; no sender's private information rules out any states.\footnote{Having already made the assumption of finite signals, this assumption is equivalent saying signals are bounded.} 

A pure strategy for sender $i$ is a mapping from private beliefs (or types) to choices of finite signal experiments: $\sigma_i: supp[b_i] \rightarrow P$. As before, we use $\Gamma_i$ to denote the interim belief produced by sender $i$'s experiment. $i$ chooses the distribution of $\Gamma_i$, $p_i \in P$ after observing her own type. Importantly, we define $\Gamma_i$ to be the interim belief the receiver holds after viewing realization of $i$'s experiment but \emph{without} updating her belief on $\Omega$ from observing the choice of $p_i$ (we formalize this updating in the next paragraph). Hence $\Gamma_i$ is the receiver's learning from $i$'s experiment ignoring information from signalling. A pure strategy profile is a vector $(\sigma_1,...,\sigma_M)$.

The receiver's posterior belief $\beta$ is a function of the signal realizations she observes as well as the experiment choices she observes. The receiver will form beliefs about each $\{b_i\}_{i=1,...,M}$ independently via a belief function $\mu^i: \Delta(\Delta(\Omega)) \rightarrow \Delta(supp[b_i])$ ($i=1,...,M$) which maps choices of experiment to a belief on the sender's type.\footnote{While other senders' experiment choices and the realizations of $(\Gamma_1,...,\Gamma_M)$ will affect the receiver's belief about each $b_i$, as players' types and experiment realizations are conditionally independent they will only affect the receiver's beliefs through learning about $\omega$. This updating will hence not affect the receiver's belief on $\omega$, which is all that players care about. The functions $\{\mu^i\}$ are what is important for evaluating senders' payoffs.} For an experiment choice of $p_i$ by sender $i$, let $\mu_i(p_i)[b]$ denote the probability the receiver assigns to $b_i =b$ under belief function $\mu^i$. For $i=1,...,M$ let $\tau^i(\mu^i,p_i) \in \Delta(\Omega)$ be the receiver's belief on $\omega$ given belief function $\mu^i$ after observing experiment choice $p_i$ from sender $i$ but not its signal realization or any other senders' experiment choices are realizations. Then for each state $l \in \Omega$:

\begin{equation}
\tau^i_l(\mu^i,p_i) = \sum_{b \in supp[b_i]} Pr(\omega=l|b_i=b) Pr(b_i=b|p_i) = \sum_{b \in supp[b_i]} b_l \mu^i(p_i)[b]
\end{equation}

For any sender $i$ let $\alpha^i(\Gamma_i,p_i,\mu^i) \in \Delta(\Omega)$ be the random variable representing the receiver's interim belief given $\mu^i$ after observing just the choice $p_i$ and the realization of $\Gamma_i$. $\alpha^i$ hence captures the receiver's belief after taking into account all information \textemdash signalling and otherwise \textemdash from sender $i$. For any $l \in \Omega$:

\begin{equation}\label{alpha}
\alpha_l^i(\Gamma_i,p_i,\mu^i) = \frac{\tau_l^i(\mu^i,p_i) \Gamma_{i,l} / \pi_l}{\sum_{k=1}^N \tau_k^i(\mu^i,p_i) \Gamma_{i,k} / \pi_k}
\end{equation}

For fixed $\{\mu^i\}_{i=1,...,M}$, after observing $\{p_i\}$ and realizations $\{\Gamma_i\}$, the receiver updates by Bayes rule to a posterior belief for each $l \in \Omega$:

\begin{equation} \label{private_beta}
\beta_l(\{\Gamma_i\},\{p_i\},\{\mu^i\}) = \frac{\Pi_{i=1}^M \alpha^i_l(\Gamma_i,p_i,\mu^i) / \pi_l^{M-1}}{\sum_{k=1}^N \Pi_{s=1}^M \alpha^i_k(\Gamma_i,p_i,\mu^i)/ \pi_k^{M-1}}
\end{equation}

A PBE (in pure strategies) is a strategy profile $(\sigma_1,...,\sigma_M)$ and a set of belief functions $(\mu^1,...,\mu^M)$ satisfying two conditions. First, no sender $i$ can strictly gain from deviating from $\sigma_i(b_i)$ for any $b_i \in supp[b_i]$:

\begin{equation}
\begin{split}
& \forall i \in \{1,...,M\} \text{, } b_i \in supp[b_i]: \text{ } \mathbb{E}_{\{B_j\}_{j \neq i}}[\mathbb{E}_{\sigma_i(b_i),\{p_j\}_{j \neq i}}  [u_i(\beta)]|b_i] \\
& \geq \mathbb{E}_{\{B_j\}_{j \neq i}} [\mathbb{E}_{p_i',\{p_j\}_{j \neq i}}[u_i(\beta)]|b_i] \text{ for all } p_i' \in P
\end{split}
\end{equation}

where the receiver's posterior $\beta$ is formed using \eqref{private_beta}. It is important to note that sender $i$ may not know $\{p_j\}_{j \neq i}$ but forms beliefs about these given her own private information to evaluate expected utility.

Second, beliefs must follow Bayes rule on path:

\begin{equation}
\begin{split}
& \forall i \in \{1,...,M\} \text{ and } p \in P \text{ s.t. } \exists b_i \in supp[b_i] \text{ with } \sigma_i(b_i) =p: \\
& \text{ } \forall b \in supp[b_i] \text{, } \mu^i(p)[b] = \frac{\mathbbm{1}_{\sigma(b)=p} B_i(b)}{\sum_{b' \in supp[b_i]}\mathbbm{1}_{\sigma(b')=p} B_i(b')}
\end{split}
\end{equation}

The result below is Lemma 1 adapted to this setting with private information.

\begin{lemma2} \label{lemma1_private_info}
Take any equilibrium $(\sigma_1,...,\sigma_M)$, $(\mu^1,...,\mu^M)$. For any sender $i$ and $b \in supp[b_i]$, conditional on $b_i =b$ sender $i$ gets expected utility $0$.
\end{lemma2}

\begin{proof}
Fix any equilibrium. No sender $i$ can get expected utility strictly less than $0$ conditional on any $b_i =b \in supp[b_i]$, as playing the fully revealing experiment guarantees expected utility $0$. This means each sender $i$'s expected utility unconditional on type,

\begin{equation} \label{unconditional_payoff}
\sum_{b' \in supp[b_i]}\mathbb{E}_{\{B_j\}_{j \neq i}}[\mathbb{E}_{\sigma_i(b'),\{p_j\}_{j \neq i}}[u_i(\beta)]] B_i(b')
\end{equation}

is weakly positive. As the game is zero-sum, the sum of all senders' unconditional expected utilities must be $0$; as each of these payoffs is weakly positive, it must be Equation \ref{unconditional_payoff} is equal to $0$ for each $i$. But then as each term in the summation of Equation \ref{unconditional_payoff} is weakly positive, sender $i$'s expected utility conditional on $b_i = b$ cannot be strictly positive, and hence it must be $0$.
\end{proof}

We redefine state pooling in the game with private information as follows. An equilibrium $(\sigma_1,...,\sigma_M)$, $(\mu^1,...,\mu^M)$ does not pool a set of states $\Omega' \subseteq \Omega$ if $Pr(\beta_l>0 \text{ } \forall l \in \Omega')=0$. Otherwise, the equilibrium pools $\Omega'$. The following Lemma is useful for the main results; it says an equilibrium pools $\Omega'$ if and only if $\alpha^i$ does for every sender $i$.

\begin{lemma2} \label{lemma_private_inf_pooling}
An equilibrium $(\sigma_1,...,\sigma_M)$, $(\mu^1,...,\mu^M)$ pools $\Omega' \subseteq \Omega$ if and only if for all $i$ $Pr(\alpha^i_l(\Gamma_i,p_i,\mu^i)>0 \text{ } \forall l \in \Omega')>0$.
\end{lemma2}

\begin{proof}
`If' direction. If for all $i$ $Pr(\alpha^i_l(\Gamma_i,p_i,\mu^i)>0 \text{ } \forall l \in \Omega')$, then as $\alpha^i$ is conditionally (on state) across senders (as $b_i$ and $\Gamma_i$ are), $Pr(\alpha^i_l(\Gamma_i,p_i,\mu^i)>0 \text{ } \forall l \in \Omega' \text{ } \forall i=1,...,M)>0$. By equation \ref{private_beta}, $Pr(\beta_l>0 \text{ } \forall l \in \Omega')>0$.

`Only if' direction. If for some sender $i$, $Pr(\alpha^i_l(\Gamma_i,p_i,\mu^i)>0 \text{ } \forall l \in \Omega')=0$, then by equation \ref{private_beta}, $\beta_l=0$ for some $l \in \Omega'$ w.p. $1$.
\end{proof}

We now provide a sufficient condition for a pair of states to be not pooled in every equilibrium. The sufficient condition is the same as that in Section \ref{section_inf} and will lead to the same sufficient condition for full revelation in all equilibria. We note as before that this condition is satisfied for all but a knife-edge case of sender preferences. As in Section ~\ref{section_inf} we make the mild technical assumption that all sender utilities are real analytic in some neighborhood of $\delta_l$ for all states $l$ (Assumption \ref{technical_assp}).

\begin{lemma2} \label{prop2_private_info}
Suppose Assumption \ref{technical_assp} holds. Consider any pair of states $l$ and $k$. $\{l,k\}$ is not pooled in every equilibrium if for some $i$ $u_i$ satisfies Condition \ref{condition1} on $\Delta(\{l,k\})$.
\end{lemma2}
\begin{proof}

WLOG let $l=1$, $k=2$. For notational convenience let $v = v^{1,2}$. Suppose $u_i$ satisfies Condition \ref{condition1} on $\Delta(\{1,2\})$ for some $i$. WLOG we consider the case $\nabla_v u_i(\delta_2) \neq u_i(\delta_2) - u_i(\delta_1)=0$. Then by zero-sumness (derivatives must also be zero-sum where they exist for all senders) there exists a sender $j$ with $\nabla_v u_j(\delta_2)=c<0$. This implies there exists $r \in (0,1)$ such that for all $\gamma \in \Delta(\{1,2\})$ with $\gamma_2>r$, $u_j(\gamma)>0$. In other words, $j$ has a region of advantage along $\Delta(\{1,2\})$ close to $\delta_2$.

Suppose for contradiction there is an equilibrium $(\sigma_1,...,\sigma_M)$, $(\mu^1,...,\mu^M)$ that pools $\{l,k\}$. For each $j' \neq j$, let $\Lambda_{j'} = \alpha^{j'}(\Gamma_{j'},p_{j'},\mu^{j'})$. By Lemma \ref{lemma_private_inf_pooling}, for all $j' \neq j$, $\Lambda_{j'}$ pools $\{1,2\}$. Note $\Lambda_{j'}$ is also a random variable (where randomness is over $p_{j'}$ and the realization of $\Gamma_{j'}$) with finite support (due to finite support of $b_{j'}$ and $\Gamma_{j'}$). If all senders $j' \neq j$ follow the equilibrium play, $\Lambda_{j'}$ is also Bayes-plausible (with mean $\pi$); this is because $\Gamma_{j'}$ Bayes-plausible, any learning the receiver does about $b_{j'}$ must follow Bayes rule on path, and the distribution of $b_{j'}$ has mean $\pi$. Let $\Lambda_{-j}$ denote the interim belief induced by viewing realizations of all $\{\Lambda_{j'}\}_{j' \neq j}$; note this experiment also pools $\{1,2\}$; this is also a finite signal Bayes-plausible experiment.


We will now find a profitable deviation for sender $j$. This deviation will take the form of an experiment $p_j'$, which we will construct, that generates strictly positive expected utility no matter what $j$'s type is. 

We can rewrite $\beta_l$ for $l \in \Omega$ (from equation \ref{private_beta}) conditional on this deviation $p_j'$ as:

\begin{equation} \label{beta_lambda}
\beta_l(\Lambda_{-j},\alpha^j(\Gamma_j,p_j',\mu^j)) = \frac{\Lambda_{-j,l} \alpha_l^j(\Gamma_j,p_j',\mu^j) / \pi_l}{\sum_{k=1}^N \Lambda_{-j,k} \alpha_k^j(\Gamma_j,p_j',\mu^j)/ \pi_k}
\end{equation}

It is also useful to write down the probability distribution of $\Lambda_{-j}$ conditional on $\alpha^j(\Gamma_j,p_j',\mu^j)$: 

\begin{equation} \label{lambda_alpha}
Pr(\Lambda_{-j}=y|\alpha^j(\Gamma_j,p_j',\mu^j)) = \sum_{l=1}^N \frac{\alpha^j_l(\Gamma_j,p_j',\mu^j) y_l}{\pi_l}
\end{equation}

Let $Y = \{y \in supp[\Lambda_{-j}]: y_1,y_2>0\}$; note this set is nonempty as $\Lambda_{-j}$ pools $\{1,2\}$. Let $Y_0 = \{y \in supp[\Lambda_{-j}]: y_1,y_2=0\}$, $Y_1 = \{y \in supp[\Lambda_{-j}]: y_1>0,y_2=0\}$, and $Y_2 = \{y \in supp[\Lambda_{-j}]: y_1=0,y_2>0\}$. These sets partition the support of $\Lambda_{-j}$.

Consider $j$ generating interim belief $\Gamma_j=x \in \Delta(\{1,2\}) \setminus \{\delta_1,\delta_2\}$. As in Proposition 2's proof, we will try and find such an $x$ conditional on which $j$ gets a strictly positive expected payoff. We will then construct $p_j'$ which assigns positive probability to $x$. Note that $Pr(\alpha^j(\Gamma_j,p_j',\mu^j) \in \Delta(\{l,k\}) \setminus \{\delta_1,\delta_2\}|\Gamma_j =x ) =1$. This can be seen from the definition of $\alpha^j$ (equation \ref{alpha}) and is a consequence of every private belief $b \in supp[b_j]$ having $b_n>0 $ for all $n \in \Omega$ which implies $\tau^j_l(\mu^j,p_j')>0$ for all $l \in \Omega$. This implies that $Pr(\Lambda_{-j} \in Y_0| \Gamma_j = x \in \Delta(\{1,2\}) \setminus \{\delta_1,\delta_2\}) = 0$ (equation \ref{lambda_alpha}). Also note that conditional on $\Gamma_j = x \in \Delta(\{1,2\}) \setminus \{\delta_1,\delta_2\}$, $\beta = \delta_1$ when $\Lambda_{-j} \in Y_1$ and $\beta = \delta_2$ when $\Lambda_{-j} \in Y_2$ (see equation \ref{beta_lambda}); these posteriors both yield utility $0$. 

Finally, note that conditional  $\Gamma_j = x \in \Delta(\{1,2\}) \setminus \{\delta_1,\delta_2\}$, we have $Pr(\Lambda_{-j} \in Y)>0$ (by equation \ref{lambda_alpha} and $\alpha^j \in \Delta(\{l,k\}) \setminus \{\delta_1,\delta_2\}$). Note that for each $y \in Y$ and $\alpha^j(\Gamma_j,p_j',\mu^j) \in \Delta(\{1,2\})\setminus \{\delta_1,\delta_2\}$, $\beta_k(y,\alpha^j(\Gamma_j,p_j',\mu^j))$ is continuous in $\alpha^j(\Gamma_j,p_j',\mu^j)$, $\beta_k(y,\alpha^j(\Gamma_j,p_j',\mu^j)) \rightarrow 1$ as $\alpha^j(\Gamma_j,p_j',\mu^j) \rightarrow 1$, and $\beta_k(y,\alpha^j(\Gamma_j,p_j',\mu^j)) \rightarrow 0$ as $\alpha^j(\Gamma_j,p_j',\mu^j) \rightarrow 0$. As $Y$ is finite, the function $\min_{y \in Y} \beta_k(y,\alpha^j(\Gamma_j,p_j',\mu^j))$ is also continuous in its second argument for $\alpha^j(\Gamma_j,p_j',\mu^j) \in \Delta(\{1,2\})\setminus \{\delta_1,\delta_2\}$ and goes to $0$ or $1$ as $\alpha^j(\Gamma_j,p_j',\mu^j)$ goes to $0$ or $1$ respectively. By the intermediate value theorem, there exists a $\alpha_r \in (0,1)$ such that $\min_{y \in Y} \beta_k(y,\alpha_r) = r$.  When $1>\alpha^j(\Gamma_j,p_j',\mu^j) >\alpha_r$, we have $\beta_k(y,\alpha^j(\Gamma_j,p_j',\mu^j)) \in (r,1)$ for all $y \in Y$.

Note that for $\Gamma_j=x \in \Delta(\{1,2\}) \setminus \{\delta_1,\delta_2\}$, we have $\alpha^j_2(x,p_j',\mu^j) = \frac{\tau^j_2(\mu^j,p_j') x_2/\pi_2}{\tau^j_1(\mu^j,p_j') x_1/\pi_1 + \tau^j_2(\mu^j,p_j') x_2/\pi_2}$. We can rewrite this as: $\alpha^j_2(x,p_j',\mu^j) = \frac{\frac{\tau^j_2(\mu^j,p_j')}{\tau^j_1(\mu^j,p_j')} x_2/\pi_2}{x_1/\pi_1 + \frac{\tau^j_2(\mu^j,p_j')}{\tau^j_1(\mu^j,p_j')} x_2/\pi_2}$. Note that as private beliefs have finite support and cannot rule out any state, for all beliefs the receiver may hold about $j$'s type, $\mu \in \Delta(supp[b_j])$, the corresponding belief $\tau$ this induces on $\Omega$ ($\tau_l = \sum_{b \in supp[b_j]}) b_l \mu[b]$) must have $\frac{\tau_2}{\tau_1} \geq d>0$ for some $d \in (0,1)$. Hence:

\begin{equation} \label{bound_alpha}
\alpha^j_2(x,p_j',\mu^j) = \frac{\frac{\tau^j_2(\mu^j,p_j')}{\tau^j_1(\mu^j,p_j')} x_2/\pi_2}{x_1/\pi_1 + \frac{\tau^j_2(\mu^j,p_j')}{\tau^j_1(\mu^j,p_j')} x_2/\pi_2} \geq  \frac{d x_2/\pi_2}{x_1/\pi_1 + d x_2/\pi_2} 
\end{equation}

$\alpha^j_2(x,p_j',\mu^j)$ is continuous in $x$ on $\Delta(\{1,2\}) \setminus \{\delta_1,\delta_2\}$ and will also fall in $\Delta(\{1,2\}) \setminus \{\delta_1,\delta_2\}$. By equation \ref{bound_alpha}, as $x_2 \rightarrow 1$, $\alpha^j_2(x,p_j',\mu^j) \rightarrow 1$ \emph{regardless} of what beliefs the receiver holds. Hence there exist $x^* \in \Delta(\{1,2\}) \setminus \{\delta_1,\delta_2\}$ such that $1>\alpha^j_2(x^*,p_j',\mu^j)>\alpha_r$. Hence we have $\beta_k(y,\alpha^j(x^*,p_j',\mu^j)) \in (r,1)$ for all $y \in Y$. Conditional on $x^*$, $j$ gets expected utility:

\begin{align*}
& \sum_{\substack{y \in Y_1}} \underbrace{u_j(\delta_1)}_{=0} Pr(\Lambda_{-j}=y|\alpha^j(x^*,p_j',\mu^j)) + \sum_{\substack{y \in Y_2}} \underbrace{u_j(\delta_2)}_{=0} Pr(\Lambda_{-j}=y|\alpha^j(x^*,p_j',\mu^j))\\
& + \sum_{y \in Y} \underbrace{u_j(\beta(y,\alpha^j(x^*,p_j',\mu^j)))}_{>0} \underbrace{ Pr(\Lambda_{-j}=y|\alpha^j(x^*,p_j',\mu^j))}_{>0} > 0\\
\end{align*}

The proof of Lemma 1 demonstrates how any type of sender $j$ can construct a stratgy $p_j'$ which assigns positive probability only to $x^*$ and $\{\delta_1,...,\delta_N\}$. $p_j'$ yields $j$ strictly positive expected utility conditional on $x^*$ being realized and $0$ utility otherwise. Hence Lemma \ref{lemma1_private_info} is violated. Contradiction. Hence no equilibrium can pool $\{1,2\}$.

\end{proof}

Lemma \ref{prop2_private_info} implies that Condition \ref{condition1} being satisfied by some $u_i$ on each edge of the simplex is sufficient for full revelation in all equilibria. It is worth noting again that is sufficient condition is satisfied for all but a knife-edge case of sender preferences.

\textbf{Proposition \ref{prop_private}.} \emph{Suppose Assumption \ref{technical_assp} holds. The state is fully revealed in every equilibrium if for all pairs of states $l$ and $k$, there is some $u_i$ that satisfies Condition \ref{condition1} on $\Delta(\{l,k\})$.}

\begin{proof}
Argument is identical to Proposition \ref{prop_inf}'s proof.
\end{proof}

\subsection{Sequential moving senders}
Consider a sequential version of our baseline model. Senders $1,...,M$ move in order, observing all previous experiment choices (but not realizations); we are interested in pure strategy subgame perfect Nash Equilibria (henceworth just SPNE) of this game. Note that for a simultaneous game, there are multiple corresponding sequential games, one for each ordering of senders.

A few facts easily carry over from the simultaneous case. First, all senders must get utility $0$ in equilibrium (as the game is zero-sum and anyone can fully reveal the state). Second, full revelation can be supported as an SPNE outcome in a game with senders moving in any order. We can construct such an equilibrium with all senders playing $\Gamma^{FR}$ on path and playing any sequentially rational strategies off path. No sender after the first has an incentive to deviate if those upstream from them have not (as the receiver will the learn the state from upstream senders). The first sender cannot strictly gain from deviating, as a strict gain would imply a strict loss for a downstream sender; sequential rationality rules this out as the downstream sender can fully reveal the state to avoid a loss. As for the simultaneous game, the interesting question is when all equilibria (here SPNE) are fully revealing. The following results and discussion clarify the relationship between the our simultaneous model and a sequential version.

\textbf{Proposition \ref{prop_seq}.} \emph{If for $u_1,...,u_M$ there is full revelation in every SPNE of the sequential game with the senders moving in some order, then there is full revelation in every equilibrium of the simultaneous game.}

\begin{proof}
We prove the following statement, from which the result follows: if there exists a non-fully revealing equilibrium in the simultaneous game, then, for \emph{any} order of senders, there exists a non-fully revealing SPNE in the sequential game.

Choose any ordering of senders $1,...,M$. Consider any non-fully revealing equilibrium of the simultaneous game, $(\Gamma_1,...,\Gamma_M)$ and let $\Gamma$ be the experiment induced by observing the realizations of $\Gamma_1,...,\Gamma_M$. In the sequential game, consider the following strategy profile: (1) sender $1$ plays $\Gamma$. (2) each sender $i=2,...,M$ plays $\Gamma^U$ if all previous senders haven't deviated and play some seqentially rational strategies otherwise. We will show this is an SPNE. By Lemma 1 all senders get utility $0$ from following perscribed play as $\Gamma$ is the information revealed in an equilibrium of the simultaneous game and no additional information is revealed. First note sender $1$ has no strict incentive to deviate as any profitable deviation would give some downstream sender strictly negative utility. This is not possible along any path of play in an SPNE as this downstream sender can always fully reveal the state.

If sender $1$ plays $\Gamma$, senders $2,...,M$ have no incentive to deviate for the following reason. Consider any deviation $\Gamma_j' \neq \Gamma^U$ for sender $j$, $2 \leq j \leq M$. This deviation leads to a path of play producing information from $\Gamma$ as well as additional conditionally independent experiments. Suppose, for contradiction, this deviation yields $j$ strictly positive expected utility. But then $j$ has a profitable deviation from the simulatenous game equilibrium $(\Gamma_1,...,\Gamma_M)$; if $j$ unilaterally plays these additional conditionally independent experiments in addition to $\Gamma_j$, she gets strictly positive utility. 
\end{proof}

Proposition \ref{prop_seq} implies that the set of (zero-sum) utility functions under which there is full revelation in all equilibria in the simultaneous game contains the set under which for some order of senders there is full revelation in all SPNE of the sequential game. The converse is not true. It is possible for there to be full revelation in every equilibrium of the simultaneous game but, for every order of senders, non-fully revealing SPNE in sequential game. The following example demonstrates this. 

\begin{example} \label{seq_ex}
There are two senders $1,2$ and three possible states $1,2,3$. Assume that $\pi_1 < \pi_2 < \pi_3$ (this is not necessary, but eases exposition; the assumption rules out any two states having equal prior probabilities but is otherwise without loss). Suppose $u_1((1/2,1/2,0))=u_1((1/2,0,1/2))=1$ and $u_2((1/2,1/2,0))=u_2((1/2,0,1/2))=-1$. Also, $u_2((0,1/2,1/2))=1$ and $u_1((0,1/2,1/2))=-1$. At all other $\gamma \in \Delta(\Omega)$, $u_1(\gamma)=u_2(\gamma)=0$. Sender $1$ has an advantage at single points along edges $\Delta(\{1,2\})$ and $\Delta(\{1,3\})$ and sender $2$ has a single advantage on edge $\Delta(\{2,3\})$; on the rest of the simplex, neither sender has an advantage. Figure \ref{example_seq} summarizes this. By Theorem 1, the state is fully revealed in all equilibria of the simultaneous game. However we will show that regardless of the order senders move in, there is always a non-fully revealing SPNE.

\begin{center} \
\includegraphics[scale=0.7]{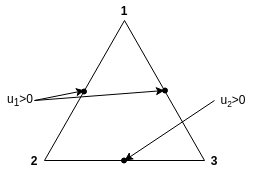}
\captionof{figure}{} \label{example_seq}
\end{center} 
 
First suppose sender $1$ plays first, then sender $2$. Consider the sender $1$ playing $\Gamma_1$ s.t. $Pr(\Gamma_1=\delta_3)= \pi_3$ and $Pr(\Gamma_1=(\frac{\pi_1}{\pi_1 + \pi_2},\frac{\pi_2}{\pi_1 + \pi_2},0)) = 1 - \pi_3$ (this distribution satisfies Bayes-plausibility). Suppose sender $2$ plays $\Gamma_2 = \Gamma^U$ whenever sender $1$ plays this and following a deviation plays any sequentially rational $\Gamma_2$. Note the posterior when following perscribed play will either be $\delta_3$ or $(\frac{\pi_1}{\pi_1 + \pi_2},\frac{\pi_2}{\pi_1 + \pi_2},0)$. The latter is not equal to $(1/2,1/2,0)$ by our assumption on the prior, and hence both senders get utility $0$ at all posteriors. Sender $1$ hence has no incentive to deviate as no experiment can yield strictly positive utility (sender $2$ can always fully reveal the state and is playing sequentially rationally). Sender $2$ has no incentive to deviate as conditional on $\omega = 3$ the receiver learns the state and conditional on $\omega \in \{1,2\}$ the posterior will lie on $\Delta(\{1,2\})$, on which sender $2$ has no points of advantage, w.p. $1$ (by Claim \ref{claim_projection}). This is a non-fully revealing equilibrium.

If sender $2$ plays first, we can construct an analogous non-fully revealing equilibrium with $\Gamma_2$ s.t. $Pr(\Gamma_2=\delta_1)= \pi_1$ and $Pr(\Gamma_2=(0,\frac{\pi_2}{\pi_2 + \pi_3},\frac{\pi_3}{\pi_2 + \pi_3})) = 1 - \pi_1$. Sender $1$ plays $\Gamma_1=\Gamma^U$ on path and any sequentially rational $\Gamma_1$ otherwise.
\end{example}

In the simultaneous version of Example \ref{seq_ex}, the state is fully revealed in every equilibrium for the following reasons. States $2$ and $3$ must not be pooled in equilibrium because if not sender $2$, who has an advantage along $\Delta(\{2,3\})$, could find an experiment yielding a strictly positive payoff (violating Lemma 1). More precisely, $\Gamma_1$ cannot pool states $\{2,3\}$, because if not sender $2$ can gain a strictly positive payoff. Similarly, $\Gamma_2$ cannot pool $\{1,2\}$ or $\{1,3\}$ or sender $1$ take advantage. Each sender is `responsible' for not pooling some states in equilibrium because their opponent has an advantage.

More generally, when the state must be fully revealed in every equilibrium of a simultaneous game, for every subset of states $\Omega'$ there is some sender $j$ who could take advantage of $\Gamma_{-j}$ pooling $\Omega'$. This is shown by Claim \ref{strong_claim_prop2} (in the proof of Prosition 2); sender $j$ must have an advantage somewhere on $\Delta(\Omega)$. When sender $j$ moves last in the sequential game, then by the same argument, all upstream senders must (collectively) not pool $\Omega'$ in any SPNE. For example, when sender $2$ moved second in the example, sender $1$ did not pool $\{2,3\}$. However, as the last moving sender may only be able to take advantage of some subsets of states being pooled, we need not get full revelation. If there is a sender $j$ who can take advantage of $\Omega'$ being pooled for any $\Omega' \in 2^{\Omega}$ s.t. $|\Omega'|>1$, then when this sender moves last, every SPNE fully reveals the state. As in Example \ref{seq_ex}, it is when there is no such sender exists that there are non-fully revealing equilibria for all orders of senders.

This logic implies the following result: if any set of states is not pooled in every equilibrium of the simultaneous game, there is an ordering of senders such that those states are not pooled in every SPNE of the sequential game.

\begin{prop2} If for utility functions $u_1,...,u_M$ a set of states $\Omega' \subseteq \Omega$ is not pooled in every equilibrium of the simultaneous game, then, for some ordering of senders, $\Omega'$ is not pooled in every SPNE of the sequential game.
\end{prop2}

\begin{proof}
Suppose $\Omega' \subseteq \Omega'$ is not pooled in every equilibrium of the simultaneous game. Then by Proposition 2, $u_i(\gamma)>0$ for some sender $i$ and some $\gamma \in \Delta(\Omega')$. By Claim \ref{strong_claim_prop2}, there exists a sender $j$ such that for every $\Gamma_{-j}$ that pools $\Omega'$, there exists a $\Gamma_j$ such that $U_j(\Gamma_j,\Gamma_{-j})>0$. Consider any ordering of senders with $j$ moving last. Then as all senders must get utility $0$, in any SPNE senders upstream from $j$ (collectively) do not pool $\Omega'$ (or else $j$'s best response yields strictly positive utility). Hence, by Lemma \ref{not_pool}, all SPNE do not pool $\Omega'$.
\end{proof}

\section{Supplementary Appendix C: Extensions to some nonzero-sum settings}

In this section we consider a few applications that fall outside the scope of our model's assumptions. We use these applications to highlight that our analysis applies more broadly to many real-world settings.

\subsection{Persuading a judge/jury}

Consider a criminal court case; there is a state $\omega \in \{I,G\}$ (the defendant is innocent or guilty). Senders $P,D$ are prosecutor/defense attorneys and will choose what evidence to search for (experiments) in order to persuade the receiver, a judge or jury. After hearing the arguments of the prosecution and defense (i.e. observing experiments and their realizations), the judge/jury will update their beliefs on $\omega$ and make a decision. If, given the evidence presented $\Pr(\omega=G) \geq 0.9$, i.e. the posterior is strongly in favor of the defendant's guilt, the judge/jury will convict the defendant. If after updating $Pr(\omega=G) \leq 0.3$, the judge/jury will acquit the defendant. Finally, if $Pr(\omega=G) \in (0.3,0.9)$, there is a mistrial and no verdict (e.g. the judge does not have enough evidence to make a ruling or the jury is hung).\footnote{If this is an adversarial legal system, we can think of $0.9$ as the threshold of `reasonable doubt'. If, more specifically, this is a common law legal system and the receiver is a jury, we can think of $0.9$ as the threshold above which jurors are all convinced of the defendant's guilt beyond reasonable doubt. Below $0.3$, jurors are unanimously convinced there is not enough evidence to convict and for beliefs in the interval $(0.3,0.9)$ the jury has a divided opinion and is hence hung.}

The prosecution gets a payoff of $1$ from conviction, $-1$ from acquittal, and $r_p \in (-1,1)$ from a mistrial. The defense gets payoff $1$ from acquittal, $-1$ from conviction, and $r_d \in (-1,1)$ from a mistrial. Note that for posterior beliefs $\beta \in [0,0.3] \cup [0.9,1]$, the game is zero-sum (for exposition, we do not normalize payoffs at $\beta=0$ and $\beta=1$ to $0$). However, we make the assumption that $r_p+r_d<0$ ; when there is a mistrial, there is a loss in sender `surplus' (or, the sum of sender payoffs). In a common law system, a mistrial can result in a retrial; this loss in surplus could be due to both attorneys bearing a cost of preparing for a retrial or from attorneys discounting their payoffs as it will take another trial to reach a verdict. Figure \ref{fig_hung_jury} gives an example of such payoffs.

\begin{samepage}
\begin{center}
\includegraphics[scale=0.7]{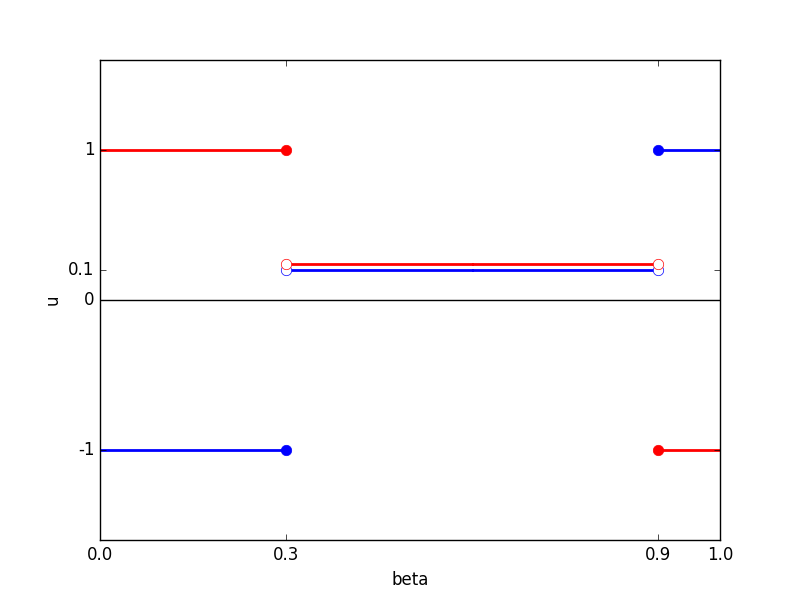}
\captionof{figure}{Example of $u_p$ (blue) and $u_d$ (red) with $r_p=r_d=0.1$}
\label{fig_hung_jury}
\end{center}
\end{samepage}

This game is not zero-sum \textemdash attorney/sender surplus is smaller at intermediate beliefs, $(0.3,0.9)$, then at extreme ones. Hence, although utilities are nonlinear, we may not expect our full revelation results to hold. Despite this:

\textbf{Result.} The state is fully revealed in every equilibrium.

To see why this holds, first note that the jury must be hung w.p. $0$ in all equilibria. For any strategy profile $(\Gamma_1,\Gamma_2)$ at which the jury is hung w.p. $>0$, $U_1(\Gamma_1,\Gamma_2)+U_2(\Gamma_1,\Gamma_2)<0$.  But full revelation guarantees that the sum of sender ex-ante expected utilities equals $0$; hence some sender must strictly benefit from deviating from this strategy profile to full revelation.

As the jury is hung w.p. $0$ in all equilibria, all equilibria place probability $1$ on extreme posteriors, $[0,0.3]\cup[0.9,1]$, for which $u_1+u_2=0$. Our argument in proving Proposition 1 used beliefs at extremes of the interval. As the game is zero-sum for such beliefs, the prosecution has an advantage on $[0.9,1]$, and the defense has an advantage on $[0,0.3]$, the same argument implies all equilibria are fully revealing.

\textbf{Generalizing the example.} This same intuition generalizes beyond this example. Consider a nonzero-sum game otherwise identical to our baseline model. At any posterior $\beta$, we define $W(\beta) = \sum_{i=1}^M u_i(\beta)$ as the sender surplus at this posterior. As in the example above, whenever $W(\beta)$ is maximized at and in the neighborhood of fully revealing posteriors and close to each fully revealing posterior some sender has an advantage, then we have full revelation in every equilibrium. This latter condition is guaranteed by Condition 1, and so we have the following sufficient condition for full revelation in all equilibria of nonzero-sum games:

\begin{prop}
Suppose that there exists neighborhoods of $\delta_1,...,\delta_N$ such that $W(\beta)$ is maximized in each of these neighborhoods. Suppose for every $l,k \in \Omega$ some $u_i$ satisfies Condition 1 on $\Delta(\{l,k\})$. Then the state is fully revealed in every equilibrium.
\end{prop}

The proof follows from the logic above and the proof of Proposition 2 applied to pairs of states (see Supplementary Appendix A). Notice that if $W(\beta)$ is \emph{uniquely} maximized at $\delta_1,...,\delta_N$ then the argument above implies that for every non-fully revealing $(\Gamma_1,...,\Gamma_M)$, some sender can profitably deviate to $\Gamma^{FR}$. Hence, if this is the case then we also have that all equilibria are fully revealing.\footnote{This logic is similar to that of \cite{gentzkow2016competition} Proposition 1.}

While it may not be surprising that we can guarantee full revelation in all equilibria of nonzero-sum games in cases where full revelation maximizes sender surplus, it should also not be surprising that our results collapse when sender surplus is not maximized by full revelation. If sender surplus is only maximized at interior posteriors, then it is easy to construct payoffs for which senders collude to reveal little information to the receiver in equilibrium.

\subsection{Competition in advertising}

Consider a game in which two firms (senders $1,2$) selling cars compete for the demand of a single consumer (the receiver). The state-space here is two dimensional. First, the consumer has a willingness of pay $w \in \{l=\frac{1}{2},h=1\}$. Second, one firm $b \in \{1,2\}$ has a higher quality car. Assume $w$ and $b$ are drawn independently and the state is the pair $\omega=(w,b)$.

At the start of the game, firms simultaneously choose advertising strategies (experiments) to reveal information about $\omega$. After viewing these ads and updating her beliefs on $\omega$, the consumer purchases exactly one car and pays her expected willingness to pay at her posterior belief; that is, she pays $\frac{1}{2} Pr(w=l) + 1 Pr(w=h)$. The consumer purchases the car $1$ if it is more favorable at her posterior, if $Pr(b=1) \geq \frac{1}{2}$, and car $2$ otherwise (the tie-breaking rule does not matter). Firm $i$ gets a payoff equal to the revenue it receives: $\frac{1}{2} Pr(w=l) + 1 Pr(w=h)$ if its car is bought and $0$ otherwise.

\textbf{Discussion of setup.} We can think of advertisements here as revealing information about the value of the car to the consumer (e.g. if an advertisement shows how a car is valuable for commuters, then the consumer may increase her estimate of the car's value if she is a commuter) and about the relative qualities of the two cars (e.g. firms can include characteristics, say mileage, which can be compared across cars). We assume that initially, information about the consumer's willingness to pay is symmetric. It is reasonable to think the consumer does not initially know her own willingness to pay for a product and that advertisements (here experiments) inform her of the value of the product.

As in the previous section, this game is not zero-sum (or constant-sum): firm/sender surplus is increasing in the consumer's expexcted willingness to pay. However, note that for any belief the consumer holds on $w$, firms have the same preferences over her belief on $b$ (up to scaling); furthermore, for any belief on $w$, firm preferences over beliefs on $b$ are constant-sum. This implies:
    
\textbf{Result.} In every equilibrium the consumer learns (fully) which car is better.
    
Senders in this game agree on one dimension of the state-space, $w$. However, their disagreement on the other dimension of the state-space, $b$, is constant-sum and is unaffected by beliefs on $w$. As for any fixed belief on $w$ the arguments of Proposition 1 will imply that $b$ is fully revealed, it must be that $b$ is fully revealed in all equilibria.

\textbf{Generalizing the example.} The same result applies whenever $\omega$ is composed of two independent dimensions with zero-sum disagreement on one dimension that is unaffected by the other. Formally, suppose $\omega_1$ and $\omega_2$ are drawn independently from finite sets $\Omega_1$ and $\Omega_2$. Let $\omega = (\omega_1,\omega_2) \in \Omega = \Omega_1 \times \Omega_2$. Any belief $\beta \in \Delta(\Omega)$ on $\omega$ can be written as $\beta=(\gamma,\eta)$, where $\gamma \in \Delta(\Omega_1)$ is a belief on $\omega_1$, and $\eta \in \Delta(\Omega_2)$ is a belief on $\omega_2$. We write sender utility functions over the receiver's posterior as $u_i(\gamma,\eta)$. We say senders have \emph{zero-sum preferences on $\Omega_1$ that are uniform in $\Omega_2$} if: 
\begin{align*}
& \text{(1) for all $\eta \in \Delta(\Omega_2)$ there exists $c \in \mathbb{R}$ such that: $\sum_{i=1}^M u_i(\gamma,\eta)=c$ for all $\gamma \in \Delta(\Omega_1)$} \\
& \text{(2) for all $\eta,\eta' \in \Delta(\Omega_2)$ there exists $d,e \in \mathbb{R}$ such that $u_i(\cdot,\eta) = d u_i(\cdot,\eta')+e$ for all senders $i$}
\end{align*}

Part (1) of this definition says that for each belief $\eta$ on $\omega_2$, senders have constant-sum preferences over the receiver's beliefs on $\omega_1$. Part (2) requires that, up to linear transformations, sender preferences over beliefs on $\omega_1$ are not affected by the receiver's belief on $\omega_2$; this will imply at sender preferences over experiments on $\omega_1$ will not be affected by the receiver's belief on $\omega_2$. Under these conditions, if for any fixed belief on $\omega_2$ sender preferences satisfy Theorem 1's conditions for fully revealing $\omega_1$, $\omega_1$ is fully revealed in all equilibria.

\begin{prop} \label{prop_dimension}
Suppose senders have zero-sum preferences on $\Omega_1$ that are uniform in $\Omega_2$. Fix any $\eta \in \Delta(\Omega_2)$. $\omega_1$ is fully revealed in every equilibrium if and only if for each $l,k \in \Omega_1$ some $u_i(\gamma,\eta)$ is nonlinear in $\gamma$ for $\gamma \in \Delta(\{l,k\})$.
\end{prop}

The proof follows immediately from the logic above. Proposition \ref{prop_dimension} is not surprising but is significant. Our analysis can be useful in (special) nonzero-sum settings for which senders have zero-sum disagreement over some dimensions but not others.
\end{document}